\documentclass[11pt]{article}

\newif\ifcomments
\commentstrue

\usepackage{fullpage}
\usepackage[utf8]{inputenc}

\usepackage{graphicx}
\usepackage{amsmath,amsthm,amssymb}
\usepackage{algorithm}
\usepackage{algpseudocode}
\usepackage{multirow}
\usepackage{makecell}

\usepackage{thm-restate,color,xspace}
\usepackage{comment}
\usepackage{thmtools}
\usepackage{xcolor}
\usepackage{nameref}
\usepackage{array}
\usepackage{bbm}

\definecolor{ForestGreen}{rgb}{0.1333,0.5451,0.1333}
\definecolor{DarkRed}{rgb}{0.65,0,0}
\definecolor{Red}{rgb}{1,0,0}
\definecolor{DarkBlue}{RGB}{0, 70, 140}
\usepackage[linktocpage=true,
pagebackref=true,colorlinks,
linkcolor=DarkRed,citecolor=ForestGreen,
bookmarks,bookmarksopen,bookmarksnumbered]
{hyperref}
\usepackage{cleveref}
\usepackage{mathrsfs}
\usepackage{enumitem}

\usepackage{mathtools}
\usepackage[most]{tcolorbox}

\declaretheorem[numberwithin=section]{theorem}
\declaretheorem[numberlike=theorem]{lemma}

\declaretheorem[numberlike=theorem]{claim}

\declaretheorem[numberlike=theorem]{assumption}
\declaretheorem[numberlike=theorem]{observation}

\declaretheorem[numberlike=theorem,style=definition]{definition}
\declaretheorem[numberlike=theorem,style=definition]{remark}

\newcounter{Qtypectr}
\newcounter{Qtildetypectr}
\newcounter{Oracletypectr}

\global\long\def\ed{\mathrm{ed}}

\global\long\def\tree{\mathrm{tree}}
\global\long\def\nc{\mathrm{nc}}
\global\long\def\cov{\mathrm{cov}}
\global\long\def\diam{\mathrm{diam}}
\global\long\def\sc{\mathrm{sc}}
\global\long\def\tmn{\mathrm{tmn}}
\global\long\def\nd{\mathrm{nd}}
\global\long\def\lk{\mathrm{lk}}

\global\long\def\qry{Q}
\global\long\def\wtilde{\widetilde}
\global\long\def\aff{\mathrm{aff}}
\global\long\def\dist{\mathrm{dist}}
\global\long\def\x{x}
\global\long\def\y{y}

\global\long\def\heavy{\mathrm{heavy}}
\global\long\def\px{\mathrm{px}}
\global\long\def\leng{\mathrm{leng}}
\global\long\def\level{\mathsf{level}}
\global\long\def\scale{\mathsf{scale}}
\global\long\def\ans{\mathsf{ans}}
\global\long\def\ProxyList{\mathsf{ProxyVertex}}
\global\long\def\ProxyVertex{\mathsf{ProxyVertex}}
\global\long\def\spar{\mathrm{spar}}
\global\long\def\satellite{\mathsf{Stl}}
\global\long\def\cluster{\mathsf{cluster}}
\global\long\def\pivot{\mathsf{pivot}}
\global\long\def\Ball{\mathrm{Ball}}

\global\long\def\alphahat{\hat{\alpha}}

\global\long\def\node{\mathsf{node}}

\global\long\def\core{\mathsf{core}}

\global\long\def\sc{\mathrm{sc}}

\global\long\def\Qtilde{\wtilde{Q}}
\global\long\def\Ptilde{\wtilde{P}}
\global\long\def\affS{\mathrm{affS}}
\global\long\def\nb{\mathrm{nb}}
\global\long\def\Nbr{\mathsf{Nbr}}
\global\long\def\Upward{\mathsf{Upward}}
\global\long\def\AdjCluster{\mathsf{AdjCluster}}
\global\long\def\Euler{\mathsf{Euler}}
\global\long\def\start{\mathsf{start}}
\global\long\def\eend{\mathsf{end}}
\global\long\def\oc{\mathrm{oc}}
\global\long\def\Ical{\mathcal{I}}
\global\long\def\larr{\leftarrow}
\global\long\def\rarr{\rightarrow}

\global\long\def\CoreNodes{\mathsf{CoreNodes}}
\global\long\def\Shortcuts{\mathsf{Shortcuts}}
\global\long\def\ppt{\mathrm{pt}}
\global\long\def\point{\mathsf{point}}
\global\long\def\vpt{v_{\mathrm{pt}}}
\global\long\def\upt{u_{\mathrm{pt}}}
\global\long\def\region{\mathsf{region}}
\global\long\def\down{\mathrm{down}}
\global\long\def\up{\mathrm{up}}

\global\long\def\desc{\prec}
\global\long\def\desceq{\preceq}
\global\long\def\anc{\succ}
\global\long\def\anceq{\succeq}

\global\long\def\fail{\mathrm{fail}}
\global\long\def\sub{\mathrm{sub}}
\global\long\def\vtx{\mathrm{vtx}}

\global\long\def\poly{\mathrm{poly}}
\global\long\def\supp{\mathrm{supp}}

\global\long\def\spars{\mathrm{spars}}
\global\long\def\sep{\mathrm{sep}}

\global\long\def\Otilde{\tilde{O}}
\global\long\def\glo{\mathrm{glo}}
\global\long\def\cost{\mathrm{cost}}

\global\long\def\portal{\mathsf{portal}}

\title{Vertex-Failure Distance Oracles and Labeling Schemes:\\ Compact and Constant-Approximate}
\author{Yaowei Long\thanks{\texttt{yaoweil@umich.edu}}\\University of Michigan}
\date{}

\begin{document}

\maketitle
\pagenumbering{gobble}
\begin{abstract}
We present new algorithms for the \emph{vertex-failure distance oracles} and \emph{labeling schemes} problems in undirected weighted graphs.

A vertex-failure distance oracle is a data structure that, given two vertices $x$ and $y$ and a failed vertex set $F$ of size at most $f$, returns an approximation to the distance between $x$ and $y$ in $G \setminus F$. In the labeling-scheme setting, the data structure needs to be stored distributively as \emph{labels} on the vertices, and each query $(x,y,F)$ must be answered by accessing only the labels of the vertices in $F \cup \{x,y\}$.

For any $f\geq 1$ and $k \ge 1$, we obtain a vertex-failure distance oracle with $O(k^{6})$ approximation, space $\tilde{O}(f^{2}n^{1+1/k})$, query time $\tilde{O}(f^{5}n^{1/k})$, and polynomial preprocessing time. In particular, this is the first time-efficient oracle for multiple vertex failures with space close to linear, as well as the first constant-approximation oracle with polynomial space when tolerating $\Omega(\log n)$ vertex failures. The previous results, due to Duan, Gu, and Ren \cite{duan2021approximate}, gave two alternatives: for any constant $c \ge 1$ and $\epsilon>0$, one oracle has $\poly(\log n,f)$ approximation, space $n^{2+1/c}\poly(\log n,f)$, and query time $\poly(\log n,f^{c})$, while the other has $(1+\epsilon)$ approximation, space $n^{2+1/c}(\log n/\epsilon)^{O(f)}$, and query time $\poly(\log n,f^{c},1/\epsilon)$. 

We also obtain a vertex-failure distance labeling scheme with $O(k^{6})$ approximation and label size $f^{3}n^{1/k}\log^{O(k)} n$. This is the first nontrivial distance labeling scheme for vertex failures.

Our techniques build on recent tools related to \emph{length-constrained vertex expanders} \cite{haeupler2025length} and also introduce a new expander-based shortcut sparsification. The latter also leads to a deterministic \emph{vertex-failure connectivity labeling scheme} of size $\tilde{O}(f^{2})$, improving the previous existential bound of $\tilde{O}(f^{3})$ \cite{parter2024connectivity,long2025connectivity} and the previous deterministic bound of $\tilde{O}(f^{4})$ \cite{long2025connectivity}.

\end{abstract}

\newpage
\setcounter{tocdepth}{2}
\tableofcontents
\newpage
\pagenumbering{arabic}

\section{Introduction}

\emph{Fault-tolerant oracles} are data structures that support efficient queries for fundamental graph properties, such as connectivity and distances, in the presence of a small batch of edge or vertex failures of size at most $f$. Since the pioneering work of Pătraşcu and Thorup \cite{patrascu2007planning}, extensive work has investigated fault-tolerant oracles for problems including pairwise connectivity \cite{duan2010connectivity,kapron2013dynamic,gibb2015dynamic,duan2020connectivity,long2022near,kosinas2023connectivity,long2024better}, global connectivity \cite{jiang2026new,kosinas2025optimal}, distances \cite{chechik2012f,chechik20171+,bilo2023approximate,HLRS26}, and diameter \cite{bilo2023fault}.

When designing fault-tolerant oracles for undirected graphs, a central goal is to achieve oracles that are both space-efficient, ideally with nearly linear space, and support fast queries. Remarkably, for pairwise connectivity, nearly-linear-space oracles have been achieved under both edge \cite{kapron2013dynamic,gibb2015dynamic,duan2020connectivity} and vertex failures \cite{long2022near,long2024better}, and these results are in fact conditionally almost-optimal.

In contrast, for pairwise distances, nearly linear space is unattainable even in the non-fault-tolerant setting if approximation is not allowed: the classic result of Thorup and Zwick \cite{thorup2005approximate} gives a distance oracle with approximation $2k-1$, space $O(kn^{1+1/k})$, and query time $O(k)$ for any $k\geq 1$, and this tradeoff between approximation and space is tight under the Erdős
girth conjecture.

\paragraph{Edge-Failure Distance Oracles.} 
Therefore, there has been substantial work on edge-failure distance oracles with different approximation-space tradeoffs. For example, some results \cite{chechik20171+,duan2022maintaining,dey2024nearly} provide exact or $(1+\epsilon)$-approximate answers, but require at least quadratic space and depend exponentially on $f$. Some other results \cite{bilo2023approximate, bilo2024improved} achieve mildly subquadratic space and small constant approximations such as $3+\epsilon$. In particular, there are two edge-failure distance oracles in the spirit of the Thorup-Zwick oracle that achieve analogous approximation-space tradeoffs while supporting fast queries. The first, due to Chechik, Langberg, Peleg, and Roditty \cite{chechik2012f}, attains approximation $(8k-2)(f+1)$, space $\Otilde(fkn^{1+1/k})$\footnote{Throughout the paper, we use $\Otilde(\cdot)$ to hide $\poly(\log n)$ factors.}, and query time $\Otilde(f)$. The second, due to Haeupler, Long, Roeyskoe, and Saranurak \cite{HLRS26}, attains approximation $O(k^4)$, space $\Otilde(f^4 n^{1+1/k})$, and query time $\Otilde(f^5 n^{1/k})$.

\paragraph{Vertex-Failure Distance Oracles.} Nevertheless, for vertex failures, despite substantial effort, most existing results only handle a single vertex failure \cite{baswana2013approximate,bernstein2008improved,bernstein2009nearly,chechik2020distance,duan2017improved,ren2022improved}. The only known distance oracles for multiple vertex failures are due to Duan, Gu, and Ren \cite{duan2021approximate}. For any constants $c \ge 1$ and $\epsilon > 0$, they give one oracle with approximation $O(f^{c+2}\log^{6}n)$, space $\Otilde(n^{2+1/c}f^{3})$, and query time $\Otilde(f^{2c+9})$, and another with approximation $1+\epsilon$, space $n^{2+1/c}\cdot (\log n/\epsilon)^{O(f)}$, and query time $\Otilde(f^{2c+6}/\epsilon)$. In particular, both oracles require superquadratic space. Moreover, the first oracle, whose space does not depend exponentially on $f$, cannot achieve a constant approximation even when $f$ is a constant. However, a line of work on \emph{vertex-fault-tolerant spanners} \cite{chechik2009fault,dinitz2011fault,bodwin2018optimal,bodwin2019trivial} shows that, without requiring fast queries, one can simultaneously achieve approximation $2k-1$ and size $O(n^{1+1/k}f^{1-1/k})$. This raises the following open problem.
\begin{center}
\textit{Is there a vertex-failure distance oracle with approximation $\poly(k)$,\\ space $\Otilde(n^{1+1/k})\cdot\poly(f)$, and query time $\Otilde(n^{1/k})\cdot \poly(f)$?}
\end{center}

\subsection{Our Results}

We answer the above open problem affirmatively with the following \Cref{thm:IntroOracle}.

\begin{theorem}
\label{thm:IntroOracle}
Given an undirected graph $G$ with non-negative, integral and polynomially bounded vertex and edge lengths, and given parameters $f\geq 1$ and $k\geq 1$, there is a randomized\footnote{Our oracle (and the distance labeling scheme) answers each query correctly with high probability $1-1/n^{10}$. Consequently, it supports a polynomial number of queries. This guarantee is standard in previous randomized fault-tolerant oracles, e.g. \cite{duan2020connectivity}. The success probability can be boosted to $1-1/n^{\Omega(f)}$ in either a while-box or black-box way. The latter incurs an extra factor of $f$ in the space, update time, and query time.} vertex-failure distance oracle with approximation $O(k^{6})$, space $\Otilde(f^{2}\cdot n^{1+1/k})$, update time $\Otilde(f^{5}\cdot n^{1/k})$, query time $\Otilde(f)$, and polynomial construction time.
\end{theorem}

Similar to some previous work \cite{duan2020connectivity,HLRS26}, our oracle supports \emph{separate} update and query phases: after a polynomial-time construction, in the update phase it processes a given failure set $F$ in extra $\Otilde(f^{5} \cdot n^{1/k})$ time, and then in the query phase it answers the approximate distance under $F$ between any given pair of vertices $x$ and $y$ in faster $\Otilde(f)$ time. 

Compared with the large-approximation oracle of \cite{duan2021approximate}, our oracle achieves significant improvements in nearly every aspect, while incurring only an additional factor of $n^{1/k}$ in the update time. Furthermore, when $k=\Omega(\log n)$, our oracle achieves $\poly\log(n)$ approximation, $\Otilde(f^{2}n)$ space, $\Otilde(f^{5})$ update time, and $\Otilde(f)$ query time. By the lower bounds for connectivity oracles, the space bound is within a factor of $f$ of the $\Omega(\min\{m,fn\})$ lower bound \cite{duan2020connectivity}, and the query time is nearly optimal \cite{henzinger2015unifying}. See \Cref{thm:OracleMain} for a detailed version of \Cref{thm:IntroOracle}.

\paragraph{Fault-Tolerant Labeling Schemes.} We also obtain new results for the harder model of \emph{fault-tolerant labeling schemes}, the distributed analogue of centralized oracles. Formally, a fault-tolerant labeling scheme assigns short \emph{labels} to vertices and edges such that, for any query specified by vertices $x,y$ and a failure set $F$, the desired graph property between $x$ and $y$ under failures, such as connectivity or distance, can be decoded solely from the labels of $x$, $y$, and the elements of $F$. In other words, unlike a centralized oracle, a labeling scheme has no random access to a global data structure and must instead evenly distribute information among vertices and edges. The main complexity measure is the (maximum) label size.

Since the introduction of this model in \cite{courcelle2007compact}, there has been a recent line of work on connectivity labeling schemes under edge failures \cite{DBLP:conf/podc/DoryP21,DBLP:conf/podc/IzumiEWM23,long2025connectivity} or vertex failures \cite{DBLP:conf/wdag/ParterP22a,parter2024connectivity,long2025connectivity}, as well as on distance labeling schemes under edge failures \cite{DBLP:conf/podc/DoryP21,HLRS26}. In particular, the labeling scheme in \cite{HLRS26} is deterministic and achieves approximation $O(k^{4})$ and label size $\Otilde(f^{4}n^{1/k})$.  Nevertheless, to the best of our knowledge, no nontrivial vertex-failure distance labeling scheme is known. A natural open problem, therefore, is to design such a scheme. Ideally, one would like to achieve tradeoffs analogous to those of the near-optimal non-fault-tolerant distance labeling scheme, which attains approximation $2k-1$ with label size $\Otilde(n^{1/k})$ \cite{thorup2005approximate}.

We present the first vertex-failure distance labeling schemes in \Cref{thm:IntroLabelScheme} below.

\begin{theorem}
\label{thm:IntroLabelScheme}
Given an undirected graph $G$ with non-negative, integral and polynomially bounded vertex and edge lengths, and given parameters $f\geq 1$ and $k\geq 1$, there is a randomized vertex-failure distance labeling scheme with approximation $O(k^{6})$ and label size $f^{3}\cdot n^{1/k}\cdot \log^{O(k)}n$. The labels can be computed in polynomial time.
\end{theorem}

In particular, when $k$ is a constant, our label size is $\Otilde(f^{3}n^{1/k})$. Although this improves on the $\Otilde(f^{4}n^{1/k})$ label size of the edge-failure distance labeling scheme in \cite{HLRS26}, we note that their scheme is deterministic, and may save a factor of $f$ if randomization is allowed. Moreover, when choosing $\omega(1)<k<o(\log n/\log\log n)$, the label size is $f^{3}n^{o(1)}$ with overhead only subpolynomial in $n$. See \Cref{thm:DistLabelMain} for a detailed version of \Cref{thm:IntroLabelScheme}.

Our approach also leads to a new result for vertex-failure \emph{connectivity} labeling schemes, as shown in \Cref{thm:IntroConnLabel}.

\begin{theorem}
\label{thm:IntroConnLabel}
Given an undirected graph $G$ and a parameter $f$, there is a deterministic vertex-failure connectivity labeling scheme with label size $\Otilde(f^{2})$. The labels can be computed in polynomial time.
\end{theorem}

Our result improves the previous state-of-the-art existential bound of $\Otilde(f^{3})$ \cite{parter2024connectivity,long2025connectivity}, which requires exponential construction time, as well as the previous deterministic bound of $\Otilde(f^{4})$ \cite{long2025connectivity}.

We note that obtaining \Cref{thm:IntroConnLabel} from the approach behind \Cref{thm:IntroLabelScheme} requires some straightforward white-box adaptation. In particular, the resulting labeling scheme is deterministic, since the relevant tools for connectivity are deterministic. See \Cref{thm:ConnLabel} for a detailed version of \Cref{thm:IntroConnLabel}.

\subsection{Our Techniques}

We briefly introduce our techniques here. A more detailed overview together with further discussion is given in \Cref{sect:Overview}. Our distance oracle and labeling scheme build on recent developments of \emph{length-constrained vertex expanders} \cite{haeupler2025length}, and we also introduce new expander-based \emph{shortcut sparsification} techniques. 

\paragraph{Length-Constrained Vertex Expanders.} The theory of \emph{length-constrained (LC) expanders} \cite{DBLP:conf/stoc/HaeuplerR022} has found numerous applications across several areas, while the recently developed \emph{LC vertex expander} variant \cite{haeupler2025length} remains relatively unexplored. As one of the earliest applications of this variant, our work brings LC vertex expanders to the study of fault-tolerant algorithms. 

In particular, LC vertex expanders allow us to extend the \emph{low-degree hierarchy} of \cite{duan2020connectivity}, a widely used framework for vertex-failure connectivity problems \cite{long2022near,parter2024connectivity,long2025connectivity}, to the distance setting. At a high level, after constructing an LC-vertex-expander hierarchy, we show that the expanding graph at each level admits a sparse neighborhood cover together with low-diameter, low-degree Steiner trees spanning the terminals of each cluster. Moreover, these trees have small overlap, in the sense that each vertex belongs to only a small number of them.

We believe this structure is likely to find further applications to distance-related problems under vertex updates.

\paragraph{New Expander-Based Shortcut Sparsification.} Many previous works on vertex-fault-tolerant algorithms are based on the idea of adding shortcuts \cite{duan2010connectivity,duan2020connectivity,long2022near,kosinas2023connectivity,parter2024connectivity,long2025connectivity}, but differ in the sparsification techniques they employ. We provide new expander-based sparsification based on a simple but generic principle of \emph{propagating heaviness/denseness}, which allows us to exploit the expanding property more broadly in an inter-cluster manner. This general principle captures a previous expander-based sparsification from \cite{long2025connectivity}, and also enables us to refine their existing notions of denseness and introduce new ones. These refinements and new notions are in fact crucial, especially for the labeling scheme setting, as they are needed for handle distances, and also lead to simpler and improved connectivity labeling schemes. For example, they allow us to bypass the coloring step in \cite{parter2024connectivity,long2025connectivity} and avoid a local sparsification step in \cite{long2025connectivity} using connectivity certificates \cite{DBLP:journals/algorithmica/NagamochiI92}.

\section{Overview}
\label{sect:Overview}

In this section, we give an overview of our vertex-failure approximate distance oracles and labeling schemes. We begin with the oracle setting. As a warm-up, we first consider the case where the input graph is a \emph{length-constrained (LC) vertex expander}. We then describe a conceptually simple distance oracle for general graphs that already achieves $\poly(f)\cdot n^{1+1/k}\cdot \log^{O(k)}n$ space, but with slightly worse $\poly(k)$ approximation and undesirable $\log^{O(k)} n$ overheads. Next, we introduce our new \emph{shortcut sparsification} technique, which improves the approximation, removes these unnatural overheads, and still preserves $\poly(f)\cdot n^{1+1/k}$ space. This sparsification is also crucial for the design of our labeling schemes. Finally, we briefly explain how to obtain labeling schemes by combining it with further shortcut sparsification techniques.

\subsection{Warm-Up: The Expander Case}

We first assume the input graph $G=(V,E)$ is an LC vertex expander. In fact, when assuming the input graph is some kind of expander, all previous expander-based fault-tolerant connectivity or distance algorithms \cite{patrascu2007planning,long2025connectivity,HLRS26} share a similar starting point. Namely, they all exploit the robustness of expanders under edge or vertex deletions. We state this property in the following lemma informally, phrased in the setting of LC vertex expanders.

\begin{lemma}[LC-Expansion Property]
\label{lemma:OverviewLCEProperty}
Let $G$ be a weighted graph with terminals $A\subseteq V$. Suppose $A$ is $h_{\ed}$-length expanding in $G$. Then for any two terminal sets $R_{1},R_{2}\subseteq A$, if
\begin{itemize}
\item $|R_{1}|,|R_{2}|\geq \rho_{\heavy} := \Otilde(f\cdot n^{1/k})$, and
\item $\dist_{G}(u,v)\leq h_{\ed}$ for all pairs $(u,v)$ of $u\in R_{1}$ and $v\in R_{2}$,
\end{itemize}
then for any failed vertex set $F$ of size at most $f$, the distance between $u$ and $v$ under failures $F$ is at most $h_{\ed}\cdot k$, i.e., $\dist_{G\setminus F}(u,v)\leq h_{\ed}\cdot k$.
\end{lemma}

Since we mainly exploit the above LC-expansion property, the precise definition of ``$A$ is $h_{\ed}$-length expanding in $G$'' is not important. Moreover, in this case, the terminal set $A$ is simply the whole set $V$. 

In addition, we need two more building blocks: \emph{sparse neighborhood covers}, augmented with associated \emph{low-diameter low-degree Steiner trees}.

Sparse neighborhood covers are commonly used in distance-related problems. Suppose we only want a non-fault-tolerant distance oracle, and only need to answer a decision query asking whether $\dist_{G}(x,y)>h$ or $\dist_{G}(x,y)\leq hk$. Then storing a sparse neighborhood cover with covering radius $h_{\cov} = h$, diameter $h_{\diam} = hk$ and width $\omega = \Otilde(n^{1/k})$ suffices. Recall that a neighborhood cover is a collection of vertex sets, called \emph{clusters}, such that every close (distance $\leq h_{\cov}$) vertex pair is covered by a cluster, every cluster has (strong) diameter $h_{\diam}$, and every vertex appears in at most $\omega$ clusters.

However, when dealing with vertex failures, we cannot just answer the query by checking whether there exists a cluster $S$ covering both $x$ and $y$, simply because $S$ can contain failed vertices (thus $S$ is \emph{affected}) and no longer certifies the $x$-$y$ distance. Therefore, we need additional structure to recover distances for affected clusters. This is where low-diameter low-degree Steiner trees come in.

\begin{lemma}
\label{lemma:OverviewSteinerTree}
Set $h_{\ed} = h_{\diam}$. When $A$ is $h_{\ed}$-length expanding in $G$, there exists a collection of Steiner trees $\{T_{S}\mid \text{clusters }S\}$ one-to-one corresponding to clusters such that $T_{S}$ spans all terminals $A\cap S$. Moreover, each $T_{S}$ has diameter $h_{\tree} = h_{\ed}\cdot k^{2}$ and maximum degree $\Delta = n^{O(1/k)}$, and each vertex appears in at most $\gamma = n^{O(1/k)}$ trees.
\end{lemma}

Note that the idea of using a \emph{low-degree} Steiner tree as the \emph{backbone} of a cluster was first introduced by \cite{duan2020connectivity} on connectivity oracles. Later, \cite{long2022near} shows that a (non-LC) vertex expander always admits a low-degree spanning tree. We generalize this argument to the distance setting by showing that in an LC vertex expander, we can augment the neighborhood cover with the above low-diameter low-degree Steiner tree, at the cost of an extra factor of $k$ in the length slack (i.e., from $k$ in \Cref{lemma:OverviewLCEProperty} to $k^{2}$ in \Cref{lemma:OverviewSteinerTree}).

Now, with the Steiner tree $T_{S}$, observe that $f$ failed vertices can only break $T_{S}$ into at most $\Delta\cdot f$ \emph{residual subtrees}. However, using a vertex-failure connectivity oracle to reconnect these residual subtrees would yield a $\Delta \cdot \poly(k)$ approximation, since each subtree contributes an additive $h_{\tree}$ loss (this is the approach in \cite{chechik2012f} which attains $O(kf)$ approximation for edge failures). Fortunately, as observed in \cite{HLRS26}, combining the backbone tree with the LC-expansion property gives a more fine-grained argument as follows. 

Consider a cluster $S$ (or equivalently, the subgraph $G[S]$) that covers the $x$-$y$ shortest path after the failures. For each residual subtree of $S$, we call it \emph{heavy} if it contains more than $\rho_{\heavy}$ terminals; otherwise it is \emph{light}. Now for each heavy residual subtree $\tau$, we contract all terminals of $S$ (equivalently all vertices of $S$ since $A=V$) spanned by $\tau$ into a single node $\tau_{\nd}$ called a \emph{subtree node}. We assign node length $h_{\tree}$ to $\tau_{\nd}$, meaning that when we go through $\tau_{\nd}$, we are willing to pay an additive loss of $h_{\tree}$, the diameter of the subtree. However, notice that each pair of (terminal sets of) heavy subtrees satisfies the conditions in \Cref{lemma:OverviewLCEProperty}, so it is feasible to add a star on all subtree nodes, and set the node length of the center of the star (called a \emph{cluster node}) to be $h_{\ed}\cdot k$, meaning that we are using the LC-expansion property to jump from one heavy subtree to another. Finally, observe that the modified graph (called the \emph{query graph}) has roughly $f^{2}\cdot n^{O(1/k)}$ vertices, since there are not too many vertices on light residual subtrees. The distance between $x$ and $y$ in the query graph only loses $2h_{\tree} + h_{\ed}\cdot k$ additively, which implies an $O(k^{3})$ approximation.

\subsection{A Simple Oracle for General Graphs}

Although the underlying ideas are similar when the input graph is assumed to have some form of robustness, such as expansion or the existence of a low-degree spanning tree, previous work \cite{patrascu2007planning,duan2010connectivity,duan2020connectivity,long2022near,parter2024connectivity,long2025connectivity,HLRS26} has proceeded quite differently in extending these ideas to general graphs. The overall idea is again to construct some sort of small query graph $Q$ (explictly or implicitly), which captures the connectivity or approximate distance between $x$ and $y$ after failures. But the question is how to construct the query graph $Q$. 

In particular, the edge-failure distance oracle in \cite{HLRS26} builds on a \emph{nested LC-expander hierarchy}. We find, however, that nestedness is not useful in the vertex-failure setting, and in fact we can bypass it using a generic idea of adding shortcuts that appears in many vertex-failure connectivity algorithms \cite{duan2010connectivity,duan2020connectivity,long2022near,kosinas2023connectivity,parter2024connectivity,long2025connectivity}. The next challenge is then to sparsify the shortcut set, since the number of shortcuts directly affects the space of the oracle. There are many and different sparsification techniques in previous work. In this subsection, we will see that combining a simple sparsification in \cite{duan2020connectivity} with a \emph{tree-structured hierarchy of neighborhood covers} already gives a reasonably good vertex-failure distance oracle. However, if we want to replace the tree-structured hierarchy (which incurs the $\log^{O(k)}n$ overheads and worse approximation), or proceed to the labeling scheme setting, we need new sparsification techniques.

\paragraph{The LC-Vertex-Expander Hierarchy.} We start by defining the \emph{LC-vertex-expander hierarchy}. Since we remove the nestedness requirement, our hierarchy can be computed in a straightforward bottom-up manner, so it is conceptually simpler. It only guarantees \emph{weak expansion}, but this is sufficient for our purposes. We directly define the hierarchy without mentioning its construction.

The hierarchy is a partition $\{A_{0},A_{1},...,A_{i_{\max}}\}$ of $V$ such that $i_{\max} = O(k)$ and for each level $i\in[0,i_{\max}]$, $A_{i}$ is weakly all-length expanding w.r.t. $G_{i}:=G\setminus A_{>i}$ and $G$. By “weakly all-length expanding” we mean the following property \Cref{lemma:OverviewWeakExpanding}. The key difference is that now the second condition requires $u,v$ to be close in $G_{i}$, but only guarantees that $u,v$ are close in $G\setminus F$ (i.e. the input graph after failures). Moreover, now the property holds for any $h_{\ed}$. 

\begin{lemma}
\label{lemma:OverviewWeakExpanding}
Suppose $A_{i}$ is weakly all-length expanding w.r.t. $G_{i}$ and $G$. Then for any $h_{\ed}\geq 1$ and each $R_{1},R_{2}\subseteq A_{i}$, if
\begin{itemize}
\item $|R_{1}|,|R_{2}|\geq \rho_{\heavy} := \Otilde(f\cdot n^{1/k})$, and
\item $\dist_{G_{i}}(u,v)\leq h_{\ed}$ for all pairs $(u,v)$ of $u\in R_{1}$ and $v\in R_{2}$,
\end{itemize}
then for any failed vertex set $F$ of size at most $f$, $\dist_{G\setminus F}(u,v)\leq h_{\ed}\cdot k$.
\end{lemma}

We note that the above definition of the expander hierarchy is very informal, and it is more technical in the main body. In particular, when working with LC vertex expanders, the cuts $C_{i}$ are fractional. We will use standard sampling techniques to turn it into an integral terminal set, and this is the only randomized part in our algorithm. Moreover, this also leads to multiple expanding graphs $G_{i,\alpha}$ at one level, one for each length scale.

\paragraph{The Tree-Structured Hierarchy of Neighborhood Covers.} Given the above expander hierarchy, a natural next step is to compute sparse neighborhood covers at different levels and for different length scales $\{h_{\alpha}:=2^{\alpha}\mid \alpha\in[1,\alpha_{\max}]\}$ where $\alpha_{\max} = O(\log n)$. However, this will incur some fundamental issues as we will discuss in the next subsection. Here, we proceed slightly differently.

For each level $i$ from $i_{\max}$ down to $0$, for each level-$(i+1)$ cluster $\hat{S}$ (with corresponding scale $\hat{\alpha}$) and each scale $\alpha\leq \hat{\alpha}$, construct a neighborhood cover ${\cal N}_{i,\alpha,\hat{S}}$ in the subgraph $G_{i}[\hat{S}]$ with covering radius $h_{\cov,\alpha} = h_{\alpha}$, strong diameter $h_{\diam,\alpha} = h_{\cov,\alpha}\cdot k^{2}$, and width $\omega = \Otilde(n^{1/k^{2}})$.  In particular, for the top level $i:=i_{\max}$, we choose $\hat{S}=V$ and $\hat{\alpha} = \alpha_{\max}$. Note that one vertex can appear in at most $n^{O(1/k)}\cdot \log^{O(k)}n$ clusters in total. 

For each cluster $S\in {\cal N}_{i,\alpha,\hat{S}}$, its terminals are $V_{\tmn}(S):=S\cap A_{i}$, and its level and length scale are $\level(S) := i$ and $\scale(S) :=\alpha$. 
Lastly, we also compute the low-diameter low-degree Steiner trees for clusters using the LC-expansion property. The tree $T_{S}$ for a cluster $S$ has diameter $h_{\tree,\alpha} = h_{\diam,\alpha}\cdot k^{2} = h_{\alpha}\cdot k^{4}$, where the $k^{2}$ factor is from \Cref{lemma:OverviewSteinerTree}.

A major benefit of this approach is that the resulting hierarchy forms a tree structure among clusters, which is well suited to our analysis, in particular when performing the \cite{duan2020connectivity}-sparsification, as we will see. For convenience, write $\hat{S}\anc S$ if $\hat{S}$ is a strict ancestor of $S$.

\paragraph{Witness Path Decomposition.} Fix a query $(x,y,F)$, and an $x$-$y$ shortest path $P$ in $G\setminus F$ called the witness path. We first define \emph{affected clusters}, which are clusters hit by any vertex in $\{x,y\}\cup F$. Intuitively, as discussed above, vertices in affected clusters will appear in the query graph (possibly as part of a contracted subtree node). We will talk about the query graph later, and let us first see where shortcuts are needed.

Assume $x,y\in A_{i_{\max}}$ for simplicity, and consider the following recursive decomposition of $P$ into \emph{leaf segments}. At the beginning, we have an \emph{initial segment} $\sigma$ which is the whole path $P$. For each segment $\sigma$, its length scale $\alpha_{\sigma}$ is such that $h_{\alpha_{\sigma}}\approx \leng(\sigma)$, and its \emph{interior} $\sigma^{\circ}$ is the subpath obtained by deleting two endpoints of $\sigma$. 

Starting from the initial segment $\sigma$, we assign to it level $i_{\sigma}:=i_{\max}$ and cluster $S_{\sigma}$ which is an arbitrary cluster in ${\cal N}_{i_{\max},\alpha_{\sigma},V}$ that covers the whole $\sigma$. Suppose $S_{\sigma}$ is affected; otherwise we are done. Next, we break $\sigma$ at the $A_{i_{\max}}$-vertices (which include $x$ and $y$ by our assumption) on it. This produces several child segments $\check{\sigma}$ whose endpoints lie in $A_{i_{\max}}$, while their interiors are disjoint from $A_{i_{\max}}$. For each $\check{\sigma}$, we assign to it level $i_{\check{\sigma}}:=i_{\max}-1$ and cluster $S_{\check{\sigma}}$ from ${\cal N}_{i_{\check{\sigma}},\alpha_{\check{\sigma}},S_{\sigma}}$ that covers its whole interior $\check{\sigma}^{\circ}$.

In general, every segment $\sigma$ other than the initial segment has interior $\sigma^{\circ}$ disjoint from $A_{> i_{\sigma}}$ and fully contained in its cluster $S_{\sigma}$. The endpoints of $\sigma$ are terminals of ancestor clusters of $S_{\sigma}$. If $S_{\sigma}$ is affected, we recursively decompose $\sigma$; otherwise, we declare it to be a leaf segment.

At the end, $P$ is decomposed into a collection of leaf segments $\sigma$. Let ${\cal S}_{\up}:=\{S\mid S\anc S_{\sigma}\text{ for some leaf segment $\sigma$}\}$ collect all strict ancestors of clusters of leaf segments. Observe that clusters in ${\cal S}_{\up}$ are all affected, and that each endpoint of a leaf segment are terminals in some $S\in{\cal S}_{\up}$. More importantly, ${\cal S}_{\up}$ are clusters of the non-leaf segments in the decomposition. Let $\cost({\cal S}_{\up})$ denote the sum of $h_{\alpha}$ for $\alpha:=\scale(S)$ over all clusters in ${\cal S}_{\up}$. Then $\cost({\cal S}_{\up})$ can be upper bounded by $O(i_{\max}\cdot \leng(P))$, which is important for our later analysis.

\paragraph{Adding Shortcuts.} Since each leaf segment $\sigma$ has its interior contained in the unaffected cluster $S_{\sigma}$, while each of its endpoints is a terminal of some affected cluster $\hat{S}\anc S_{\sigma}$. We naturally want to add shortcuts between ``neighbors'' of an unaffected cluster. 

Formally, for each cluster $S$, we create a copy $\vpt:=\point(v,S)$ for each $v\in V_{\tmn}(S)$, called \emph{points}. We say $\vpt$ has $\level(\vpt)$ and $\scale(\vpt)$ the same as those of its $\cluster(\vpt):=S$. Let $\core(S)$ denote the set of all points of $S$. Then the ``neighbors'' of an unaffected cluster $S$, denoted by $\Lambda(S)$ and called the \emph{surrounding set} of $S$, are those points of clusters $\hat{S}\anc S$ which are connected to $S$ by a short edge of length $\leq h_{\alpha}$. We simply add a clique on $\Lambda(S)$ to shortcut $S$, and each shortcut of this clique should have length $h_{\diam,\alpha} + 2h_{\alpha}$ for $\alpha = \scale(S)$, so that we do not underestimate the real distance.

Let $H$ be the \emph{shortcut graph} whose vertices are all the points and edges are all the shortcuts (we also transform original edges into shortcuts in a similar way). It is clear that the subgraph of $H$ induced by \emph{affected points}, denoted by $K:=\bigsqcup_{\text{affected $S$}}\core(S)$, can approximate $\dist_{G\setminus F}(x,y)$, as $H[K]$ preserves the shortcut version of $P$, denoted by $P_{H}$, up to a factor-$k$ approximation incurred by the neighborhood covers.

\paragraph{Exploit LC-Expansion Property.} Lastly, we construct the query graph $Q$ from the shortcut graph $H$. Recall that to further reduce the number of points, we have to contract points in the same heavy residual subtree $\tau$ into a subtree node $\tau_{\nd}$, and this subtree node has vertex length $h_{\tree,\alpha}$. If we only perform the contraction, the length of $P_{H}$ in $Q$ may become unbounded since it can pass through many subtree nodes. But recall that we have another step exploiting the LC-expansion property, which adds a star on subtree nodes belonging to the same affected cluster. Recall our discussion on leaf segments, which implies that $P_{H}$ will only pass through subtree nodes belonging to clusters in ${\cal S}_{\up}$. Using the stars, in fact we can simply modify $P_{H}$, to make sure that it pass through at most two subtree nodes of the same cluster in ${\cal S}_{\up}$. Therefore, we can upper bound the extra cost incurred by the subtree nodes by 
\[
k^{4}\cdot \cost({\cal S}_{\up}) = O(k^{4}\cdot i_{\max}\cdot \leng(P)) = O(k^{5})\cdot \leng(P),
\]
where $k^{4}$ is the ratio between $h_{\tree,\alpha}$ and $h_{\alpha}$.

This completes the definition of $Q$ and the analysis of one direction of the approximation guarantee. We do not discuss the construction of $Q$ or the underlying data structures here, as these are standard from \cite{duan2020connectivity}. The oracle size will be proportional to the size of the shortcut graph $H$, which is now roughly $n^{2+O(1/k)}\cdot \log^{O(k)}n$, since $H$ has at most $n^{1+O(1/k)}\cdot \log^{O(k)}n$ points.

\paragraph{Perform Shortcut Sparsification in \cite{duan2020connectivity}.} This sparsification simply replaces each clique on $\Lambda(S)$ (suppose $|\Lambda(S)|\geq f+1$) with a biclique between $\Lambda(S)$ and arbitrary $f+1$ points $B(S)\subseteq \Lambda(S)$. Note that there will be at least one non-failed point $w_{\ppt}$ in $B(S)$ under any failures $F$, since the points in $\Lambda(S)$ correspond to distinct vertices. Therefore, any old clique shortcut connecting $\upt,\vpt\in \Lambda(S)$ can be simulated by a two-hop path $\Ptilde_{H}:=(\upt,w_{\ppt})\circ (w_{\ppt},\vpt)$ using two new biclique shortcuts. We want to use the old shortcut $(\upt,\vpt)$ only when there is a corresponding leaf segment $\sigma$, whose $S_{\sigma} = S$. In this scenario, $w_{\ppt}$ belongs to an affected cluster and must appear in the query graph $Q$ (either individually or inside a subtree node).

When $w_{\ppt}$ belongs to a subtree node, passing through it needs to pay a large node length. This is precisely where the tree-structured hierarchy becomes useful.  Now, if $w_{\ppt}$ is contained in a subtree node $\tau_{\nd}$, then $\tau$ must belong to some cluster $\hat{S}\anc S$, which means $\hat{S}\in{\cal S}_{\up}$. Recall that our earlier argument shows that we can pass through subtree nodes of clusters in ${\cal S}_{\up}$ arbitrarily many times ``effectively for free''.

After the sparsification, the number of shortcuts is
\[
O(f\cdot \sum_{S}|\Lambda(S)|)\leq f\cdot m\cdot n^{O(1/k)}\cdot \log^{O(k)}n.
\]
To eliminate the dependence on $m$, we can replace the input graph with a vertex-fault-tolerant spanner with $m=O(f^{1-1/k}\cdot n^{1+1/k})$ in the very beginning, incurring only an extra factor of $k$ in the approximation. Finally, to obtain polynomial construction time, we in fact need to pay one additional factor of $k$ in the length slack of the LC-expansion property. Therefore, this simpler approach gives approximation $O(k^{6})$ for exponential-time construction and $O(k^{7})$ for polynomial-time construction.

\subsection{An Improved Oracle}

In this subsection, we discuss how to replace the tree-structured hierarchy of neighborhood covers with a natural collection of neighborhood covers built independently on each expanding graph $G_{i}$. This requires a new shortcut sparsification technique that exploits the LC-expansion property in an \emph{inter-cluster} manner, rather than solely within individual clusters. Furthermore, this new shortcut sparsification turns out to be crucial for labeling schemes.

First, we formally define the new neighborhood covers. For each level $i\in[0,i_{\max}]$ and each scale $\alpha\in[1,\alpha_{\max}]$, we construct a neighborhood cover ${\cal N}_{i,\alpha}$ with covering radius $h_{\cov,\alpha}=h_{\alpha}$, diameter $h_{\diam,\alpha} = h_{\alpha}\cdot k$, and width $\omega = \Otilde(n^{1/k})$. Naturally, we still define affected clusters to be those hit by $F\cup\{x,y\}$.

The earlier argument for decomposing the witness path $P$ still applies. Consequently, if we add a clique of shortcuts on each $\Lambda(S)$, we still obtain an oracle with $\Otilde(n^{2+1/k})$ space. The difficulty arises when one attempts to apply the shortcut sparsification technique of \cite{duan2020connectivity}. First, without the tree-structured hierarchy, when $S$ serves as $S_{\sigma}$ of a leaf segment, points in $\Lambda(S)$ may still belong to unaffected clusters. This means that choosing $f+1$ points as $B(S)$ may not hit any non-failed point that appears in the query graph. Even if we are lucky enough to hit such a point $w_{\ppt}$, there is no guarantee that $w_{\ppt}$ must belong to a cluster in ${\cal S}_{\up}$, and thus the case where $w_{\ppt}$ is inside a subtree node will be problematic (since now passing through $w_{\ppt}$ is no longer ``effectively for free'').

\paragraph{Propagate Heaviness to Denseness.}
To motivate our techniques, we first discuss a more flexible way of applying \Cref{lemma:OverviewWeakExpanding}. The standard application considers two sets $R_{1}$ and $R_{2}$ that are close in $G_{i}$ and are each \emph{$A_{i}$-heavy} (i.e. containing many $A_{i}$-vertices). We use the following extension: if two sets $Z_{1}$ and $Z_{2}$ are close in $G_{i}$ (but could be \emph{$A_{i}$-light}), and each $Z_{k}$ is close to $R_{k}$ in $G_{i}\setminus F$, then we can still conclude that $Z_{1}$ and $Z_{2}$ are close in $G\setminus F$. In other words, the heaviness of $R_{k}$ can propagate to $Z_{k}$ through ``$G_{i}\setminus F$-closeness''. To distinguish this extended setting from the standard one, we say that $Z_{k}$ is \emph{$A_{i}$-dense}. Moreover, we point out that $A_{i}$-heaviness/denseness can also be propagated using $G_{\check{i}}\setminus F$-closeness for any $\check{i}\leq i$, since $G_{\check{i}}\setminus F$-closeness implies $G_{i}\setminus F$-closeness.

\paragraph{Nearby Denseness of Clusters.} We use the above idea to define the \emph{nearby denseness of clusters}. Consider a cluster $S$ with level $i$ and scale $\alpha$. We use $\Lambda(S)$ to denote ``neighbor'' points and vertices simultaneously, depending on the context, i.e., it now includes all (points of) $A_{i'}$-vertices for all $i'>i$ that connect to $S$ via a short edge of length $\leq h_{\alpha}$. Let $\Lambda_{i'}(S) = \Lambda(S)\cap A_{i'}$. We say $S$ is nearby $A_{i'}$-dense if $|\Lambda_{i'}(S)|\geq \rho_{\heavy}$, meaning that when $S$ is unaffected, the whole cluster obtains $A_{i'}$-denseness from the $A_{i'}$-heavy set $\Lambda_{i'}(S)$.

\paragraph{Adjacent-Cluster Denseness of Points.} Next, we define \emph{adjacent-cluster denseness} for points. For a point $\vpt$ of a cluster $S$ with level $i$ and scale $\alpha$, we say $\vpt$ is adjacent-cluster dense if there exists an unaffected cluster $\check{S}$ with level $\check{i}<i$ and scale $\check{\alpha}\leq \alpha$ such that $\vpt\in\Lambda(\check{S})$ and $\check{S}$ is nearby $A_{i}$-dense. 

\paragraph{Shortcut Sparsification.} Consider a cluster $S$ at level $i$ and scale $\alpha$. There are three cases.

\medskip

\noindent{\underline{Case 1.}} First, for two levels $i_{1},i_{2}>i$, suppose $S$ is both nearby $A_{i_{1}}$- and $A_{i_{2}}$-sparse. We simply add a biclique between points in $\Lambda_{i_{1}}(S)$ and $\Lambda_{i_{2}}(S)$. This introduces only a small number of shortcuts.

\medskip

\noindent{\underline{Case 2.}} For two levels $i_{1}\geq i_{2}>i$, suppose $S$ is $A_{i_{1}}$-dense. In this case, we need one more observation: if a shortcut $(\upt,\vpt)$ with $\upt\in \Lambda_{i_{1}}(S)$ and $\vpt \in \Lambda_{i_{2}}(S)$ is used by some leaf segment, let $S_{1}$ and $S_{2}$ be the clusters of $\upt$ and $\vpt$, then we must have $\alpha_{1}:=\scale(S_{1})\geq \alpha_{2}:=\scale(S_{2})$. Namely, the larger-smaller relations of $(i_{1},i_{2})$ and $(\alpha_{1},\alpha_{2})$ are consistent. Furthermore, $S_{1}$ must intersect $S_{2}$.

Therefore, from the $\vpt$ side, when $\vpt$ is adjacent-cluster $A_{i_{1}}$-dense, we add an \emph{expansion-based link} from $\vpt$ to the cluster node of $S'_{1}$ (i.e. the center of the $S_{1}$-star) for all $S'_{1}$ at level $i_{1}$ intersecting $S_{2}$. By doing so, one can observe that, when we want to use the missing shortcut $(\upt,\vpt)$, there is a replacement path $(\vpt,S_{1,\nd})\circ(S_{1,\nd},\upt)$, and importantly $S_{1}\in {\cal S}_{\up}$ (and thus we can use $S_{1,\nd}$ freely).

\medskip

\noindent{\underline{Case 3.}} For two levels $i_{1}\geq i_{2}>i$, suppose $S$ is $A_{i_{1}}$-sparse but $A_{i_{2}}$-dense. In fact, this third case is more subtle because a low-level point can only obtain denseness at a high level, and not conversely. Therefore, we need new ideas involving \emph{satellite neighborhood covers} and \emph{satellite clusters}, which we will not explain here.

\medskip

Overall, since we only add bicliques between sparse levels of the surrounding set $\Lambda(S)$, now each cluster creates only $\Otilde(\rho_{\heavy}^{2}) = \Otilde(f^{2})\cdot n^{O(1/k)}$ shortcuts. The total number of shortcuts is then $\Otilde(f^{2})\cdot n^{1+O(1/k)}$, which is also the space complexity of our oracle. We note that the data structures used to add expansion-based links are separate from those used for shortcuts, and the former require even less space. 

For approximation, this approach removes one $k$ factor from the previous tree-structured hierarchy of neighborhood covers. Unexpectedly, however, there remain $\Otilde(m)\cdot n^{O(1/k)}$ “trivial” shortcuts from original edges, so the initial sparsification via a vertex-fault-tolerant spanner is still needed. Thus, this factor of $k$ cannot yet be avoided. We believe that our further sparsification techniques for labeling schemes may remove it, since our labeling schemes do not use vertex-fault-tolerant spanners.

In summary, the final approximation is $O(k^{5})$ for exponential-time construction and $O(k^{6})$ for polynomial-time construction. We point out that, for exponential-time construction, the remaining $k^{4}$ factor seems inherent to our approach. It comes from four different sources: sparse neighborhood covers, LC vertex-expander decompositions, low-degree low-diameter Steiner trees, and the number of levels in the hierarchy. Removing any of them would be an interesting open problem.

\subsection{The Labeling Schemes}
\label{sect:OverviewLabelingSchemes}

Lastly, we briefly discuss our approach to vertex-failure labeling schemes. Recall that a labeling scheme should distribute the data structure \emph{evenly} across vertices. In our oracle, the main source of unevenness lies in the data structures used to store shortcuts.

Our new sparsification attains some kind of evenness, in the sense that each cluster creates only about $f^{2}$ shortcuts. Therefore, it is affordable to store at each vertex the shortcuts created by clusters containing it, and in the query phase, we can obtain the full list of shortcuts created by affected clusters from the labels of $F\cup\{x,y\}$. This is particularly useful since we need to remove these invalid shortcuts in the query phase. Previous connectivity labeling schemes \cite{parter2024connectivity,long2025connectivity} achieved a similar goal using a different coloring technique, but at the cost of an additional factor of $f$ in the label size.

Therefore, to reduce the number of shortcuts incident to each vertex, we introduce the notions of \emph{upward denseness} and \emph{downward denseness}. While both notions are implicit in \cite{long2025connectivity}, we give a new definition of downward denseness here, together with a new method for detecting it. In both works, checking downward denseness is the most technically challenging part. As we explain below, the approach of \cite{long2025connectivity} does not extend to the distance setting. Our new approach is entirely expander-based and does not rely on connectivity certificates or fault-tolerant spanners. Thus it is simpler, more flexible, and even saves a factor of $f$ in the connectivity setting.

\paragraph{Upward Denseness of Points.} For technical reasons, in the labeling scheme we go back to the tree-structured hierarchy of neighborhood covers, which is easier to deal with. The overall idea is, if a point $\vpt$ at a low level $i$ connects to many points at a higher level $i'\geq i$ via unaffected shortcuts (i.e. shortcuts created by unaffected clusters or original edges). Then $\vpt$ will obtain $A_{i'}$-denseness from this \emph{upward neighbors} set. Naturally, if $\vpt$ is upward $A_{i'}$-sparse, then its upward shortcuts are stored in its label. If, on the other hand, $\vpt$ is upward $A_{i'}$-dense, then we add an expansion-based link from $\vpt$ to the cluster node corresponding to its level-$i'$ ancestor cluster. Again, the idea is to simulate the unstored shortcuts using a two-hop path of expansion-based links.

\paragraph{Downward Denseness of Points.} However, the above argument has an issue. Consider an unaffected shortcut $(\upt,\vpt)$ where $\upt$ is the higher endpoint at level $i'$ and $\vpt$ is the lower endpoint at level $i$. Let $S'$ and $S$ be the clusters of $\upt$ and $\vpt$ respectively. If this shortcut is missed because $\vpt$ is upward $A_{i'}$-dense, we have connected $\vpt$ to $S'_{\nd}$. But then how do we identify $\upt$, so that we can also connect $\upt$ to $S'_{\nd}$? Here comes the downward denseness. Very informally, we say $\upt$ is downward $A_{i'}$-dense if $\upt$ is connected to some upward $A_{i'}$-dense lower-level points via unaffected shortcuts. Now, if any such shortcut $(\upt,\vpt)$ is missed, it means $\upt$ is downward $A_{i'}$-dense, and we will connect $\upt$ to $S'_{\nd}$.

However, checking downward denseness is the most challenging part. \cite{long2025connectivity} uses a counting strategy, which requires some sort of low-arboricity property of the original shortcuts before the second-stage sparsification using upward and downward denseness. To obtain the low-arboricity property, they perform local sparsification by invoking a connectivity certificate \cite{DBLP:journals/algorithmica/NagamochiI92}. However, this approach does not work in the distance setting if we simply replace the connectivity certificate with a vertex-fault-tolerant spanner. Roughly speaking, the reason is that this sparsification simulates some shortcut with an unknown path in the spanner, and it will be problematic if this path goes into a cluster not in ${\cal S}_{\up}$. We are not going to explain our approach here, and we refer to \Cref{sect:SolveDownDense} for more details.

\medskip

In summary, the deterministic vertex-failure labeling scheme in \cite{long2025connectivity} has label size $\Otilde(f^{4})$. Our sparsification based on adjacent-cluster denseness removes an $f$ factor. In the distance setting, we use a new definition and approach for downward denseness, which leads to the final bound of $f^{3}\cdot n^{O(1/k)}\cdot\log^{O(k)}n$. In the connectivity setting, the new approach can even shave one more $f$ factor, achieving a deterministic bound of $\Otilde(f^{2})$.

\section{Preliminaries}

Throughout the paper, we use the standard graph theoretic notation. All graphs in this paper are undirected. For any graph, we use $V(\cdot)$ and $E(\cdot)$ to denote its vertex set and edge set, respectively. In particular, when $G$ is the input graph, we write $V$ and $E$ for its vertex and edge sets, and let $n:=|V|$ and $m:=|E|$. 

Given a graph $G=(V(G),E(G))$, for any $S\subseteq V(G)$, we let $G[S]$ denote the subgraph of $G$ induced by vertices in $S$. Also, for any $S\subseteq V(G)$, we use $G\setminus S$ to denote the graph after removing vertices in $S$ and edges incident to them. Similarly, for any $F\subseteq E(G)$, $G\setminus F$ denotes the graph after removing edges in $F$. 

\paragraph{Lengths.} A graph $G$ in this paper is generally weighted, meaning that it has both vertex lengths and edge lengths. We use $\ell_G: V(G)\cup E(G)\to \mathbb{R}_{\ge 0}$ to denote its length function on vertices and edges. When $G$ is clear from the context, we usually omit the subscript $G$. 

We require that the input graph has non-negative and integral lengths which are upper bounded by $W=\poly(n)$. In fact, the assumption $W=\poly(n)$ is essentially without loss of generality. Indeed, by paying an additional $\log(W)$ factor in the space, update time, and query time, standard scaling techniques allow us to reduce to the case $W=O(n^{3})$, since we allow a constant approximation. 

The following \Cref{assume:input} about the input graph is also without loss of generality.
\begin{assumption}
\label{assume:input}
we assume the following about vertex and edge lengths on the input graph $G$ without loss of generality.
\begin{itemize}
\item $G$ has only edge lengths, meaning that every vertex in $G$ has length $0$.
\item The edge lengths in $G$ are all positive integers.
\item There are at most $O(\log n)$ distinct edge lengths in $G$.
\end{itemize}
\end{assumption}

\paragraph{Distances and Diameters.} Let $G$ be a graph with length function $\ell$. This quantity includes the lengths of both endpoints, and when $P$ is not simple, counts each occurrence of a vertex or edge on $P$ with multiplicity. 

For each pair of vertices $u,v$ in $V(G)$, the \emph{distance} between them, denoted by $\dist_{G}(u,v)$, is the length of the shortest $u$-$v$ path. Note that $\dist_{G}(u,u) = \ell(u)$. More generally, for two vertex sets $S,T\subseteq V(G)$, $\dist_{G}(S,T) = \min_{u\in S,v\in T}\dist_{G}(s,t)$. The \emph{weak diameter} between $S$ and $T$, denoted by $\diam_{G}(S,T)$, is $\max_{u\in S,v\in T}\dist_{G}(s,t)$. For simplicity, we write $\diam_{G}(S)$ for $\diam_{G}(S)$. The \emph{strong diameter} of $S$ is $\diam_{G[S]}(S)$, i.e. the diameter of the induced subgraph.

\subsection{Length-Constrained Objects} 

Throughout this section, we consider a graph $G$ with length function $\ell$.

\paragraph{Moving Cuts.} A \emph{(vertex) moving cut} $C:E(G)\cup V(G)\to [0,1]$ assigns to each edge and vertex a non-negative real number between $0$ and $1$. A moving cut $C$ is \emph{$h$-length} if each value is further a multiple of $\frac{1}{h}$.
The size of $C$ is defined as $|C| := \sum_e C(e) +\sum_{v} C(v)$. 

Given a length parameter $h\geq 1$, we denote by $G - h\cdot C$ the graph $G$ with lengths $l_G + h\cdot C$, where $h\cdot C$ is $C$ scaled up by $h$. This notation can be used even when $C$ is not $h$-length. 

\begin{remark}
In literature involving length-constrained expanders, for an $h$-length moving cut, usually $G - C$ is used to denote the graph $G$ with lengths $\ell_{G} + h\cdot C$. We use a slightly different notation since we will use it even when $C$ is not $h$-length.
\end{remark}

\paragraph{Node Weightings.} A node weighting $A:V(G)\to \mathbb{R}_{\geq 0}$ assigns each vertex a non-negative real number. The size of $A$ is $|A| = \sum_{v}A(v)$. A node weighting is \emph{pure} if $A(v)$ at each vertex $v$ is either $0$ or $1$. Equivalently, a pure node weighting $A$ represents a vertex set $\{v\mid A(v) = 1\}$, so when $A$ is pure, we will use $A$ to denote this set instead. 

\paragraph{Demands.} A \textit{demand} $D : V \times V \rightarrow \mathbb{R}_{\geq 0}$ assigns a non-negative real number to each ordered pair of vertices $u, v \in V$. The size $|D|$ of a demand is defined as $|D| := \sum_{u, v \in V} D(u, v)$. A demand $D$ is called \textit{$A$-respecting} for a node weighting $A$ if for each vertex $u$, $\sum_{v}D(u,v) \leq A(u)$. A demand $D$ is called \textit{$h$-length-constrained} (or $h$-length for short) if it assigns positive demand values only to pairs of vertices that are within distance at most $h$, i.e. for all $(u, v) \in \supp(D)$, $\dist_G(u, v) \leq h$.

\subsubsection{Length-Constrained Vertex Expanders}

\paragraph{$h$-Length Separation.} Let $C$ be an $h$-length moving cut. We say a pair of vertices $u, v \in V$ are \emph{$h$-length separated} by $C$ if their distance in $G - h\cdot C$ is strictly larger than $h$, i.e. $\dist_{G - h\cdot C}(u, v) > h$. For an $h$-length demand $D$, the \emph{$h$-length separated demand value} $\sep_h(C, D)$ of $C$ and $D$ is the total demand value between pairs of vertices the cut $h$-separates, i.e.
\begin{equation*}
    \sep_h(C, D) := \sum_{\substack{u, v \in V\\\dist_{G - h\cdot C}(u, v) > h}} D(u, v).
\end{equation*}

\paragraph{$h$-Length Sparsity.} The \textit{$h$-length sparsity} of an $h$-length moving cut $C$ with respect to a \emph{demand} $D$ is the ratio of $C$'s size to the $h$-length separated demand value of $C$ and $D$, i.e.
\begin{equation*}
    \spars_h(C, D) := |C| / \sep_h(C, D).
\end{equation*}
For a \textit{length slack} $s \geq 1$, the \textit{$(h, s)$-length sparsity} of a $hs$-length moving cut $C$ with respect to a node weighting $A$ is defined as the \emph{minimum} $(hs)$-length sparsity of $C$ with respect to any \emph{$h$}-length $A$-respecting demand $D$, i.e.
\begin{equation*}
    \spars_{h, s}(C, A) := \min_{\text{$A$-respecting, $h$-length $D$}} \spars_{h s}(C, D).
\end{equation*}
We say a moving cut $C$ is $(h, s)$-length $\phi$-sparse with respect to $A$ if $\spars_{h, s}(C, A) \leq \phi$.

\paragraph{Length-Constrained Expansion.} For a length $h$, a length slack $s \geq 1$ and a sparsity $\phi$, a graph $G$ is $(h, s)$-length $\phi$-expanding for a node weighting $A$ if there exists no $(hs)$-length cut $C$ that has $(h, s)$-length sparsity with respect to $A$ strictly less than $\phi$, i.e. the following is satisfied:
\begin{equation*}
    \phi \leq \min_{\text{$(hs)$-length cut } C} \spars_{h, s}(C, A).
\end{equation*}
In this case, we equivalently say $A$ is $(h, s)$-length $\phi$-vertex-expanding in $G$, or $G$ is an $(h,s)$-length $\phi$-vertex expander w.r.t. $A$.

\paragraph{Length-Constrained Flows.} A \textit{multicommodity flow} $F$ in $G$ is a function that assigns a \textit{flow value} $F(P) \geq 0$ to each simple path $P$ in $G$. A path $P$ is a \textit{flow path} of $F$ if $F(P) > 0$, i.e. $P \in \supp(F)$. The \textit{value} $|F|$ of the flow is the total flow across all paths, i.e. $|F| = \sum_P F(P)$.

The \textit{vertex congestion} $\mathrm{cong}_F(e)$ of a vertex is the ratio of total flow of paths using $v$, i.e. $\mathrm{cong}_F(e) := \sum_{v \in P} F(P)$, and the vertex congestion of the flow is the maximum congestion of any vertex. The \textit{length} of the flow is the maximum length $\leng_G(P)$ of a flow path $P \in \supp(F)$ of $F$. A flow may be referred to as congestion-$\eta$ and length-$h$ even if it has congestion \emph{at most} $\eta$ and length \emph{at most} $h$.

The \textit{demand routed by the flow} $D_F$ is the demand where $D_F(u, v) = \sum_{P \text{ a $(u, v)$-path}} F(P)$, i.e. the $(u, v)$-demand is the total flow from $u$ to $v$. A demand $D$ is said to be \textit{routable} with congestion $\eta$ and length $h$ if there exists a congestion-$\eta$, length-$h$ flow $F$ routing $D$ (i.e. $D = D_F$).

\paragraph{Routing Characterization of Length-Constrained Vertex Expanders.}

\begin{theorem}[Theorem 4.2 in \cite{haeupler2025length}]\label{thm:lce-routing}
    For any graph $G$, node weighting $A$, length $h \geq 1$, length slack $s \geq 1$ and sparsity $\phi < 1$, the following hold:
    \begin{itemize}
        \item If $A$ is $(h, s)$-length $\phi$-vertex-expanding in $G$, then any $h$-length $A$-respecting demand on $G$ can be routed with vertex congestion $O(\log(n) / \phi)$ and length $hs$.
        \item If $A$ is not $(h, s)$-length $\phi$-vertex-expanding in $G$, then there exists a $h$-length $A$-respecting demand on $G$ that cannot be routed with congestion $1 / (6\phi)$ and length $(hs) / 2$.
    \end{itemize}
\end{theorem}

\paragraph{Length-Constrained (LC) Vertex-Expander Decompositions.}

\begin{definition}[LC Vertex-Expander Decomposition]
\label{def:StrongLCED}
Given a graph $G$ with length function $\ell$, a node weighting $A$, a length parameter $h\geq 1$, and a sparsity parameter $0<\phi< 1$, an \emph{$(h,s)$-length $\phi$-vertex-expander decomposition} w.r.t. $A$ with length slack $s$ and cut slack $\kappa$ is an $(hs)$-length moving cut $C$ of size at most $\kappa\cdot \phi|A|$ such that $A$ is $(h,s)$-length $\phi$-vertex-expanding in $G-(hs)\cdot C$.
\end{definition}

\begin{theorem}[Existence of LC Vertex-Expander Decompositions, Theorem 4.1 in \cite{haeupler2025length}]
\label{thm:VertexED}
Let $G$ be a graph with node-weighting $A$. Given a length parameter $h\geq 1$, a length slack $s$ and a sparsity parameter $0<\phi< 1$, there exists an $(h,s)$-length $\phi$-vertex expander decomposition $C$ of $A$ in $G$ with cut slack 
\[
\kappa = n^{O(1/s)}\cdot\log n.
\]
\end{theorem}

\begin{theorem}[Polynomial-Time Construction]
\label{thm:PolyED}
Let $G$ be a graph with node-weighting $A$. Given a length parameter $h\geq 1$, a length slack $s$ which is at least some constant and a sparsity parameter $0<\phi< 1$, there exists an algorithm that computes an $(h,s)$-length $\phi$-vertex expander decomposition $C$ of $A$ in $G$ with cut slack 
\[
\kappa = n^{O(1/\sqrt{s})}\cdot\Otilde(1).
\]
The algorithm runs in $\poly(n,h)$ time and it is deterministic.
\end{theorem}

The above \Cref{thm:PolyED} states that an LC vertex-expander decomposition can be computed in polynomial time, but with slightly worse cut slack compared to the existential result. Although this constructive result does not appear explicitly in the literature, it follows by a straightforward generalization of the polynomial-time construction for LC edge-expander decompositions. We refer to \cite{HLRS26} for a proof of the latter. We sketch the proof of \Cref{thm:PolyED}, highlighting the modifications needed relative to the proof in \cite{HLRS26} in \Cref{sect:PolyEDProof} rather than giving a complete argument.

\paragraph{Weakly LC Vertex-Expanding.} 

\begin{definition}[Weakly LC-Vertex Expanding] 
\label{def:WeakExpanding}
Consider two graphs $G'$ and $G$ that differ only in their length functions $\ell'$ and $\ell$. A node weighting $A$ is \emph{weakly $(h,s)$-length $\phi$-expanding w.r.t. $G'$ and $G$} if any $h$-length $A$-respecting demand on $G'$ can be routed in $G$ with vertex congestion $O(\log n/\phi)$ and length $hs$.
\end{definition}

In fact, in the following sections, we only use this definition with $G$ as the input graph. In other words, when we say $A$ is weakly $(h,s)$-length $\phi$-expanding in $G'$, it means $A$ is weakly expanding w.r.t. $G'$ and $G$.

We note that, in contrast to saying that $A$ is (strongly) $(h,s)$-length $\phi$-vertex-expanding in $G'$ (according to the routing characterization \Cref{thm:lce-routing}), the weak variant here considers $h$-length $A$-respecting demands on $G'$, but only requires routability in $G$ (rather than in $G'$ as in the strong variant). In other words, $A$ remains weakly expanding under arbitrary length increases to $G'$. This property generally does not hold for the strong variant.

\subsection{Length-Constrained Vertex-Expander Routing}

\begin{lemma}[Expander Routing]
\label{lemma:ExpanderRouting}
Let $G$ be a graph with length function $\ell$, and a node weighting $A$ that is (strongly) $(h,s)$-length $\phi$-vertex-expanding in $G$. Given parameters $h,\phi$ and an $h$-length $A$-respecting demand $D$, there is an algorithm that computes an integral flow $F$ that routes $D$ in $G$ with length $hs$ and vertex congestion $\gamma = \Otilde(1/\phi)$. The algorithm runs in $\poly(n,h)$ time and it is deterministic.
\end{lemma}

The above lemma on LC vertex-expander routing, although not stated explicitly in the literature, is a straightforward corollary of the multi-commodity length-constrained cutmatches in \cite{haeupler2023maximum}. We includes a proof of \Cref{lemma:ExpanderRouting} in \Cref{sect:ProofRouting} for completeness.

We note that the above integral flow $F$ is outputted in its \emph{path representation}, which includes each flow path $P$ explicitly and its flow value $F(P)$.

\subsection{Sparse Neighborhood Covers}

\begin{definition}[Sparse Neighborhood Covers]
Let $G$ be a graph with length function $\ell$. A \emph{cluster} $S$ is a subset of vertices. A \emph{(sparse) neighborhood cover} ${\cal N}$ is a collection of \emph{clusters}. A neighborhood cover ${\cal N}$ has covering radius $h_{\cov}$, (strong) diameter $h_{\diam}$, width/load $\omega$, if it satisfies the following.
\begin{itemize}
\item For each vertex $v$, its $h_{\cov}$-neighborhood, i.e. the set $\{u\mid \dist_{G}(u,v)\leq h_{\cov}\}$, is fully contained in some cluster in ${\cal N}$. We use ${\cal N}(v)$ to denote this cluster.
\item Each cluster $S\in{\cal N}$ has strong diameter at most $h_{\diam}$, i.e. $\diam_{G[S]}(S)\leq h_{\diam}$.
\item Each vertex $v$ is inside at most $\omega$ clusters.
\end{itemize}
\end{definition}

\begin{theorem}[Constructive Sparse Neighborhood Covers \cite{awerbuch1998near}]
\label{thm:ConstructNeighborhoodCover}
Let $G$ be a graph with length function $\ell$. Given a length parameter $h\geq 1$ and a length slack $s\geq 1$, there is an algorithm that 
\begin{itemize}
\item computes a neighborhood cover ${\cal N}$ in $G$ with covering radius $h_{\cov} = h$, diameter $h_{\diam} = hs$ and width $\omega = O(s\cdot n^{1/s})$.
\item for each vertex $v$, identifies a cluster ${\cal N}(v)\in{\cal N}$ that fully contains the $h_{\cov}$-neighborhood of $v$. 
\end{itemize}
The algorithm runs in $\poly(n)$ time and it is deterministic.
\end{theorem}

We note that in fact the sparse neighborhood cover algorithm in \cite{awerbuch1998near} runs in nearly linear time.

\section{Vertex-Failure Approximate Distance Oracles}

In this section, we present our vertex-failure approximate distance oracles result,  as stated below.

\begin{theorem}
\label{thm:OracleMain}
Let $G$ be an undirected graph with integral, nonnegative, polynomially bounded edge and vertex lengths $\ell$. Given parameters $f\geq 1$ and $k\geq 1$, there is a data structure that supports the following operations.
\begin{itemize}
\item (Update) Set the set of failed vertices to some given $F\subseteq V$ of size at most $f$.
\item (Query) Given two vertices $x,y$, outputs their approximate distance under the current failed set $F$, i.e., a value $d$ such that
\[
\dist_{G\setminus F}(x,y)\leq d\leq s\cdot \dist_{G\setminus F}(x,y).
\]
\end{itemize}
The approximation is $s=O(k^{6})$ for a polynomial-time construction, and $s=O(k^{5})$ for an exponential-time construction. The data structure takes $\Otilde(f^{2}\cdot n^{1+1/k})$ space, $\Otilde(f^{5}\cdot n^{1/k})$ update time, and $\Otilde(f)$ query time. Moreover, the data structure is randomized which can answer each query correctly with probability $1-n^{-10}$.
\end{theorem}

At the very beginning, we first replace the input graph $G$ by an $f$-vertex-fault-tolerant spanner. In \cite{parter2022ptimal}, there is an randomized algorithm that computes w.h.p. an $f$-vertex-fault-tolerant spanner with stretch $2k-1$ and size $\Otilde(f^{1-1/k}\cdot n^{1+1/k})$ in time $\Otilde(m)$. This step will increase the approximation by a multiplicative factor of $2k-1$, and the construction time by an additive factor of $\Otilde(m)$. Throughout this section, we assume the input graph $G$ is this spanner.

We first introduce some building blocks in \Cref{sect:OracleBlocks}. In \Cref{sect:OracleQueryGraph}, we will define the query graph $Q$. In \Cref{sect:OracleQueryAlgo}, we will show how to answer queries assuming the query graph and its all-pair distances are given. Lastly, in \Cref{sect:OracleConstruction}, we construct the oracle and discuss the update algorithm for computing $Q$ and its all-pair distances. We give a proof of \Cref{thm:OracleMain} in \Cref{sect:ProofOracleMain}, which includes pointer to each individual lemmas.

\subsection{Building Blocks}
\label{sect:OracleBlocks}

We first define some parameters as follows.

\paragraph{Basic Parameters.} Define \emph{based length scales} 
\[
\{h_{\alpha} = 2^{\alpha}\mid 1\leq \alpha\leq \alpha_{\max}\},
\]
where $\alpha_{\max} = \lceil \log(nW)\rceil = O(\log n)$. Recall that $W = \poly(n)$ upper bounds the lengths in $G$.

We also define
\[
\rho_{\heavy} = f\cdot O(\log n/\phi) + 1 \leq \Otilde(f)\cdot n^{O(1/k)}
\]
for some $\phi^{-1} = \Otilde(n^{O(1/k)})$ defined in \Cref{sect:OracleHierarchy}, and the constant factor behind $O(\log n/\phi)$ is determined by the first item of \Cref{thm:lce-routing}. This parameter $\rho_{\heavy}$ will be used in all concepts about \emph{heaviness} or \emph{denseness}.

\paragraph{The level function $\level(\cdot)$ and scale function $\scale(\cdot)$.} For any object (e.g., vertices, clusters, or Steiner trees, as introduced later) that is naturally associated with a \emph{level} (we will define level in the following \Cref{sect:OracleHierarchy}), we use $\level(\cdot)$ denote to its level. Similarly, we use $\scale(\cdot)$ to denote its corresponding length scale (between $1$ and $\alpha_{\max}$).

\subsubsection{LC Vertex Expander Hierarchy} 
\label{sect:OracleHierarchy}

\paragraph{Parameters for Expander Hierarchy.}

The parameters we define here rely on some other parameters 
\[
s_{\nc} = k,\qquad h_{\diam,\alpha} = \Theta(s_{\nc}\cdot h_{\alpha}),\qquad \hat{h}_{\diam,\alpha} = \Theta(s_{\nc}^{2}\cdot h_{\alpha}).
\]
which will be define in \Cref{sect:OracleNeighborhoodCover}. Define length parameters $\{h_{\ed,\alpha}\mid 1\leq \alpha\leq \alpha_{\max}\}$, length slack $s_{\ed}$ and expansion $\phi$, where
\begin{align*}
h_{\ed,\alpha} &:= 5 h_{\diam,\alpha} + 8h_{\alpha} = \Theta(s_{\nc}\cdot h_{\alpha}) = \Theta(k\cdot h_{\alpha}),\\
\hat{h}_{\ed,\alpha} &:= \hat{h}_{\diam,\alpha} + 4h_{\diam,\alpha} + 8h_{\alpha} = \Theta(s_{\nc}^{2}\cdot h_{\alpha}) = \Theta(k^{2}\cdot h_{\alpha})\\
s_{\ed} &:= \left\{
\begin{aligned}
&k, &&\text{for existence}\\
&k^{2}, &&\text{for polynomial-time construction}
\end{aligned}
\right.\\
\phi &:= 1/(n^{1/k}\cdot \kappa_{\ed}\cdot 3000k\cdot \log n \cdot s^{2}_{\nc}\cdot s_{\ed}\cdot \alpha_{\max}) = (\Otilde(1)\cdot n^{O(1/k)})^{-1},
\end{align*}
where $\kappa_{\ed} = \Otilde(1)\cdot n^{O(1/k)}$ by \Cref{thm:VertexED} or \Cref{thm:PolyED} and our choice of $s_{\ed}$. For better understanding, the reader may assume $s_{\ed} = k$ throughout this section.

\paragraph{Construction.} We construct an \emph{expander hierarchy} ${\cal H}$ with expansion $\phi$ and maximum level $i_{\max}$, where $i_{\max}$ will be bounded by \Cref{lemma:Hierarchy} later. The hierarchy ${\cal H}$ consists of pure node-weightings $\{A'_{i}\mid 0\leq i\leq i_{\max}\}$ and vertex moving cuts $\{C_{i,\alpha}\mid 1\leq i\leq i_{\max},1\leq\alpha\leq \alpha_{\max}\}$.

The construction of the hierarchy is standard. Roughly speaking, starting with the entire vertex set as the node-weighting, at each level, we perform expander decompositions on the current node-weighting, and let the new cuts be the next-level node-weighting. The only nonstandard aspect is that, at each level, we perform expander decompositions with two types of length parameters $h_{\ed,\alpha}$ and $\hat{h}_{\ed,\alpha}$. This corresponds to the two types of neighborhood covers $\mathcal{N}_{i,\alpha}$ and $\hat{\cal N}_{i,\alpha}$ introduced later in \Cref{sect:OracleNeighborhoodCover}. The formal construction algorithm of the hierarchy $\mathcal{H}$ is as follows.

\medskip

\noindent\underline{Initialization.} Initialize the bottom node-weighting $A'_{0}$ to be $V(G)$.

\medskip

\noindent\underline{Loop.} For each level $i=1,2,3,...$,
\begin{enumerate}
\item for each length scale $\alpha$, let $C'_{i,\alpha}$ be an $(h_{\ed,\alpha},s_{\ed})$-length $\phi$-vertex expander decomposition of $A'_{i-1}$ in $G$, and let $\hat{C}'_{i,\alpha}$ be an $(\hat{h}_{\ed,\alpha},s_{\ed})$-length $\phi$-vertex expander decomposition of $A'_{i-1}$ in $G$. Both $C'_{i,\alpha}$ and $\hat{C}'_{i,\alpha}$ can be computed by \Cref{thm:PolyED}.
\item Modified $C'_{i,\alpha}$ into $C_{i,\alpha}$ as follows.
\begin{equation*}
C_{i,\alpha}(v\text{ or }e) =\left\{
\begin{aligned}
&C'_{i,\alpha}(v) + \sum_{e \text{ incident to }v}C'_{i,\alpha}(e)/2,\ &&\text{for vertex $v$}\\
&0, &&\text{for edge $e$}
\end{aligned}
\right.,
\end{equation*}
Obtain $\hat{C}_{i,\alpha}$ from $\hat{C}'_{i,\alpha}$ in a similar way.
\item Let $A'_{i}$ be a pure node-weighting obtained by sampling each vertex $v\in V(G)$ independently with probability 
\[
1500k\cdot \log n\cdot s^{2}_{\nc}\cdot s_{\ed}\cdot \sum_{1\leq \alpha\leq \alpha_{\max}}\left(C_{i,\alpha}(v) + \hat{C}_{i,\alpha}(v)\right).
\]
\end{enumerate}

\noindent{\underline{Terminate.}} We reach the top level $i_{\max}:=i$ once $A'_{i}$ is empty.

\medskip

We make several remarks about the algorithm above. Step 1 is the main step of building the hierarchy by performing the decompositions. Step 2 moves the cut values on edges to their incident vertices, which is mainly for simplifying the analysis. For example, this ensures that when applying the modified cuts $C_{i,\alpha}$ and $\hat{C}_{i,\alpha}$ on $G$ (i.e., the weakly expanding graphs defined in \Cref{sect:OracleExpandingGraphs}), the edge lengths will not change. Step 3 is a sampling procedure that rounds the fractional moving cuts to integral vertex cuts. We note that this is a standard technique in the literature of LC-expanders (see e.g. \cite{DBLP:conf/stoc/HaeuplerR022}), and this is the only randomized part of our algorithm\footnote{The edge-failure distance oracles in \cite{HLRS26} also rely on (some stronger form of) LC-expander hierarchies, and they shows how to derandomize this part using the path restoration lemma of \cite{DBLP:journals/dc/AfekBKCM02}. However, in the vertex-failure setting, there is no analogous lemma.}. We choose the sampling rate so as to establish the following \Cref{lemma:Hierarchy} and \Cref{lemma:HighProbabilityEvent}.

\paragraph{Bounding the Number of Levels.} We conclude the construction of the hierarchy by bounding its number of levels, as shown in \Cref{lemma:Hierarchy}. The bound follows from the simple observation that the number of terminals decreases geometrically with each level.

\begin{lemma}
\label{lemma:Hierarchy}
With probability $1-n^{-10}$, the hierarchy ${\cal H}$ has maximum level
\[
i_{\max} \leq 15k.
\]
\end{lemma}
\begin{proof}
For each level $i\geq 1$, by \Cref{thm:VertexED}, we have 
\[
|C_{i,\alpha}| = |C'_{i,\alpha}|\leq \kappa_{\ed}\cdot \phi\cdot |A'_{i-1}|,
\]
and similarly $|\hat{C}_{i,\alpha}|\leq \kappa_{\ed}\cdot\phi \cdot |A'_{i-1}|$.

By the construction of $A'_{i}$, in expectation, we have
\begin{align*}
|A'_{i}|&\leq 1500k\cdot \log n\cdot s_{\nc}^{2}\cdot s_{\ed}\cdot\sum_{\alpha}\left(|C_{i,\alpha}| + |\hat{C}_{i,\alpha}|\right)\\
&\leq 3000k\cdot \log n\cdot \alpha_{\max}\cdot s^{2}_{\nc}\cdot s_{\ed}\cdot \kappa_{\ed}\cdot \phi \cdot |A'_{i-1}|\\
&\leq |A'_{i-1}|/n^{1/k}.
\end{align*}
Therefore, in expectation, we have $|A'_{12k}| \leq n^{-11}$. By Markov, $|A'_{12k}|<1$, i.e., $|A'_{12k}| = 0$, with probability at least $1-n^{-10}$. Thus, we can conclude that $i_{\max}\leq 12 k$.

\end{proof}

\subsubsection{The Expanding Graphs} 
\label{sect:OracleExpandingGraphs}

Next, we define two classes of graphs obtained by applying the moving cuts in the hierarchy: strongly expanding graphs and weakly expanding graphs. Weakly expanding graphs are more central, and we will work with them throughout, except when constructing low-degree Steiner trees (\Cref{sect:SteinerTrees}).

\paragraph{The Strongly Expanding Graphs.} For each level $i\in[0,i_{\max}]$ and length scale $\alpha\in[1,\alpha_{\max}]$, define
\[
G^{\star}_{i,\alpha}:= G - (h_{\ed,\alpha}\cdot s_{\ed})\cdot C'_{i,\alpha}.
\]
By \Cref{thm:PolyED}, $A'_{i-1}$ is strong $(h_{\ed,\alpha},s_{\ed})$-length $\phi$-vertex-expanding in $G^{\star}_{i,\alpha}$. %
The only reason we need strongly expanding graphs in the construction of low-degree Steiner trees is that the expander-routing result in \Cref{lemma:ExpanderRouting} applies only to such graphs.

\paragraph{The Weakly Expanding Graphs.} Before defining the weakly expanding graphs, we first define the terminals $\{A_{i}\mid 0\leq i\leq i_{\max}\}$. For each level $i\in[0,i_{\max}]$, define the \emph{level-$i$ terminals} to be 
\[
A_{i} := A'_{i}\setminus A'_{>i}.
\]
We introduce the new terminal sets $\{A_i\}$, rather than using the old pure node-weightings $\{A'_i\}$, is because we want to partition the vertices into distinct levels. Formally speaking, $\{A_{i}\mid 0\leq i\leq i_{\max}\}$ forms a partition of $V$, so naturally, for each vertex $v\in V$, its $\level(v)$ is the unique level $i$ with $A_{i}\ni v$.

Now we define the weakly expanding graphs. For each length scale $\alpha\in[1,\alpha_{\max}]$, define the \emph{level-$i$ scale-$\alpha$ expanding graph} to be
\[
G_{i,\alpha} := G - ((h_{\ed,\alpha}\cdot s_{\ed})\cdot C_{>i} + (\hat{h}_{\ed,\alpha}\cdot s_{\ed})\cdot \hat{C}_{>i})
\]
where $C_{>i} = \sum_{i'>i,\alpha}C_{i',\alpha}$ and $\hat{C}_{>i} = \sum_{i'>i,\alpha}\hat{C}_{i',\alpha}$. This graph $G_{i,\alpha}$ is weakly expanding as formalized in \Cref{ob:WeakExpanding}.

\begin{observation}[Weakly Expanding]
\label{ob:WeakExpanding}
For each level $i$ and length scale $\alpha$, $A_{i}$ is weakly $(h_{\ed,\alpha},s_{\ed})$-length $\phi$-vertex-expanding and also weakly $(\hat{h}_{\ed,\alpha},s_{\ed})$-length $\phi$-vertex-expanding in $G_{i,\alpha}$.
\end{observation}
\begin{proof}
We focus on proving that $A_i$ is weakly
$(h_{\ed,\alpha},s_{\ed})$-length $\phi$-vertex-expanding in
$G_{i,\alpha}$, and the proof of the second half is identical.

Recall that $A'_{i}$ is strongly $(h_{\ed,\alpha},s_{\ed})$-length $\phi$-vertex-expanding in $G - (h_{\ed,\alpha}\cdot s_{\ed})\cdot C'_{i+1}$, i.e. the strongly expanding graph $G^{\star}_{i,\alpha}$. Now, observe that $A_i\subseteq A'_i$ and that distances in
$G_{i,\alpha}$ are no smaller than the corresponding distances in
$G^\star_{i,\alpha}$, so $A_{i}$ is weakly $(h_{\ed,\alpha},s_{\ed})$-length $\phi$-vertex-expanding in $G_{i,\alpha}$ by \Cref{def:WeakExpanding}.
\end{proof}

After defining the weakly expanding graphs, \Cref{lemma:HighProbabilityEvent} shows that the terminal set $A_{>i}$ can be viewed as an integral analogue of the moving cuts $C_{>i}$ and $\hat{C}_{>i}$. More precisely, for any path disjoint from $A_{>i}$ (or equivalently, any path that survives the integral cut $A_{>i}$), with high probability, it has length in $G_{i,\alpha}$ at most an additive $h_{\alpha}$ larger than its length in the original graph $G$. Thus, such a path is cut only slightly by the moving cuts $C_{>i}$ and $\hat{C}_{>i}$. In fact, \Cref{lemma:HighProbabilityEvent} fixes a path $P$ in advance and shows that the above statement holds simultaneously for all subpaths $P'$ of $P$ with high probability using a standard union bound. When we apply \Cref{lemma:HighProbabilityEvent} in the analysis of the query algorithm, the path $P$ will be chosen as an optimal shortest path.

\begin{lemma}
\label{lemma:HighProbabilityEvent}
Fix a path $P$ in $G$. The following events happen simultaneously with probability $1-n^{-10}$.
For all levels $i$, length scales $\alpha$ and all subpaths $P'$ of $P$, if $P'$ is disjoint from $A_{>i}$, then
\[
\leng_{G_{i,\alpha}}(P')\leq \leng_{G}(P') + h_{\alpha}\cdot \frac{i_{\max}-i}{i_{\max}}
\]
\end{lemma}
\begin{proof}
We first show that with probability $1-n^{-10}$, for all levels $i\in[1,i_{\max}]$ and subpaths $P'$ such that
\[
C_{i+1}(P')>1/(10\cdot s^{2}_{\nc}\cdot s_{\ed}\cdot i_{\max})\qquad\text{or}\qquad \hat{C}_{i+1}(P')>1/(10\cdot s^{2}_{\nc}\cdot s_{\ed}\cdot i_{\max}),
\]
we have $P'$ intersects $A'_{i}$. Recall that a vertex $v$ is sampled into $A'_{i}$ with probability 
\[
1500k\cdot \log n\cdot s^{2}_{\nc}\cdot s_{\ed}\cdot \sum_{1\leq \alpha\leq \alpha_{\max}}\left(C_{i,\alpha}(v) + \hat{C}_{i,\alpha}(v)\right).
\]
It means $A'_{i}\cap P$ has expected size $E[|A'_{i}\cap P|]\geq 10\log n$ (recall that $i_{\max}\leq 15k$). By Chernoff, $A'_{i}$ intersects $P$ with probability $1-n^{-13}$. Applying the union bound on all $i$ and $P'$ gives what we want.

The original lemma is a simple corollary of the above. To prove it, we use induction on the levels. The inequality holds for the level $i=i_{\max}$, all scales $\alpha$ and all subpaths $P'$ disjoint from $A_{>i}$, because $G_{i_{\max},\alpha} = G$. Now, consider a level $i<i_{\max}$, length scale $\alpha$, and a subpath $P'$ disjoint from $A_{>i}$. We have
\begin{align*}
\leng_{G_{i,\alpha}}(P') &= \leng_{G_{i+1,\alpha}}(P') + h_{\ed,\alpha}\cdot s_{\ed}\cdot C_{i+1}(P') + \hat{h}_{\ed,\alpha}\cdot s_{\ed}\cdot \hat{C}_{i+1}(P')\\
&\leq \leng_{G}(P') + h_{\alpha}\cdot \frac{i_{\max}-i-1}{i_{\max}} + 10\cdot h_{\alpha}\cdot s^{2}_{\nc}\cdot s_{\ed}\cdot (C_{i+1}(P') + \hat{C}_{i+1}(P'))\\
&\leq \leng_{G}(P') + h_{\alpha}\cdot \frac{i_{\max}-i}{i_{\max}}.
\end{align*}
The first inequality is by the induction hypothesis, $h_{\ed,\alpha} = h_{\alpha}\cdot s_{\nc}$ and $\hat{h}_{\ed,\alpha} = h_{\alpha}\cdot s_{\nc}^{2}$. The second inequality is by $C_{i+1}(P'),\hat{C}_{i+1}(P')\leq 1/(10 \cdot s^{2}_{\nc}\cdot s_{\ed}\cdot i_{\max})$. 
\end{proof}

Lastly, we conclude with two additional observations about weakly expanding graphs. \Cref{ob:LargestVertexLength} states that, all uncut vertices in $A_{\leq i}$ have small vertex length in the weakly expanding graph $G_{i,\alpha}$ (if $A_{>i}$ is viewed as an integral vertex cut). Moreover, the edge lengths remain unchanged from the original graph $G$ to the weakly expanding graph $G_{i,\alpha}$.

\begin{observation}
\label{ob:LargestVertexLength}
For each $G_{i,\alpha}$, any vertex $v\subseteq A_{\leq i}$ has vertex length $\ell_{G_{i,\alpha}}(v)\leq h_{\alpha}$, and any edge $e$ has edge length $\ell_{G_{i,\alpha}}(e) = \ell_{G}(e)$.
\end{observation}
\begin{proof}
First we show that a vertex $v\subseteq A_{\leq i}$ has $\ell_{G_{i,\alpha}}(v)\leq h_{\alpha}$. Assume the contrary. Then for some $i'>i$, $\sum_{\alpha} C_{i',\alpha}(v)\cdot(h_{\ed,\alpha}\cdot s_{\ed})\geq h_{\alpha}/i_{\max}$, but then $v$ will be sampled in $A'_{i'}$ with probability $1$, a contradiction. The statement on edge length is by our definition of cuts $C$ and $\hat{C}$.
\end{proof}

The following \Cref{ob:ScaleChange} states that for two weakly expanding graphs $G_{i_{1},\alpha_{1}}$ and $G_{i_{2},\alpha_{2}}$, if $G_{i_1,\alpha_1}$ has both level and length scale at least those of $G_{i_2,\alpha_2}$, then its vertex and edge lengths, normalized by the corresponding length scale, are no larger. The proof follows directly from the definitions and is therefore omitted. This observation is particularly useful for exploiting the LC-expansion property across different weakly expanding graphs.

\begin{observation}[Scale Changing]
\label{ob:ScaleChange}
For each level $i_{1}\geq i_{2}$ and length scale $\alpha_{1}\geq \alpha_{2}$, we have
\[
\ell_{G_{i_{1},\alpha_{1}}}(\cdot)\leq (h_{\alpha_{1}}/h_{\alpha_{2}})\cdot \ell_{G_{i_{2},\alpha_{2}}}(\cdot).
\]
\end{observation}

\subsubsection{The Neighborhood Covers}
\label{sect:OracleNeighborhoodCover}

\paragraph{Parameters for Neighborhood Covers.} For each length scale $\alpha$, define parameters
\begin{align*}
s_{\nc} &= k,\\
\omega &= O(n^{1/s_{\nc}}\cdot s_{\nc}) = \Otilde(n^{1/k}),\\
h_{\cov,\alpha} &= 2\cdot h_{\alpha}, %
&&h_{\diam,\alpha} = h_{\cov,\alpha}\cdot s_{\nc} = \Theta(k\cdot h_{\alpha}),\\
\hat{h}_{\cov,\alpha} &= 3\cdot h_{\diam,\alpha} = \Theta(k\cdot h_{\alpha}),%
&&\hat{h}_{\diam,\alpha} = \hat{h}_{\cov,\alpha}\cdot s_{\nc} = \Theta(k^{2}\cdot h_{\alpha}).
\end{align*}

\paragraph{Construction.} For each level $i\in[0,i_{\max}]$ and length scale $\alpha\in[1,\alpha_{\max}]$, construct the following using \Cref{thm:ConstructNeighborhoodCover}. 
\begin{itemize}
\item Construct a neighborhood cover ${\cal N}_{i,\alpha}$ in $G_{i,\alpha}$ with covering radius $h_{\cov,\alpha}$, diameter $h_{\diam,\alpha}$ and width $\omega$. Also for each vertex $v$, identify a cluster ${\cal N}_{i,\alpha}(v)\in{\cal N}_{i,\alpha}$ that fully contains the $h_{\cov,\alpha}$-neighborhood of $v$ in $G_{i,\alpha}$.

\item Construct a \emph{satellite neighborhood cover} $\hat{\cal N}_{i,\alpha}$ in $G_{i,\alpha}$ with covering radius $\hat{h}_{\cov,\alpha}$, diameter $\hat{h}_{\diam,\alpha}$ and width $\omega$. Similarly, for each vertex $v$, identify a satellite cluster $\hat{\cal N}_{i,\alpha}(v)\in \hat{\cal N}_{i,\alpha}$ containing $v$'s $\hat{h}_{\cov,\alpha}$-neighborhood.
\end{itemize}
Whenever we refer to a cluster $S$, we mean a cluster in some cluster ${\cal N}_{i,\alpha}$. By contrast, any $\hat{S}\in \hat{\cal N}_{i,\alpha}$ will be explicitly called a \emph{satellite cluster}. We explain the intuition behind satellite neighborhood covers in \Cref{ob:SatelliteCover}, after introducing the necessary notation.

\paragraph{Notations.} For each cluster $S\in {\cal N}_{i,\alpha}$, we define its \emph{terminal vertices} to be 
\[
V_{\tmn}(S):=S\cap A_{i}.
\]
We also pick an arbitrary vertex in $S$, denoted by $\pivot(S)$, as its \emph{pivot}.

For each level $i$, length scale $\alpha$, and cluster $S\in{\cal N}_{i,\alpha}$, and for every level $i'>i$ and length scale $\alpha'\geq\alpha$, we assign to $S$ the satellite cluster $\hat{\cal N}_{i',\alpha'}(\pivot(S))$, which we denote by
\[
\satellite_{i',\alpha'}(S):=\hat{\cal N}_{i',\alpha'}(\pivot(S)).
\]

\paragraph{The Surrounding Set $\Lambda(S)$.}
For each level $i$, length scale $\alpha$ and cluster $S\in{\cal N}_{i,\alpha}$, define
\[
\Lambda(S) := A_{>i} \cap\{v\mid v\in S\text{ or there is an edge in $G$ connecting $v$ and $S$ of length $\leq h_{\alpha}$}\}.
\]
Note that the input graph $G$ and all weakly expanding graphs $G_{i,\alpha}$ have the same edge lengths (see \Cref{ob:LargestVertexLength}), so the second term remains unchanged if we consider $G_{i,\alpha}$ instead of $G$. Intuitively, $\Lambda(S)$ consists of the high-level vertices that either belong to $S$ or are connected to $S$ by short edges. For convenience, for each level $i'>i$, we use
\begin{align*}
\Lambda_{i'}(S)&:= \Lambda(S)\cap A_{i'}.
\end{align*}
to denote the level-$i'$ vertices in the surrounding set of $S$. Naturally, define $\Lambda_{\leq i'}(S):=\Lambda(S)\cap A_{\leq i'}$

The following \Cref{ob:SurroundingDiam} basically bounds the diameter of the surrounding set of a cluster in different graphs, including the original graph after failures and the weakly expanding graphs.

\begin{observation}
\label{ob:SurroundingDiam}
Consider a level $i$, a length scale $\alpha$, and a cluster $S\in {\cal N}_{i,\alpha}$. For each level $i'>i$ and length scale $\alpha'\geq \alpha$, we have 
\[
\diam_{G_{i',\alpha'}}(\Lambda_{\leq i'}(S))\leq h_{\diam,\alpha'} + 4h_{\alpha'}.
\]
For any failed set $F$ that is disjoint from $S$, we have $\diam_{G\setminus F}(\Lambda(S)\setminus F)\leq h_{\diam,\alpha'} + 2h_{\alpha'}$.
\end{observation}
\begin{proof}
We show the first inequality. Consider two vertices $u,v\in \Lambda_{\leq i'}(S)$. we have
\[
\dist_{G_{i',\alpha'}}(u,v)\leq \diam_{G_{i',\alpha'}}(S) + 2h_{\alpha} + \ell_{G_{i',\alpha'}}(u) + \ell_{G_{i',\alpha'}}(v),
\]
since each of $u$ and $v$ is either inside $S$ or connected to $S$ through an edge of length at most $h_{\alpha}$. Then $\dist_{G_{i',\alpha'}}(u,v)$ is further bounded by $h_{\diam,\alpha'} + 4h_{\alpha'}$ by the following reasons. First $\ell_{G_{i',\alpha'}}(w)\leq h_{\alpha'}$ for each vertex $w\subseteq A_{\leq i'}$ by \Cref{ob:LargestVertexLength}. Moreover, since $\diam_{G_{i,\alpha}}(S)\leq h_{\diam,\alpha}$, we have $\diam_{G_{i',\alpha'}}(S)\leq h_{\diam,\alpha'}$ by \Cref{ob:ScaleChange}.

The second inequality can be proven similarly. In particular, the second term is now $2h_{\alpha'}$ instead of $4h_{\alpha'}$ because the input graph $G$ has zero vertex length.
\end{proof}

The following \Cref{ob:SatelliteCover} illustrates the motivation for defining satellite clusters: we use the satellite cluster $\hat{S}$ to cover both the cluster $S$ itself and all clusters $S'$ ``adjacent'' to $S$.

\begin{observation}
\label{ob:SatelliteCover}
For each level $i$, length scale $\alpha$, and cluster $S\in {\cal N}_{i,\alpha}$, let $\hat{S}=\satellite_{i',\alpha'}(S)$ be the satellite cluster assigned to $S$ for some other level $i'>i$ and length scale $\alpha'\geq \alpha$. Then 
\begin{itemize}
\item $\Lambda_{i'}(S)\subseteq \hat{S}$,
\item for each vertex $v\in \Lambda_{i'}(S)$, every cluster $S'\in{\cal N}_{i',\alpha'}$ containing $v$ has $S'\subseteq \hat{S}$.
\end{itemize}
\end{observation}
\begin{proof}
Let $u:=\pivot(S)$ be the pivot vertex of $S$. By definition, $\hat{S}$ covers the $\hat{h}_{\cov,\alpha'}$-neighborhood of $u$ in $G_{i',\alpha'}$. Recall that $\hat{h}_{\cov,\alpha'} = 3h_{\diam,\alpha'}$.

For each vertex $v\in \Lambda_{i'}(S)$, we have 
\[
\dist_{G_{i',\alpha'}}(u,v) - \ell_{G_{i',\alpha'}}(v)\leq h_{\diam,\alpha'} + h_{\alpha'}.
\]
Here, the term $h_{\diam,\alpha'}$ accounts for the diameter of $S$, after applying the scale-changing rule in \Cref{ob:ScaleChange}, while the term $h_{\alpha'}$ accounts for the possibility that $v$ is connected to $S$ by an edge of length at most $h_{\alpha}\leq h_{\alpha'}$. Therefore, we have $\dist_{G_{i',\alpha'}}(u,v)\leq h_{\diam,\alpha'} + 2h_{\alpha'}\leq \hat{h}_{\cov,\alpha'}$ and thus $v\in \hat{S}$.

Now we show the second claim. Consider a cluster $S'\in{\cal N}_{i',\alpha'}$ containing $v$. For each vertex $w\in S'$, we have
\[
\dist_{G_{i',\alpha'}}(u,w)\leq \dist_{G_{i',\alpha'}}(u,v) - \ell_{G_{i',\alpha'}}(v) + \diam_{G_{i',\alpha'}}(S') \leq (h_{\diam,\alpha'} + h_{\alpha'}) + h_{\diam,\alpha'}\leq \hat{h}_{\cov,\alpha'},
\]
as desired.
\end{proof}

\subsubsection{Low-Degree Low-Diameters Steiner Trees}
\label{sect:SteinerTrees}

\paragraph{Parameters for Steiner Trees.} Define 
\begin{align*}
s_{\tree} &= k,\\
h_{\tree,\alpha} &= 4h_{\ed,\alpha}\cdot s_{\ed}\cdot s_{\tree} = \Theta(k^{3}\cdot h_{\alpha}),\\
\Delta &= \Otilde(n^{1/s_{\tree}}\cdot \omega/\phi) = \Otilde(1)\cdot n^{O(1/k)},\\
\gamma_{\tree} &= \Otilde(n^{1/s_{\tree}}\cdot \omega/\phi)\cdot i_{\max}\cdot \alpha_{\max} = \tilde{O}(1)\cdot n^{O(1/k)}.
\end{align*}

\paragraph{Construction.} For each level $i\in[0,i_{\max}]$ and each length scale $\alpha\in [1,\alpha_{\max}]$, using the following \Cref{thm:SteinerTree}, construct a collection of Steiner trees $\{T_{S}\mid S\in{\cal N}_{i,\alpha}\}$ in $G$, each of which spans the terminals of its cluster and has diameter $h_{\tree,\alpha}$ and maximum degree $\Delta$. Moreover, each vertex appears in at most $\gamma_{\tree}$ many Steiner trees (over all levels, length scales and clusters).

\begin{theorem}
\label{thm:SteinerTree}
For each level $i$ and scale $\alpha$, there is an algorithm that computes in $G_{i,\alpha}$ a collection of Steiner trees $\{T_{S}\mid S\in{\cal N}_{i,\alpha}\}$ one-to-one corresponding to clusters in ${\cal N}_{i,\alpha}$ with the following properties.
\begin{itemize}
\item Each Steiner tree $T_{S}$ spans the terminals $V_{\tmn}(S)$ of its cluster $S$.
\item Each $T_{S}$ has diameter $h_{\tree,\alpha} = 4h_{\ed,\alpha}\cdot s_{\ed}\cdot s_{\tree}$.
\item Each tree has maximum degree $\Delta = \Otilde(n^{1/s_{\tree}}\cdot \omega/\phi)$.
\item Each vertex appears in at most $\gamma' = \Otilde(n^{1/s_{\tree}}\cdot \omega/\phi)$ Steiner trees.
\end{itemize}
\end{theorem}
\begin{proof}
For each cluster $S\in{\cal N}_{i,\alpha}$, we fix a \emph{virtual tree} $\wtilde{T}_{S}$, which is an arbitrary unit-edge-length tree (no need to be a subgraph of $G$) with vertex set $V(\wtilde{T}_{S}) = V_{\tmn}(S)$, maximum degree $\Delta' = n^{1/s_{\tree}}$, and diameter $2s_{\tree}$. It is clear that $\wtilde{T}_{S}$ can be partitioned into at most $j_{\max}:=2\Delta'$ matchings $\{M_{S,j}\}_{j\in[j_{\max}]}$.

For each $j\in j_{\max}$, we define a demand $D_{j}$ by setting $D_{j}(u,v)$ to be the number of (virtual) edges between $u$ and $v$ in matchings $M_{S,j}$ over all $S$. By definition, 
\begin{itemize}
\item $D_{j}$ is a $(\omega\cdot A_{i})$-respecting demand, so it is also $(\omega\cdot A'_{i})$-respecting since $A_{i}\subseteq A'_{i}$.
\item $D_{j}$ is $h_{\diam,\alpha}$-length in $G_{i,\alpha}$ since ${\cal N}_{i,\alpha}$ has diameter $h_{\diam,\alpha}$. Thus, $D_{j}$ is also $h_{\diam,\alpha}$-length in $G^{\star}_{i,\alpha}$, since the length function of $G_{i,\alpha}$ is entry-wise no smaller than that of $G^{\star}_{i,\alpha}$.
\end{itemize}
We then apply \Cref{lemma:ExpanderRouting} on the strong expanding graph $G^{\star}_{i,\alpha}$ with node weighting $\omega\cdot A'_{i}$, the demand $D_{j}$, length parameter $h_{\ed,\alpha}$ and sparsity parameter $\phi/\omega$. Recall that $h_{\ed,\alpha}\geq h_{\diam,\alpha}$, and that $A'_{i}$ is strongly $(h_{\ed,\alpha},s_{\ed})$-length $\phi$-vertex-expanding in $G^{\star}_{i,\alpha}$, or equivalently $(\omega\cdot A'_{i})$ is strongly $(h_{\ed,\alpha},s_{\ed})$-length $(\phi/\omega)$-vertex-expanding in $G^{\star}_{i,\alpha}$. Thus, we obtain (the path representation of) an integral flow $F_{j}$ that routes $D_{j}$ in $G^{\star}_{i,\alpha}$ with length $h_{\ed,\alpha}\cdot s_{\ed}$ and vertex congestion $\Otilde(\omega/\phi)$. We assign the flow paths of $F_{j}$ back to the matching edges in all matchings $\{M_{S,j}\mid S\}$. 

After doing this for every $j$, for each virtual tree $\wtilde{T}_{S}$ and each virtual edge $(u,v)$ on it, we have a flow path connecting $u$ and $v$ with length $h_{\ed,\alpha}\cdot s_{\ed}$. Moreover, the collection of flow paths over all virtual edges on all virtual tree has vertex congestion $\Otilde(\Delta'\cdot\omega/\phi)$.

Finally, for each cluster $S$, we compute the Steiner tree $T_{S}$ by selecting an arbitrary shortest path tree in the graph $R_{S}$, where $R_{S}$ is the union of the flow paths of all virtual edges in $\wtilde{T}_{S}$. We have the following. 

\begin{itemize}
\item 

$T_{S}$ has diameter $4h_{\ed,\alpha}\cdot s_{\ed}\cdot s_{\tree}$. Note that the graph $R_{S}$ has diameter $2h_{\ed,\alpha}\cdot s_{\ed}\cdot s_{\tree}$: each flow path has length at most $h_{\ed,\alpha}\cdot s_{\ed}$, and the diameter of $\wtilde{T}_{S}$ is $2s_{\tree}$. Therefore, $T_{S}$, as a shortest path tree in $R_{S}$, has a diameter at most twice the diameter of $R_{S}$.

\item 
$T_{S}$ has maximum degree $\Otilde(\Delta'\cdot\omega/\phi) = \Otilde(n^{1/s_{\tree}}\cdot \omega/\phi)$. This is because $R_{S}$ has maximum degree $\Otilde(\Delta'\cdot\omega/\phi)$. Recall that $R_{S}$ is the union of flow paths, and these flow paths have vertex congestion $\Otilde(\Delta'\cdot\omega/\phi)$.

\item 
A vertex can appear in at most $\Otilde(\Delta'\cdot\omega/\phi) = \Otilde(n^{2/s_{\tree}}\cdot \omega/\phi)$ many $T_{S}$. This is again due to the low congestion of flow paths.
\end{itemize}
\end{proof}

\subsection{The Query Graph}
\label{sect:OracleQueryGraph}

In this section, we define the query graph $Q$ for a given set $F$ of failed vertices with size $|F|\leq f$. To distinguish between vertices and edges in the input graph $G$ and the query graph $Q$, we call the vertices and edges in $Q$ \emph{nodes} and \emph{links}.

At a high level, the query graph $Q$ captures the distances among a small number of ``affected regions'' under the failures $F$. Since $Q$ is small, we can afford to construct it explicitly and precompute all pairwise distances during the update phase. During the query phase, the algorithm identifies two nodes of $Q$ representing affected regions close to the two query vertices and returns their precomputed distance (with some additive adjustment).

In \Cref{sect:OracleAffected}, we introduce \emph{affected clusters} and \emph{residual subtrees}, which provide the key structural ingredients underlying the notion of ``affected regions.'' We formally construct the query graph $Q$ in \Cref{sect:OracleQNode} and \Cref{sect:OracleQLink}, in which we specify the nodes and links of $Q$ respectively. Finally, in \Cref{sect:OracleLowerBound}, we show that the query graph $Q$ will never underestimate distances in $G\setminus F$.

\subsubsection{Affected Clusters and Residual Subtrees}
\label{sect:OracleAffected}

\paragraph{Affected Clusters.} For each neighborhood cover ${\cal N}_{i,\alpha}$ and cluster $S\in {\cal N}_{i,\alpha}$, $S$ is \emph{affected} if $S$ intersects $F$, otherwise it is \emph{unaffected}. Similarly, a satellite cluster $\hat{S}\in\hat{\cal N}_{i,\alpha}$ is affected if $\hat{S}$ intersects $F$, and unaffected if not.

\begin{lemma}
\label{lemma:NumAffectedCluster}
The number of affected clusters and affected satellite clusters is at most 
\[
O(i_{\max}\cdot\alpha_{\max}\cdot\omega)=\Otilde(f)\cdot n^{O(1/k)}.
\]
\end{lemma}
\begin{proof}
This is because each vertex can appear in at most $O(i_{\max}\cdot \alpha_{\max})\cdot\omega$ many (satellite) clusters. Recall that $\omega = \Otilde(n^{1/k})$ is the width of (satellite) neighborhood covers.
\end{proof}

\paragraph{Residual Subtrees.} For each neighborhood cover ${\cal N}_{i,\alpha}$ and \underline{affected} cluster $S\in{\cal N}_{i,\alpha}$, under failures $F$, its corresponding tree $T_{S}$ may break into several (possibly just one) subtrees $\tau$ in $T_{S}\setminus F$, called \emph{residual subtrees}. We note that although the tree $T_{S}$ of an unaffected cluster $S$ may also be broken by the failures $F$ (since $T_{S}$ may go outside $S$), we do \underline{not} define residual subtrees for this unaffected $S$.

For a residual subtree $\tau$ of some cluster $S$, we define its terminal set to be
\[
V_{\tmn}(\tau) := V_{\tmn}(S)\cap \tau.
\]

\subsubsection{Nodes in the Query Graph}
\label{sect:OracleQNode}

We will create four types of nodes to $Q$ and assign lengths to these nodes. The construction relies on the following notions of heavy and light residual subtrees.

\begin{definition}[Heaviness of Residual Subtrees]
\label{def:Heaviness}
A residual subtree $\tau$ is \emph{heavy} if $|V_{\tmn}(\tau)|\geq \rho_{\heavy}$, otherwise $\tau$ is \emph{light}.
\end{definition}

With \Cref{def:Heaviness} in place, we create a \emph{singleton node} for each non-failed terminal in a light residual subtree, a \emph{subtree node} for each heavy residual subtree, a \emph{cluster node} for each affected cluster, and a \emph{satellite node} for each satellite cluster. Formally, consider each level $i$ and length scale $\alpha$. For each affected cluster $S\in{\cal N}_{i,\alpha}$, add the following nodes. All such nodes $q_{\nd}$ have $\level(q_{\nd}):= i$, $\scale(q_{\nd}):=\alpha$ and $\cluster(q_{\nd}):=S$.

\begin{itemize}
\item (Singleton Nodes) For each vertex $v\in V_{\tmn}(S)\setminus F$ such that $v$ is inside a light residual subtree of $S$, create a \emph{singleton node} $v_{\nd}$ with node length $0$.

\item (Subtree Nodes) For each heavy residual subtree $\tau$ of $S$, create a \emph{subtree node} $\tau_{\nd}$ with node length $h_{\tree,\alpha}$.

\item (Cluster Nodes) Create a \emph{cluster node} $S_{\nd}$ with node length $h_{\ed,\alpha}\cdot s_{\ed} + 2h_{\diam,\alpha} + 4h_{\alpha} = \Theta(s_{\nc}\cdot s_{\ed}\cdot h_{\alpha})$.  %
\end{itemize}
For each affected satellite cluster $\hat{S}\in\hat{\cal N}_{i,\alpha}$, add the following node.
\begin{itemize}
\item(Satellite Nodes) create a \emph{satellite node} $\hat{S}_{\nd}$ with node length $\hat{h}_{\ed,\alpha}\cdot s_{\ed} + 2h_{\diam,\alpha} + 4h_{\alpha} = \Theta(s^{2}_{\nc}\cdot s_{\ed}\cdot h_{\alpha})$. Set $\level(\hat{S}_{\nd}):= i$ and $\scale(\hat{S}_{\nd}):= \alpha$. 
\end{itemize}
The following \Cref{lemma:NumQNodes} bounds the number of nodes in the query graph $Q$. 
\begin{lemma}
\label{lemma:NumQNodes}
The number of nodes in $Q$ is at most 
\[
O((i_{\max}\cdot \alpha_{\max}\cdot \omega\cdot f + \Delta\cdot\gamma_{\tree}\cdot f)\cdot \rho_{\heavy}) = \Otilde(1)\cdot n^{O(1/k)}\cdot f^{2}.
\]
\end{lemma}
\begin{proof}
\Cref{lemma:NumAffectedCluster} shows that the number of affected (satellite) cluster is at most $\lambda:=i_{\max}\cdot\alpha_{\max}\cdot \omega\cdot f$.
The number of residual subtrees is at most $f\cdot \gamma_{\tree}\cdot \Delta + \lambda$ by the following reasons. First, removing one vertex from one Steiner tree can contribute $\Delta$ residual subtrees (due to the maximum degree $\Delta$) and a failed vertex can appear in at most $\gamma_{\tree}$ many Steiner trees. In addition, an affected cluster could have its Steiner tree disjoint from $F$, which contributes $1$ residual subtree. 

Finally, each light residual subtree can create at most $\rho_{\heavy}$ many singleton nodes. This gives the final bound.
\end{proof}

\paragraph{Mappings Between Vertices and Concrete Nodes.} We classify nodes into \emph{concrete nodes} and \emph{virtual nodes}: singleton nodes and subtree nodes are concrete, while cluster nodes and satellite nodes are virtual.

Each concrete node $q_{\nd}$ in $Q$ \emph{represents} a set of original vertices $V(q_{\nd})$ defined as follows. For each singleton node $v_{\nd}$, it naturally represents a single vertex $V(v_{\nd}):=\{v\}$. For each subtree node $\tau_{\nd}$, it represents $V(\tau_{\nd}):= V_{\tmn}(\tau)$, i.e., the terminals in this heavy residual subtree $\tau$.

Inversely, for each vertex $v\in V(G)$, we use $V_{\nd}(v)$ to denote the set of concrete nodes representing it. Namely,
\[
V_{\nd}(v):=\{q_{\nd}\mid v\in V(q_{\nd})\}.
\]
Note that a vertex $v$ may be represented by multiple singleton nodes and subtree nodes simultaneously. Moreover, each failed vertex $v\in F$ must have an empty $V_{\nd}(v)$.

\subsubsection{Links in the Query Graph} 
\label{sect:OracleQLink}

There are seven types of links, which can be divided into two categories (except for the last type): \emph{shortcut-based links} and \emph{expansion-based links}.

For better understanding, let us first explain the intuition behind the classification of shortcut-based and expansion-based links. Recall that we have created nodes to represent the ``affected regions'' (i.e., terminals of affected clusters) of $G$ under failures. However, a large portion of $G$ remains unaffected, and we want to add shortcut-based links to encode distances through these unaffected regions. 

More formally, the scenario of adding shortcut-based links consists of three clusters $S_{1}, S_{2}$ and $S$ ($S_{1}$ and $S_{2}$ can be the same). We are also provided the following two guarantees of $S_{1},S_{2}$ and $S$ which come from how we exploit shortcut-based links in the proof of \Cref{lemma:FindingProxy}. For clarity, let
\[
(i_{1},i_{2},i)=(\level(S_{1}),\level(S_{2}),\level(S)),\qquad (\alpha_{1},\alpha_{2},\alpha) = (\scale(S_{1}),\scale(S_{2}),\scale(S)).
\]
Then it is guaranteed that $i_{1}\geq i_{2}>i$ and $\alpha_{1}\geq \alpha_{2}\geq \alpha$. In other words, as we move from $S_{1}$ through $S_{2}$ to $S$ and the level decreases, the scale also goes down. Moreover, it is also guaranteed that $S_{1}$ intersects $S_{2}$. 

The most standard way to add shortcut-based links is that, for each pair of vertices $u\in V_{\tmn}(S_{1})\cap \Lambda(S)$ and $v\in V_{\tmn}(S_{2})\cap \Lambda(S)$, add a link between (i) node $u_{\nd}\in V_{\nd}(u)$ with $\cluster(u_{\nd})=S_{1}$ and (ii) node $v_{\nd}\in V_{\nd}(v)$ with $\cluster(v_{\nd})=S_{2}$. This means we want to add a biclique between (nodes of) terminals of $S_{1}$ and $S_{2}$ that are inside the surrounding set $\Lambda(S)$. Intuitively, this biclique can ``shortcut'' any path between $S_{1}$ and $S_{2}$ going through $S$.

However, this standard construction may create too many shortcut-based links. Later in \Cref{sect:OracleConstruction}, such links will be stored in a 2D-range counting table in the oracle, so the number of links directly affects space usage. Therefore, we need to perform sparsification, which will introduce the seven types of links below.

We classify the first three types as shortcut-based links, since they are all created with respect to some unaffected cluster $S$ (or an edge $e$ which can be viewed as a degenerated unaffected cluster) and two higher levels $i_{1},i_{2}$. Also, they are all stored inside the 2D-range counting table. In contrast, links of types \ref{Item:Type4Link}, \ref{Item:Type5Link} and \ref{Item:Type6Link} are classified as expansion-based. A distinguishing feature of these links is that they are not created with respect to any specific unaffected cluster $S$ and are not stored in the 2D range-counting table.

\paragraph{Shortcut-Based Links.} We first formally describe the shortcut-based links. Note that some shortcut-based links below are unnecessary, since we sometimes, for simplicity, do not require both linked levels to be higher than that of $S$. 

\medskip

\setcounter{Oracletypectr}{0}
\refstepcounter{Oracletypectr}\label{Item:Type1Link}
\noindent\underline{\textcolor{DarkRed}{Type \arabic{Oracletypectr}}.} The first-type links are added by the original edges. Note that we can view an original edge as a degenerated unaffected cluster. %

\begin{tcolorbox}[
  colback=black!3,
  colframe=black,
  boxrule=1.5pt,
  arc=0pt,
  left=6pt,
  right=6pt,
  top=4pt,
  bottom=4pt
]
For each original edge $e=(u,v)\in G$, add a link between each node pair
$(u_{\nd},v_{\nd})\in V_{\nd}(u)\times V_{\nd}(v)$ with link length
$\ell_{G}(e)$.
\end{tcolorbox}

\medskip

\refstepcounter{Oracletypectr}\label{Item:Type2Link}
\noindent\underline{\textcolor{DarkRed}{Type \arabic{Oracletypectr}}.} The second-type links are the standard shortcut-based links we discussed above, i.e., the bicliques. However, to perform sparsification, we this type of links are only between the levels $i$ at which the unaffected cluster $S$ is ``nearby $A_{i}$-sparse''. \Cref{def:nearbyDensenes} formally defines nearby denseness and sparseness of $S$.

\begin{definition}[Nearby Denseness]
\label{def:nearbyDensenes}
For a cluster $S\in {\cal N}_{i,\alpha}$ and a level $i'>i$, $S$ is \emph{nearby $A_{i'}$-dense} if $|\Lambda_{i'}(S)|\geq \rho_{\heavy}$, otherwise it is \emph{nearby $A_{i'}$-sparse}.
\end{definition}

According to the definition, the number of $A_{i}$-terminals in the surrounding set $\Lambda(S)$ is quite small when $S$ is nearby $A_{i}$-sparse, so we will not create too many second-type links.

\begin{tcolorbox}[
  colback=black!3,
  colframe=black,
  boxrule=1.5pt,
  arc=0pt,
  left=6pt,
  right=6pt,
  top=4pt,
  bottom=4pt
]
For each
\begin{itemize}
\item level $i$, length scale $\alpha$ and unaffected cluster $S\in{\cal N}_{i,\alpha}$, 
\item higher levels $i_{1},i_{2}>i$ such that $S$ is both nearby $A_{i_{1}}$- and $A_{i_{2}}$-sparse, 
\item vertex pair $(u,v)\in \Lambda_{i_{1}}(S)\times \Lambda_{i_{2}}(S)$,
\item node pair $(u_{\nd},v_{\nd})\in V_{\nd}(u)\times V_{\nd}(v)$,
\end{itemize}
add a link between $u_{\nd}$ and $v_{\nd}$ with link length $h_{\diam,\alpha} + 2h_{\alpha}$.
\end{tcolorbox}

\medskip

\refstepcounter{Oracletypectr}\label{Item:Type3Link}
\noindent\underline{\textcolor{DarkRed}{Type \arabic{Oracletypectr}}.} The third type is different from the second type in the following ways. First, the third-type shortcut-based links for an unaffected cluster $S$ are between a level $i_{1}$ at which $S$ is nearby $A_{i_{1}}$-sparse and a lower level $i_{2}<i_{1}$ at which $S$ is nearby $A_{i_{2}}$-dense. However, adding a biclique is too costly in this case, as the surrounding set $\Lambda(S)$ may contain many $A_{i_{2}}$-vertices. This is where the satellite clusters come into play. 

Roughly speaking, we will use the satellite clusters $\hat{S}:=\satellite_{i_{2},\alpha_{2}}(S)$ (for all $\alpha_{2}\geq \scale(S)$) as ``representatives'' of $A_{i_{2}}$-vertices in $\Lambda(S)$, meaning that we connect the shortcuts to the satellite nodes $\hat{S}_{\nd}$ of these $\hat{S}$. Later, the type-\ref{Item:Type6Link} links will connect $\hat{S}$ to $\Lambda_{i_{2}}(S)$. We note that in this case, the condition that $S$ is nearby $A_{i_{2}}$-dense is crucial, as the feasibility argument (\Cref{lemma:GeneralSolLowerBound}) relies on both nearby denseness and the LC-expansion property.

\begin{tcolorbox}[
  colback=black!3,
  colframe=black,
  boxrule=1.5pt,
  arc=0pt,
  left=6pt,
  right=6pt,
  top=4pt,
  bottom=4pt
]
For each
\begin{itemize}
\item level $i$, length scale $\alpha$ and unaffected cluster $S\in {\cal N}_{i,\alpha}$,
\item higher levels $i_{1},i_{2}$ s.t. $i_{1}>i_{2}>i$ and $S$ is nearby $A_{i_{1}}$-sparse but nearby $A_{i_{2}}$-dense,
\item vertex $u\in \Lambda_{i_{1}}(S)$,
\item node $u_{\nd}\in V_{\nd}(u)$,
\end{itemize}
add a link between $u_{\nd}$ and the satellite node $\hat{S}_{\nd}$ created for the satellite cluster $\hat{S}:=\satellite_{i_{2},\alpha_{2}}(S)$ for all $\alpha_{2}\geq \alpha$. The link length is $0$. 
\end{tcolorbox}

\paragraph{Expansion-Based Links.} Next, we will introduce expansion-based links. Such links are constructed based on different notions of heaviness and denseness, which essentially exploit the LC-expansion property in different ways.

\medskip

\refstepcounter{Oracletypectr}\label{Item:Type4Link}
\noindent\underline{\textcolor{DarkRed}{Type \arabic{Oracletypectr}}.} The fourth type is based on the heaviness of residual subtrees (see \Cref{def:Heaviness}). Intuitively, all heavy residual subtrees of the same cluster are close to each other, so we can safely link them to the cluster node. 

\begin{tcolorbox}[
  colback=black!3,
  colframe=black,
  boxrule=1.5pt,
  arc=0pt,
  left=6pt,
  right=6pt,
  top=4pt,
  bottom=4pt
]
For each
\begin{itemize}
\item level $i_{1}$, length scale $\alpha_{1}$, and affected cluster $S_{1}\in{\cal N}_{i_{1},\alpha_{1}}$, 
\item
subtree node $\tau_{\nd}$ with $\cluster(\tau_{\nd}) = S_{1}$,
\end{itemize}
add a link between $\tau_{\nd}$ and the cluster node $S_{1,\nd}$ created for $S_{1}$ with link length $0$.
\end{tcolorbox}

\medskip

\refstepcounter{Oracletypectr}\label{Item:Type5Link}
\noindent\underline{\textcolor{DarkRed}{Type \arabic{Oracletypectr}}.}  
The fifth type relies on the notion of adjacent-cluster denseness in \Cref{def:OracleAdjClusterDenseness}. Roughly speaking, a level-$i_{2}$ node $q_{\nd}$ can obtain $A_{i_{1}}$-denseness of a high level $i_{1}\geq i_{2}$ if $V(q_{\nd})$ is ``adjacent'' to a lower-level nearby $A_{i_{1}}$-dense unaffected cluster.

\begin{definition}[Adjacent-Cluster Denseness]
\label{def:OracleAdjClusterDenseness}
Consider a level $i_{2}$, length scale $\alpha_{2}$, a cluster $S_{2}\in{\cal N}_{i_{2},\alpha_{2}}$, and a concrete node $q_{\nd}$ with $\cluster(q_{\nd}) = S_{2}$. For each level $i_{1}\geq i_{2}$, we say $q_{\nd}$ is \emph{adjacent-cluster $A_{i_{1}}$-dense} if 
\begin{itemize}
\item there exist level $i<i_{2}$, length scale $\alpha\leq \alpha_{2}$ and an unaffected cluster $S\in{\cal N}_{i,\alpha}$ such that $S$ is nearby $A_{i_{1}}$-dense and $V(q_{\nd})$ intersects $\Lambda(S)$.
\end{itemize}
Otherwise $q_{\nd}$ is \emph{adjacent-cluster $A_{i}$-sparse}. %
\end{definition}

\begin{tcolorbox}[
  colback=black!3,
  colframe=black,
  boxrule=1.5pt,
  arc=0pt,
  left=6pt,
  right=6pt,
  top=4pt,
  bottom=4pt
]
For each 
\begin{itemize}
\item level $i_{2}$, length scale $\alpha_{2}$, affected cluster $S_{2}\in {\cal N}_{i_{2},\alpha_{2}}$, and concrete node $q_{\nd}$ with $\cluster(q_{\nd}) = S_{2}$, 
\item level $i_{1}\geq i_{2}$, length scale $\alpha_{1}\geq \alpha_{2}$,
\end{itemize}
if $q_{\nd}$ is adjacent-cluster $A_{i_{1}}$-dense,
then for each 
\begin{itemize}
\item affected cluster $S_{1}\in{\cal N}_{i_{1},\alpha_{1}}$ such that $S_{1}$ intersects $S_{2}$,
\end{itemize}
add a link between $q_{\nd}$ and the cluster node $S_{1,\nd}$ created for $S_{1}$. The link length is $0$.
\end{tcolorbox}
\medskip

\refstepcounter{Oracletypectr}\label{Item:Type6Link}
\noindent\underline{\textcolor{DarkRed}{Type \arabic{Oracletypectr}}.} The sixth type is also based on adjacent-cluster denseness. The difference is that, this time we only exploit the same-level adjacent-cluster denseness of a node $q_{\nd}$, and then link $q_{\nd}$ to all the affected satellite clusters that cover the cluster of $q_{\nd}$.

\begin{tcolorbox}[
  colback=black!3,
  colframe=black,
  boxrule=1.5pt,
  arc=0pt,
  left=6pt,
  right=6pt,
  top=4pt,
  bottom=4pt
]
For each 
\begin{itemize}
\item level $i_{2}$, length scale $\alpha_{2}$, affected cluster $S_{2}\in{\cal N}_{i_{2},\alpha_{2}}$, and concrete node $q_{\nd}$ with $\cluster(q_{\nd}) = S_{2}$, 
\end{itemize}
if $q_{\nd}$ is adjacent-cluster $A_{i_{2}}$-dense, then for each
\begin{itemize}
\item affected satellite cluster $\hat{S}_{2}\in\hat{\cal N}_{i_{2},\alpha_{2}}$ such that $\hat{S}_{2}\ni \pivot(S_{2})$,
\end{itemize}
add a link between $q_{\nd}$ and the satellite node $\hat{S}_{2,\nd}$ created for $\hat{S}_{2}$. The link length is $10\cdot \hat{h}_{\ed,\alpha_{2}}\cdot s_{\ed}$.
\end{tcolorbox}

\paragraph{Links Based on Common Vertices.} \ 

\refstepcounter{Oracletypectr}\label{Item:Type7Link}
\noindent\underline{\textcolor{DarkRed}{Type \arabic{Oracletypectr}}.} The last type of links simply connects nodes that share a common vertex. We note that links of this type are rarely used.
\begin{tcolorbox}[
  colback=black!3,
  colframe=black,
  boxrule=1.5pt,
  arc=0pt,
  left=6pt,
  right=6pt,
  top=4pt,
  bottom=4pt
]
For each pair of concrete nodes $q_{\nd}$ and $q'_{\nd}$, if $V(q_{\nd})$ intersects $V(q_{\nd})$, then add a link between $q_{\nd}$ and $q'_{\nd}$ with link length 0.
\end{tcolorbox}

\subsubsection{Lower Bounding Distances in the Query Graph}
\label{sect:OracleLowerBound}

Now we prove \Cref{lemma:GeneralSolLowerBound} for establishing correctness of our (later) query algorithm on the lower-bound side. Roughly speaking, the query graph $Q$ constructed above never underestimates the real distances after failures.

Before proving \Cref{lemma:GeneralSolLowerBound}, we need the following \Cref{lemma:LCExpanding} which captures the key property of weak LC vertex expansion that we will exploit.

\begin{lemma}[LC-Expansion Property]
\label{lemma:LCExpanding}
Consider a weak expanding graph $G_{i,\alpha}$.
For any two vertex subsets $T_{1},T_{2}\subseteq A_{i}$, if
\begin{itemize}
\item $|T_{1}|,|T_{2}|\geq \rho_{\heavy}$,
\item $\diam_{G_{i,\alpha}}(T_{1},T_{2})\leq h_{\ed,\alpha}$,%
\end{itemize}
then for any vertex failures $F\subseteq V$ of size at most $f$, we have
\[
\dist_{G\setminus F}(T_{1}\setminus F,T_{2}\setminus F)\leq h_{\ed,\alpha}\cdot s_{\ed}.
\]
\end{lemma}
\begin{proof}
Recall from \Cref{ob:WeakExpanding} that $A_{i}$ is weakly $(h_{\ed,\alpha},s_{\ed})$-length $\phi$-vertex-expanding w.r.t. $G_{i,\alpha}$ and $G$, meaning that for any $A_{i}$-respecting demand $D$ that is $h_{\ed,\alpha}$-length in $G_{i,\alpha}$, it can be routed in $G$ with length $h_{\ed,\alpha}\cdot s_{\ed}$ and vertex congestion $\gamma = O(\log n/\phi)$.

Now we define this demand $D$ to be an arbitrary matching between $\rho_{\heavy}$ many $T_{1}$-vertices and $\rho_{\heavy}$ many $T_{2}$-vertices. Recall that we set $\rho_{\heavy} = \gamma\cdot f+1$. The routing of $D$ in $G$ has vertex congestion $\gamma$, so removing $f$ vertices can destroy at most $\gamma\cdot f$ units of flows. Therefore, there are still non-zero $(h_{\ed,\alpha}\cdot s_{\ed})$-length flows between $T_{1}$ and $T_{2}$ in $G\setminus F$. This completes the proof.
\end{proof}

\begin{lemma}
\label{lemma:GeneralSolLowerBound}
Let $x_{\nd}$ and $y_{\nd}$ be two arbitrary concrete nodes in $Q_{}$. For each vertices $x\in V(x_{\nd})$ and $y\in V(y_{\nd})$, we have
\[
\dist_{G\setminus F}(x,y)\leq \dist_{Q_{}}(x_{\nd},y_{\nd}).
\]
\end{lemma}
\begin{proof}
Let $P$ be a shortest path in $Q$ between $x_{\nd}$ and $y_{\nd}$. We decompose $P$ into maximal subpaths that are internally disjoint from concrete nodes. Let $P_{\sub}$ be such a subpath. By definition, $P_{\sub}$'s endpoints, denoted by $u_{\nd}$ and $v_{\nd}$, are both concrete nodes, and all middle nodes of $P_{\sub}$ are virtual nodes. To prove the lemma, it suffices to show that
\begin{equation*}
\dist_{G\setminus F}(V(u_{\nd}),V(v_{\nd}))\leq \leng_{Q}(P) - \ell(u_{\nd}) - \ell(v_{\nd}),
\end{equation*}
since for each concrete node $v_{\nd}$, $\diam_{G\setminus F}(V(v_{\nd}))\leq \ell(v_{\nd})$.

By our construction of $Q$, there is no link directly connecting two virtual nodes. Hence, $P_{\sub}$ contains exactly one link or two links. 

\medskip

\noindent{\underline{Case 1.}} If $P_{\sub}$ is a single link $e_{\lk}:=(u_{\nd},v_{\nd})$, observe that $\leng_{Q}(P)-\ell(u_{\nd}) - \ell(v_{\nd})$ is exactly $\ell(e_{\lk})$, so we only need to show $\dist_{G\setminus F}(V(u_{\nd}),V(v_{\nd}))\leq \ell(e_{\lk})$. There are several subcases.
\begin{itemize}
\item Suppose $e_{\lk}$ has type \ref{Item:Type1Link}. By the rule of adding type-\ref{Item:Type1Link} link, there exists an edge $e=(u,v)$ such that $u\in V(u_{\nd})$, $v\in V(v_{\nd})$ and $\ell_{G}(e) = \ell_{}(e_{\lk})$, which implies $\dist_{G\setminus F}(V(u_{\nd}),V(v_{\nd}))\leq \ell_{}(e_{\lk})$.
\item Suppose $e_{\lk}$ has type \ref{Item:Type2Link}. By the rule, there exist vertices $u\in V(u_{\nd})$, $v\in V(v_{\nd})$ and an unaffected cluster $S\in{\cal N}_{i,\alpha}$ such that $u,v\in \Lambda(S)$. Since $S$ is unaffected, by \Cref{ob:SurroundingDiam} we have
\[
\dist_{G\setminus F}(V(u_{\nd}),V(v_{\nd}))\leq \diam_{G\setminus F}(\Lambda(S))\leq h_{\diam,\alpha} + 2h_{\alpha} = \ell_{}(e_{\lk}).
\]
\item The case where $e_{\lk}$ has type \ref{Item:Type7Link} is trivial.
\end{itemize}

\noindent{\underline{Case 2.}} Suppose $P_{\sub}$ has two links $e_{\lk,1}:= (u_{\nd}, S_{1,\nd})$ and $e_{\lk,2}:=(S_{1,\nd}, v_{\nd})$ with a cluster node $S_{1,\nd}$ in the middle. It suffices to show $\dist_{G\setminus F}(V(u_{\nd}),V(v_{\nd}))\leq \ell(S_{1,\nd})$. Observe that only links with types \ref{Item:Type4Link} and \ref{Item:Type5Link} can be adjacent to the cluster node $S_{1,\nd}$. Let $S_{1} = \cluster(S_{1,\nd}),i_{1}=\level(S_{1,\nd})$ and $\alpha_{1} = \scale(S_{1,\nd})$.

For better understanding and to avoid tedious case analysis, let us first discuss our strategy of exploiting the LC-expansion property \Cref{lemma:LCExpanding} w.r.t. $A_{i_{1}}$ in $G_{i_{1},\alpha_{1}}$. Take the expansion-based link $(u_{\nd},S_{1,\nd})$ as an example. Regardless it has which type, at the end we want to find a vertex set $V_{u}$ such that
\begin{enumerate}[label=(\alph*)]
\item $V_{u}\subseteq A_{i_{1}}$ and $|V_{u}|\geq \rho_{\heavy}$.
\item each vertex $u'\in V_{u}$ is close to $S$ \underline{in $G_{i_{1},\alpha_{1}}$}, or formally, $\dist_{G_{i_{1},\alpha_{1}}}(u',S)\leq \theta$ for some $\theta$ such that $2\theta + h_{\diam,\alpha_{1}}\leq h_{\ed,\alpha_{1}}$ %
\item each vertex $u'\in V_{u}\setminus F$ is close to $V(u_{\nd})$ \underline{in $G\setminus F$}, or formally, $\dist_{G\setminus F}(u',V(u_{\nd}))\leq \eta$ for some $\eta$ such that $2\eta + h_{\ed,\alpha_{1}}\cdot s_{\ed}\leq \ell(S_{1,\nd})$. %
\end{enumerate}
We first argue that why finding such a set $V_{u}$ (together with such a set $V_{v}$ for the ``$v$-side'') will upper bound $\dist_{G\setminus F}(V(u_{\nd}),V(v_{\nd}))$. First, we have 
\[
\diam_{G_{i_{1},\alpha_{1}}}(V_{u},V_{v})\leq 2\theta + \diam_{G_{i_{1},\alpha_{1}}}(S_{1})\leq 2\theta + h_{\diam,\alpha_{1}}\leq h_{\ed,\alpha_{1}}.
\]
Therefore, %
we can apply \Cref{lemma:LCExpanding} on graph $G_{i_{1},\alpha_{1}}$, node-weighting $A_{i_{1}}$, and the two sets $V_{u},V_{v}$, which gives
\[
\dist_{G\setminus F}(V_{u}\setminus F,V_{v}\setminus F)\leq h_{\ed,\alpha_{1}}\cdot s_{\ed}.
\]
Finally, using (c), we have 
\[
\dist_{G\setminus F}(V(u_{\nd}),V(v_{\nd}))\leq 2\eta + h_{\ed,\alpha_{1}}\cdot s_{\ed}\leq \ell(S_{1,\nd}),
\]
and this completes the proof.

Now we define $V_{u}$, $\theta$ and $\eta$ for $e_{\lk,1}=(u_{\nd},S_{1,\nd})$ depending on whether $e_{\lk,1}$ has type \ref{Item:Type4Link} or type \ref{Item:Type5Link}, and show that they satisfies all requirements (a), (b) and (c). %
For the link $e_{\lk,2}=(v_{\nd},S_{1,\nd})$ on the $v$-side, the argument is symmetric and thus omitted.

\medskip

\noindent\underline{Subcase: $e_{\lk,1}$ has type \ref{Item:Type4Link}.} Select $V_{u} = V(u_{\nd})$, $\theta = 0$ and $\eta = 0$. Note that $|V_{u}|\geq \rho_{\heavy}$ because $u_{\nd}$ is a subtree node created for a heavy residual subtree of $S_{1}$, so (a) is satisfied. Requirements (b) and (c) are obviously satisfied.

\medskip

\noindent\underline{Subcase: $e_{\lk,1}$ has type \ref{Item:Type5Link}.} Let $(i_{2},\alpha_{2},S_{2}) = (\level(u_{\nd}),\scale(u_{\nd}),\cluster(u_{\nd}))$. Before selecting $V_{u}$, recall that a type-\ref{Item:Type5Link} link guarantees that $V(u_{\nd})$ is adjacent-cluster $A_{i_{1}}$-dense and that $S_{2}$ intersects $S_{1}$. The adjacent-cluster $A_{i_{1}}$-denseness means there is a cluster $S\in{\cal N}_{i,\alpha}$ at level $i<i_{2}$ and length scale $\alpha<\alpha_{2}$ such that $\Lambda(S)$ intersects $V(u_{\nd})$ and $S$ is nearby $A_{i_{1}}$-dense. Select 
\[
V_{u} = \Lambda_{i_{1}}(S),\qquad \theta = 2h_{\diam,\alpha_{1}} + 4h_{\alpha_{1}},\qquad \eta = h_{\diam,\alpha_{1}} + 2h_{\alpha_{1}}.
\] 
Now we show that our choices satisfy all (a), (b) and (c). First, (a) is satisfied since $S$ is nearby $A_{i_{1}}$-dense. For (b), observe that for each $u'\in V_{u}$
\begin{align*}
&~~~~\dist_{G_{i_{1},\alpha_{1}}}(u',S_{1})\\
&\leq \diam_{G_{i_{1},\alpha_{1}}}(\Lambda_{\leq i_{1}}(S)) + \diam_{G_{i_{1},\alpha_{1}}}(S_{2})&&\text{(by $\Lambda(S)$ intersects $V(u_{\nd})\subseteq S_{2}$, and}\\& &&\text{\ \ \ \ \ $S_{2}$ intersects $S_{1}$)}\\
&\leq (h_{\diam,\alpha_{1}} + 4h_{\alpha_{1}})+ h_{\diam,\alpha_{1}}&&\text{(by \Cref{ob:SurroundingDiam})}\\
&\leq 2h_{\diam,\alpha_{1}} + 4h_{\alpha_{1}}\\
&=\theta,
\end{align*}
and also $2\theta + h_{\diam,\alpha_{1}}\leq h_{\ed,\alpha_{1}}$ by definition. For (c), observe that for each $u'\in V_{u}\setminus F\subseteq \Lambda(S)\setminus F$,
\begin{align*}
&~~~~\dist_{G\setminus F}(u',V(u_{\nd}))\\
&\leq \diam_{G\setminus F}(\Lambda(S)\setminus F)&&\text{(by $\Lambda(S)\setminus F$ intersects $V(u_{\nd})$)}\\
&\leq h_{\diam,\alpha_{1}} + 2h_{\alpha_{1}}&&\text{(by \Cref{ob:SurroundingDiam} and $S$ is unaffected)}\\
&=\eta,
\end{align*}
and also $2\eta + h_{\ed,\alpha_{1}}\cdot s_{\ed}\leq \ell(S_{1,\nd})$ by definition.

\medskip

\noindent\underline{Case 3.} Suppose $P_{\sub}$ has two links $e_{\lk,1}:=(u_{\nd},\hat{S}_{2,\nd})$ and $e_{\lk,2}:=(\hat{S}_{2,\nd},v_{\nd})$ with a satellite node $\hat{S}_{2,\nd}$ in the middle. It suffices to show $\dist_{G\setminus F}(V(u_{\nd}),V(v_{\nd}))\leq \ell_{}(\hat{S}_{2,\nd})$. Only links with types \ref{Item:Type3Link} and \ref{Item:Type6Link} can be adjacent to the satellite node $\hat{S}_{2,\nd}$. Let $i_{2}=\level(\hat{S}_{2,\nd}),\alpha_{2} = \scale(\hat{S}_{2,\nd})$ and $\hat{S}_{2}\in \hat{\cal N}_{i_{2},\alpha_{2}}$ be the satellite cluster for which $\hat{S}_{2,\nd}$ is created.

Similar to Case 2, we want to exploit \Cref{lemma:LCExpanding} w.r.t. $A_{i_{2}}$ in $G_{i_{2},\alpha_{2}}$. Again it suffices to find a vertex set $V_{u}$ such that
\begin{enumerate}[label=(\alph*)]
\item $V_{u}\subseteq A_{i_{2}}$ and $|V_{u}|\geq \rho_{\heavy}$.
\item each vertex $u'\in V_{u}$ is close to $\hat{S}_{2}$ \underline{in $G_{i_{2},\alpha_{2}}$}, or formally, $\dist_{G_{i_{2},\alpha_{2}}}(u',\hat{S}_{2})\leq \theta$ for some $\theta$ such that $2\theta + \hat{h}_{\diam,\alpha_{2}}\leq \hat{h}_{\ed,\alpha_{2}}$.
\item each vertex $u'\in V_{u}\setminus F$ is close to $V(u_{\nd})$ \underline{in $G\setminus F$}, or formally, $\dist_{G\setminus F}(u',V(u_{\nd}))\leq \eta$ for some $\eta$ such that $2\eta + \hat{h}_{\ed,\alpha_{2}}\cdot s_{\ed}\leq \ell(\hat{S}_{2,\nd})$.
\end{enumerate}
We emphasize that, this time the restrictions on $\theta$ and $\eta$ are slightly different.

\medskip

\noindent\underline{Subcase: $e_{\lk,1}$ has type \ref{Item:Type3Link}.} This type-\ref{Item:Type3Link} link guarantees that there exist level $i<i_{2}$, length scale $\alpha\leq \alpha_{2}$ and an unaffected cluster $S\in \mathcal{N}_{i,\alpha}$ such that $\Lambda(S)$ intersects $V(u_{\nd})$, $S$ is nearby $A_{i_{2}}$-dense, and $\hat{S}_{2} = \satellite_{i_{2},\alpha_{2}}(S)$. Select
\[
V_{u} = \Lambda_{i_{2}}(S),\qquad \theta = 0,\qquad \eta = h_{\diam,\alpha_{2}} + 2h_{\alpha_{2}}.
\]
We have (a) since $S$ is nearby $A_{i_{2}}$-dense. Next, we have (b) follows from $V_{u}\subseteq \hat{S}_{2}$ (by \Cref{ob:SatelliteCover}, i.e., $\hat{S}_{2}$ covers the $A_{i_{2}}$-vertices in the surrounding set $\Lambda(S)$). Finally, (c) follows from the fact that $V(u_{\nd})$ intersects $\Lambda(S)$, together with \Cref{ob:SurroundingDiam} (i.e., the surrounding set $\Lambda(S)$ has low diameter). %

\medskip

\noindent\underline{Subcase: $e_{\lk,1}$ has type \ref{Item:Type6Link}.} This type-\ref{Item:Type6Link} link guarantees that 
\begin{itemize}
\item $\level(u_{\nd}) = i_{2}$, $\scale(u_{\nd}) = \alpha_{2}$, and $\hat{S}_{2}\ni\pivot(S_{2})$ for $S_{2} = \cluster(u_{\nd})$, which implies $\hat{S}_{2}$ intersects $S_{2}$.
\item there exist level $i<i_{2}$, length scale $\alpha\leq \alpha_{2}$ and an unaffected cluster $S$ such that $\Lambda(S)$ intersects $V(u_{\nd})$ and $S$ is nearby $A_{i_{2}}$-dense.
\end{itemize}
Therefore, select
\[
V_{u} = \Lambda_{i_{2}}(S),\qquad \theta = 2h_{\diam,\alpha_{2}} + 4h_{\alpha_{2}},\qquad \eta = h_{\diam,\alpha_{2}} + 2h_{\alpha_{2}}.
\]
Again, (a) follows that $S$ is nearby $A_{i_{2}}$-dense. For (b), for each $u'\in V_{u}$,
\begin{align*}
&~~~~\dist_{G_{i_{2},\alpha_{2}}}(u',\hat{S}_{2})\\
&\leq \diam_{G_{i_{2},\alpha_{2}}}(\Lambda_{\leq i_{2}}(S)) + \diam_{G_{i_{2},\alpha_{2}}}(S_{2})&&\text{(by $\Lambda(S)$ intersects $V(u_{\nd})\subseteq S_{2}$, and}\\& &&\text{\ \ \ \ \ $S_{2}$ intersects $\hat{S}_{2}$)}\\
&\leq (h_{\diam,\alpha_{2}} + 4h_{\alpha_{2}})+ h_{\diam,\alpha_{2}}&&\text{(by \Cref{ob:SurroundingDiam})}\\
&\leq 2h_{\diam,\alpha_{2}} + 4h_{\alpha_{2}}\\
&=\theta.
\end{align*}
For (c), observe $u'\in V_{u}\setminus F\subseteq \Lambda(S)\setminus F$,
\begin{align*}
&~~~~\dist_{G\setminus F}(u',V(u_{\nd}))\\
&\leq \diam_{G\setminus F}(\Lambda(S)\setminus F)&&\text{(by $\Lambda(S)\setminus F$ intersects $V(u_{\nd})$)}\\
&\leq h_{\diam,\alpha_{2}} + 2h_{\alpha_{2}}&&\text{(by \Cref{ob:SurroundingDiam} and $S$ is unaffected)}\\
&=\eta.
\end{align*}

\end{proof}

\subsection{The Query Algorithm}
\label{sect:OracleQueryAlgo}

We first present the query algorithm, assuming that we know the explicit query graph $Q$ and the all-pair distances in $Q$.

\begin{algorithm}[htbp]
\caption{The Query Algorithm: \textsc{Query}$(\x,\y,F)$}
\label{algo:QueryAlgorithm}
\begin{algorithmic}[1]
\State Initialize $\ans = \infty$.
\Statex
\For{$i\in[0,i_{\max}],\alpha\in[1,\alpha_{\max}]$}
\If{(${\cal N}_{i,\alpha}(x)$ is unaffected and $y\in {\cal N}_{i,\alpha}(x)$) or (${\cal N}_{i,\alpha}(y)$ is unaffected and $x\in {\cal N}_{i,\alpha}(y)$)}
\State $\ans\gets \min\{\ans,h_{\diam,\alpha}\}$
\EndIf
\EndFor

\Statex

\For{$i_{x},i'_{x}\in[\level(x),i_{\max}]$, $i_{y},i'_{y}\in[\level(y),i_{\max}],\alpha_{x},\alpha_{y},\alpha'_{x},\alpha'_{y}\in[1,\alpha_{\max}]$}

\Statex

\If{${\cal N}_{i_{x},\alpha_{x}}(x)$ is affected}
\State Let $x'_{\px,\nd}$ be the node in $V_{\nd}(x)$ with $\cluster(x'_{\px,\nd}) = {\cal N}_{i_{x},\alpha_{x}}(x)$. \Comment{Pick $x'_{\px} = x$}
\Else
\If{$i'_{x}>i_{x}$, $\alpha'_{x}\geq\alpha_{x}$, ${\cal N}_{i'_{x},\alpha'_{x}}(x)$ is affected}
\State Pick an $x'_{\px}\in \ProxyList(x,i_{x},\alpha_{x},i'_{x},\alpha'_{x})\setminus F$. \Comment{If unsuccessful, skip this iteration}
\State Take the node $x'_{\px,\nd}\in V_{\nd}(x'_{\px,\nd})$ with $\cluster(x'_{\px,\nd})={\cal N}_{i'_{x},\alpha'_{x}}(x)$.
\EndIf
\EndIf

\Statex

\If{${\cal N}_{i_{y},\alpha_{y}}(y)$ is affected}
\State Let $y'_{\px,\nd}$ be the node in $V_{\nd}(y)$ with $\cluster(y'_{\px,\nd}) = {\cal N}_{i_{y},\alpha_{y}}(y)$. \Comment{Pick $y'_{\px} = y$}
\Else
\If{$i'_{y}\geq i_{y}$, $\alpha'_{y}\geq \alpha_{y}$, ${\cal N}_{i'_{y},\alpha'_{y}}(y)$ is affected}
\State Pick a $y'_{\px}\in\ProxyVertex(y,i_{y},\alpha_{y},i'_{y},\alpha'_{y})\setminus F$. \Comment{If unsuccessful, skip this iteration}
\State Take the node $y'_{\px,\nd}\in V_{\nd}(y'_{\px,\nd})$ with $\cluster(y'_{\px,\nd})={\cal N}_{i'_{y},\alpha'_{y}}(y)$.

\EndIf
\EndIf

\Statex

\State If both $x'_{\px,\nd}$ and $y'_{\px,\nd}$ are picked successfully,
\Statex ~~~~~~~~~~~~~~$\ans\gets\min\{\ans,2h_{\diam,\alpha_{x}} + 2h_{\diam,\alpha_{y}} + \dist_{Q_{}}(x'_{\px,\nd},y'_{\px,\nd})\}$.

\Statex

\EndFor
\State Return $\ans$.

\end{algorithmic}
\end{algorithm}

\paragraph{The Proxy Lists.}  The query algorithm also requires the following additional data structures called the \emph{proxy lists}, which should be precomputed and stored in the construction phase.

For each vertex $x$, each levels $i_{x},i'_{x}$ s.t. $\level(x)\leq i_{x}<i'_{x}\leq i_{\max}$, and each length scale $\alpha_{x},\alpha'_{x}$ s.t. $\alpha_{x}\leq \alpha'_{x}\leq \alpha_{\max}$, we will store a set of at most $f+1$ vertices, denoted by $\ProxyList(x, i_{x}, \alpha_{x}, i'_{x}, \alpha'_{x})$, chosen from 
\[
V_{\tmn}({\cal N}_{i'_{x},\alpha'_{x}}(x))\cap \Lambda({\cal N}_{i_{x},\alpha_{x}}(x)).
\]
Recall that ${\cal N}_{i_{x},\alpha_{x}}(x)$ is the cluster covering the neighborhood of $x$ identified by \Cref{thm:ConstructNeighborhoodCover}.
If the intersection has size at most $f+1$, then $\ProxyList(x, i_{x}, \alpha_{x}, i'_{x}, \alpha'_{x})$ is exactly this set. Otherwise, it is an arbitrary subset of size $f+1$.

\paragraph{The Query Algorithm.} Consider a query between vertices $\x$ and $\y$ under vertex failures $F$. Fix an arbitrary simple shortest path $P$, called the \emph{witness path}, connecting $x$ and $y$ in $G\setminus F$. This means $\leng_{G}(P) = \dist_{G\setminus F}(x,y)$. For clarity, we first include the pseudocode of the query algorithm in \Cref{algo:QueryAlgorithm}.

From a very high level, the intuition of \Cref{algo:QueryAlgorithm} is as follows. Just for analysis, we let $\alpha^{\star}$ be the unique length scale such that
\[
\leng_{G}(P)\leq \alpha^{\star}<2\cdot\leng_{G}(P).
\]
The first For loop (Lines 2 to 6) is checking the corner case (i.e. the following \Cref{lemma:Unaffectedxy}) that $P$ is fully contained in some unaffected cluster with length scale $\alpha^{\star}$. If so, the diameter of that cluster is already a good solution. Otherwise, the second For loop (lines 7 to 19) will find two nodes in $Q_{}$ that are ``close'' to the vertices $x$ and $y$, and then output the precomputed distance between these two nodes in $Q_{}$ (with some additive adjustment).

In fact, the second loop corresponds to \Cref{lemma:FindingProxy} right below, where guarantees 1(a), 1(b), 2(a), and 2(b) correspond one-to-one to the branches in the pseudocode. This \Cref{lemma:FindingProxy} may be difficult to digest at first glance. However, we note that part of the complexity in the lemma statement and its proof arises from separating the update phase from the query phase to obtain faster query time. Indeed, if the query vertices $x$ and $y$ were already known in the update phase, and we defined the affected clusters to be all clusters hit by $F\cup\{x,y\}$ rather than just $F$, then the simpler guarantees 1(a) and 2(a) would always hold. 

Finally, we note that the nodes $x'_{\px,\nd}$ and $y'_{\px,\nd}$ guessed by the algorithm may not be the same as the nodes $x_{\px,\nd}$ and $y_{\px,\nd}$ selected by \Cref{lemma:FindingProxy} (this only happens in cases 1(b) and 2(b)). Therefore, we need an additional \Cref{lemma:UpperBoundSolution} which finally establishes the correctness of upper-bound direction by showing $x_{\px,\nd}$ and $x'_{\px,\nd}$ (and also $y_{\px,\nd}$ and $y'_{\px,\nd}$) are close in the query graph. 

The correctness of the lower-bound direction is straightforward given \Cref{lemma:GeneralSolLowerBound}. We deferred the proof to  \Cref{sect:OracleSolLowerBound}.

\subsubsection{Upper Bounding the Solution}

\begin{lemma}
\label{lemma:Unaffectedxy}
If there exist either
\begin{itemize}
\item level $i_{x}$ such that ${\cal N}_{i_{x},\alpha^{\star}}$ is unaffected and $y\in {\cal N}_{i_{x},\alpha^{\star}}(x)$, or
\item level $i_{y}$ such that ${\cal N}_{i_{y},\alpha_{\star}}$ is unaffected and $x\in {\cal N}_{i_{y},\alpha^{\star}}(y)$,
\end{itemize}
then $\ans\leq 2\cdot s_{\nc}\cdot \dist_{G\setminus F}(x, y)$.
\end{lemma}
\begin{proof}
Suppose $i_{x}$ exists (the case in which $i_{y}$ exists is analogous). Then $\dist_{G_{i_{x},\alpha^{\star}}\setminus F}(x,y)\leq h_{\diam,\alpha^{\star}}$, since the cluster ${\cal N}_{i_{x},\alpha^{\star}}(x)$ includes both $x$ and $y$ and has no failed vertex. Therefore, $\ans\leq h_{\diam,\alpha^{\star}}\leq s_{\nc}\cdot h_{\alpha^{\star}}\leq 2\cdot s_{\nc}\cdot \dist_{G\setminus F}(x,y)$.
\end{proof}

\begin{assumption}
\label{assume:HighProbability}
Assume the condition of \Cref{lemma:Unaffectedxy} is not fulfilled. Also, assume the high probability event in \Cref{lemma:HighProbabilityEvent} happens.
\end{assumption}

\begin{restatable}{lemma}{FindingProxy}
\label{lemma:FindingProxy}
Under \Cref{assume:HighProbability}, there exist 
\begin{itemize}
\item two vertices $\x_{\px}$ and $\y_{\px}$ on $P$ called proxy vertices,
\item levels $i_{x}\geq \level(x_{\px}),i_{y}\geq \level(y_{\px})$, length scales $\alpha_{x},\alpha_{y}\leq \alpha^{\star}$, and clusters $S_{x}:={\cal N}_{i_{x},\alpha_{x}}(x)$, $S_{y}:={\cal N}_{i_{y},\alpha_{y}}(y)$,
\item  two nodes $\x_{\px,\nd}\in V_{\nd}(x_{\px})$ and $\y_{\px,\nd}\in V_{\nd}(y_{\px})$ called proxy nodes
\end{itemize}
further satisfying the following.
\begin{enumerate}
\item For $x_{\px}$ and $x_{\px,\nd}$, either
\begin{enumerate}[label=(\alph*)]
\item $x_{\px} = x$, $S_{x}$ is affected, and $\cluster(x_{\px,\nd})$ is exactly $S_{x}$, or
\item $x_{\px}\in \Lambda(S_{x})$, $S_{x}$ is unaffected, and $\cluster(x_{\px,\nd})$ is exactly the affected cluster $S'_{x}:={\cal N}_{i'_{x},\alpha'_{x}}(x)$ for some $i'_{x}>i_{x}$ and $\alpha'_{x}\in[\alpha'_{x},\alpha^{\star}]$.
\end{enumerate}

\item Similarly, for $y_{\px}$ and $y_{\px,\nd}$, either
\begin{enumerate}[label=(\alph*)]
\item $y_{\px} = y$, $S_{y}$ is affected, and $\cluster(y_{\px,\nd})$ is exactly $S_{y}$, or
\item $y_{\px}\in \Lambda(S_{y})$, $S_{y}$ is unaffected, and $\cluster(y_{\px,\nd})$ is exactly the affected cluster $S'_{y}:={\cal N}_{i'_{y},\alpha'_{y}}(y)$ for some $i'_{y}>i_{y}$ and $\alpha'_{y}\in[\alpha'_{y},\alpha^{\star}]$.
\end{enumerate}

\item $\dist_{Q_{}}(x_{\px,\nd},y_{\px,\nd})\leq O((s_{\tree}+s_{\nc})\cdot s_{\nc}\cdot s_{\ed}\cdot i_{\max})\cdot \dist_{G}(x,y)$.

\end{enumerate}
\end{restatable}

\begin{restatable}{lemma}{UpperBoundSolution}
\label{lemma:UpperBoundSolution}
Under \Cref{assume:HighProbability}, there exist levels $i_{x},i'_{x},i_{y},i'_{y}$ and length scales $\alpha_{x},\alpha'_{x},\alpha_{y},\alpha'_{y}$ s.t. $\level(x)\leq i_{x}<i'_{x}\leq i_{\max}$, $\level(y)\leq i_{y}<i'_{y}\leq i_{\max}$, $1\leq \alpha_{x}\leq \alpha'_{x}\leq \alpha^{\star}$, and $1\leq \alpha_{y},\alpha'_{y}\leq \alpha^{\star}$, so that on line 24, $\ans$ will be updated, and
\[
2h_{\diam,\alpha_{x}} + 2h_{\diam,\alpha_{y}} + \dist_{Q_{}}(x'_{\px,\nd},y'_{\px,\nd}) \leq O((s_{\tree} + s_{\nc})\cdot s_{\nc}\cdot s_{\ed}\cdot i_{\max})\cdot \dist_{G\setminus F}(x,y).
\]
\end{restatable}

\subsubsection{Decompose the Witness Path $P$: Proof of \Cref{lemma:FindingProxy,lemma:UpperBoundSolution}}
\label{sect:OraclePathDecompose}

Consider the following recursive decomposition procedure, which recursively decomposes the witness path $P$ into subpaths called \emph{segments}. Naturally, the procedure starts with the whole path $P$, which is the unique \emph{initial segment}. Each of the other segments appearing in this procedure will be classified according to two independent criteria:
\begin{itemize}
\item as \emph{recursive} or \emph{leaf}, and
\item as \emph{leftmost}, \emph{middle}, or \emph{rightmost}.
\end{itemize}
In particular, the initial segment is also viewed as a recursive segment.

\paragraph{Notations on Segments.} Before describing the procedure, we introduce some general notations. For two vertices $w$ and $z$ on $P$, we let $P[w,z]$ denotes the subpath of $P$ with $w$ and $z$ as the left and right endpoints. We use $\sigma$ to denote segments, which is always of the form $\sigma=P[w,z]$. %

\medskip

\noindent\underline{The Corresponding Level, Length Scale, and Clusters (or Edge).} Each segment $\sigma$ corresponds to a level $i_{\sigma}$, a length scale $\alpha_{\sigma}$, and a cluster $S_{\sigma}\in {\cal N}_{i_{\sigma}, \alpha_{\sigma}}$. In a few corner cases, the corresponding cluster $S_{\sigma}$ degenerates into a single original edge $e_{\sigma}$. The level $i_{\sigma}$ and the cluster $S_{\sigma}$ (or edge $e_{\sigma}$) are determined by the procedure, whereas the length scale $\alpha_{\sigma}$ is known in advance, which is the unique length scale such that
\[
\leng_{G}(\sigma)\leq h_{\alpha_{\sigma}}<2\leng_{G}(\sigma).
\]
We note that two segments may represent the same subpath of $P$ while having different corresponding levels or corresponding clusters (or edges).

\medskip

\noindent\underline{The Interior of a Segment.} For each segment $\sigma = P[w,z]$, its \emph{interior}, denoted by $\sigma^{\circ}$, is a (possibly empty) subpath of this segment, defined as follows. If $\sigma$ is an initial, leftmost, rightmost, or middle segment, then its interior $\sigma^{\circ}$ is $P[w_{}, z_{}]$, $P(w_{}, z_{}]$, $P[w_{}, z_{})$, and $P(w_{}, z_{})$, respectively. Here, $P(w_{}, z_{}]$ denotes the subpath of $P$ obtained by deleting from $P[w_{},z_{}]$ the left endpoint $w_{}$. The other notations are defined analogously.

\paragraph{The Decomposition Procedure.} Now, we are ready to describe the decomposition procedure. 

\medskip

\noindent\underline{The Initial Segment.} The initial segment $\sigma = P[w,z]$ is the entire path $P$ (i.e., $w=x$ and $z=y$). %
Its corresponding level $i_{\sigma}$ is the maximum level such that $\sigma$ intersects $A_{i_{\sigma}}$. The cluster $S_{\sigma}$ is defined to be $S_{\sigma}:={\cal N}_{i_{\sigma},\alpha_{\sigma}}(w_{})$. We note that choosing such $i_{\sigma}$ (rather than starting from $i_{\max}$) is purely for convenience and is not essential (see \Cref{remark:initialI} for more discussion).

For better understanding, we point out that the above definition guarantees that the interior $\sigma^{\circ}$ (which is exactly $\sigma$ since $\sigma$ is the initial segment) intersects $A_{i_{\sigma}}$ and is disjoint from $A_{>i_{\sigma}}$. Moreover, we have $\alpha_{\sigma} = \alpha^{\star}$ by definition, and thus \Cref{assume:HighProbability} guarantees that both ${\cal N}_{i_{\sigma},\alpha_{\sigma}}(w_{})$ and ${\cal N}_{i_{\sigma},\alpha_{\sigma}}(z_{})$ are affected. 

\medskip

\noindent\underline{Decomposing an Recursive Segment.} Let $\sigma$ be a recursive segment. The procedure will guarantee that, at this moment, the level $i_{\sigma}$, length scale $\sigma_{\sigma}$ and cluster $S_{\sigma}$ have been determined.

\begin{itemize}
\item 
Let $r_{1},r_{2},...,r_{k}$ be the $A_{i_{\sigma}}$-vertices on $\sigma^{\circ}$ (ordered from the left to the right), called the \emph{breakpoints} of $\sigma$. It is possible that $\sigma^{\circ}$ is disjoint from $A_{i_{\sigma}}$ and thus such $r_{1},...,r_{k}$ do not exist. %

\item The segments generated in this decomposition step are naturally $P[w_{},r_{1}]$, $P[r_{1},r_{2}]$, ..., $P[r_{k-1},r_{k}]$, and $P[r_{k},z_{}]$ (excluding all the empty segments containing no edge). We use $\check{\sigma}$ to denote such a \emph{generated segment}.

\end{itemize}
For each generated segment $\check{\sigma}$, we define its level $i_{\check{\sigma}}:= i_{\sigma}-1$. Next, we classify the generated segments as leftmost, middle, or rightmost segments (define their corresponding clusters or edges along the way), and then classify them as recursive or leaf segments.

\medskip

\noindent\underline{Leftmost/Middle/Rightmost Classification.}
\begin{itemize}
\item If (i) $\sigma$ is the initial segment or a leftmost segment, and (ii) $w_{}$ and $r_{1}$ are different, then we classify $\check{\sigma}:=P[w_{},r_{1}]$ as a leftmost segment, and define its corresponding cluster $S_{\check{\sigma}}:={\cal N}_{i_{\check{\sigma}},\alpha_{\check{\sigma}}}(w_{})$.
\item If (i) $\sigma$ is the initial segment or a rightmost segment, and (ii) $z_{}$ and $r_{k}$ are different, then we classify $\check{\sigma}:=P[r_{k},z_{}]$ as a rightmost segment, and define its corresponding cluster $S_{\check{\sigma}}:={\cal N}_{i_{\check{\sigma}},\alpha_{\check{\sigma}}}(z_{})$.
\item For each of the remaining generated segment $\check{\sigma}:=P[\check{w},\check{z}]$, we classify them as an middle segment. If $\check{\sigma}$ is a single edge $(\check{w},\check{z})$, then define its corresponding edge $e_{\check{\sigma}}:=(\check{w},\check{z})$, otherwise, define its corresponding cluster $S_{\check{\sigma}}:={\cal N}_{i_{\check{\sigma}},\alpha_{\check{\sigma}}}(\check{w}^{\circ})$, where $\check{w}^{\circ}$ is the vertex on $\check{\sigma}$ adjacent to the left endpoint $\check{w}$ (i.e., $\check{w}^{\circ}$ is the left endpoint of $\check{\sigma}$'s interior). 
\end{itemize}
\begin{remark}
\label{remark:initialI}
We emphasize when $\sigma$ is the initial segment, since we define $i_{\sigma}$ to be the maximum level s.t. $\sigma$ intersects $A_{i_{\sigma}}$, there will be at least one breakpoint $r_{1},...,r_{k}$, and thus a generated segment $\check{\sigma}$ will not be classified as a leftmost segment and a rightmost segment simultaneously. This is purely for convenience and is not essential. The proof still goes through if we set $i_{\sigma}:=i_{\max}$ for the initial segment.
\end{remark}

\medskip

\noindent\underline{Recursive/Leaf Classification.} For each generated segment $\check{\sigma}$, if it has a corresponding cluster $S_{\check{\sigma}}$ and $S_{\check{\sigma}}$ is affected, then $\check{\sigma}$ is a recursive segment (and thus it will be further decomposed), otherwise, it is a leaf segment.

\medskip

This finishes the description of the decomposition process. By considering the recursion tree, we can naturally define the \emph{ancestor-descendant relationships} among segments. For better understanding, we note that the leftmost segments only appear in the ``leftmost root-to-leaf path'' in the recursion tree, and this “leftmost root-to-leaf path” begins with the initial segment, followed by zero or more leftmost segments and then zero or more middle segments. Similar observations hold for the rightmost segments and the ``rightmost root-to-leaf path''.

\paragraph{Properties of Segments.} Next, we state some properties of the segments that follow directly from the description of the procedure.
\begin{enumerate}
\item\label{Item:SegProperty1} For each segment $\sigma$, its endpoints outside the interior $\sigma^{\circ}$ belong to $A_{>i_{\sigma}}$, and its interior $\sigma^{\circ}$ is disjoint from $A_{>i_{\sigma}}$. %

\item\label{Item:SegProperty2} For each segment $\sigma$ whose $S_{\sigma}$ exists, its interior $\sigma^{\circ}$ is fully contained in $S_{\sigma}$. To see this, note that 
\[
\leng_{G_{i_{\sigma},\alpha_{\sigma}}}(\sigma^{\circ})\leq \leng_{G}(\sigma^{\circ}) + h_{\alpha_{\sigma}}\leq 2h_{\alpha_{\sigma}} = h_{\cov,\alpha_{\sigma}},
\]
where the first inequality is by \Cref{assume:HighProbability}, \Cref{lemma:HighProbabilityEvent} and that $\sigma^{\circ}$ is disjoint from $A_{>i_{\sigma}}$ (Property \ref{Item:SegProperty1}). Also, for the endpoints $w^{\circ}$ and $z^{\circ}$ of $\sigma^{\circ}$, by definition the cluster $S_{\sigma}$ covers either $\Ball_{G_{i_{\sigma},\alpha_{\sigma}}}(w^{\circ}, h_{\cov,\sigma})$ or $\Ball_{G_{i_{\sigma},\alpha_{\sigma}}}(z^{\circ}, h_{\cov,\sigma})$. Hence, $S_{\sigma}$ covers the whole $\sigma^{\circ}$.

\item\label{Item:SegProperty3} For each segment $\sigma$ whose $S_{\sigma}$ exists, its endpoints outside the interior $\sigma^{\circ}$ belong to $\Lambda(S_{\sigma})$. To see this, let $w$ be such an endpoint. Note that $\sigma^{\circ}\subseteq S_{\sigma}$ (Property \ref{Item:SegProperty2}) and $w$ has an edge of length at most $\leng_{G}(\sigma)\leq h_{\alpha_{\sigma}}$ connecting $\sigma^{\circ}$, so by definition, $w\in \Lambda(S_{\sigma})$. 

\item\label{Item:SegProperty4} Each recursive segment has an affected $S_{\sigma}$, while each leaf segment has either an unaffected $S_{\sigma}$ or a corresponding edge $e_{\sigma}$.

\item\label{Item:SegProperty5} If $\sigma$ is the initial segment, both ${\cal N}_{i_{\sigma},\alpha_{\sigma}}(x)$ and ${\cal N}_{i_{\sigma},\alpha_{\sigma}}(y)$ are affected. The initial segment has a corresponding cluster $S_{\sigma} = {\cal N}_{i_{\sigma},\alpha_{\sigma}}(x)$. Each leftmost segment $\sigma$ has $S_{\sigma} = {\cal N}_{i_{\sigma},\alpha_{\sigma}}(x)$. Similarly, each rightmost segment has $S_{\sigma} = {\cal N}_{i_{\sigma},\alpha_{\sigma}}(y)$.

\end{enumerate}

In the following analysis, we focus on the leaf segments. We have following property about the leftmost/middle/rightmost classification of leaf segments.
\begin{enumerate}[start=6]
\item\label{remark:LeafSegClassify}
Concatenating the leaf segments reconstructs the original path $P$. Furthermore, only the first (resp. last) leaf segment in this concatenation can be a leftmost (resp. rightmost) leaf segment. It is also possible, however, that the first or last leaf segment is instead classified as a middle segment.
\end{enumerate}

\noindent{\underline{Endpoint Nodes of Leaf Segments.}} Now, for each leaf segment $\sigma$, and each endpoint $w$ of $\sigma$ that lies outside its interior $\sigma^{\circ}$, we will define its \emph{endpoint node} $w_{\nd}\in V_{\nd}(w)$ as follows. In other words, we define endpoint nodes for all leaf-segment endpoints, except for the left endpoint of the leftmost leaf segment and the right endpoint of the rightmost leaf segment.

It is clear from the decomposition that such an endpoint $w$ must be a breakpoint of a unique recursive segment $\sigma_{w}$, and $\sigma_{w}$ must be an ancestor of $\sigma$. Then naturally, we define $w_{\nd}$ to be the node in $V_{\nd}(w)$ with $\cluster(w_{\nd}) = S_{\sigma_{w}}$. Note that $w_{\nd}$ is well-defined because $w\in V_{\tmn}(S_{\sigma_{w}})$ (since both $w$ and $S_{\sigma_{w}}$ are at level $i_{\sigma_{w}}$) and $S_{\sigma_{w}}$ is affected (by Property \ref{Item:SegProperty4}).

\paragraph{Mapping Middle Leaf Segments to the Query Graph.} Now, we map each middle leaf segment $\sigma = P[w,z]$ to a path $P_{\qry,\sigma}$ on the query graph $Q_{}$. The endpoints of $P_{\qry,\sigma}$ will be exactly the endpoint nodes $w_{ \nd}$ and $z_{ \nd}$ of $\sigma$. Furthermore, when defining $P_{\qry,\sigma}$, we may designate a node on it as its \emph{portal}, denoted by $\portal(\sigma)$. Our goal is to define $P_{\qry,\sigma}$ and $\portal(\sigma)$ satisfies the following \Cref{lemma:SingleMappingPathLength}. 
\begin{lemma}
\label{lemma:SingleMappingPathLength}
If $P_{\qry,\sigma}$ has no portal, then 
\[
\leng_{Q_{}}(P_{\qry,\sigma})\leq (h_{\diam,\alpha_{\sigma}}/h_{\alpha_{\sigma}}+2)\cdot \leng_{G}(\sigma).
\]
If $P_{\qry,\sigma}$ has a portal $\portal(\sigma)$, then $\portal(\sigma)$ is the cluster node for the $S_{\sigma'}$ of some recursive segment $\sigma'$ in the decomposition, and 
\[
\leng_{Q_{}}(P_{\qry,\sigma})\leq 5\cdot \max\left\{1,\frac{h_{\tree,\alpha}}{h_{\ed,\alpha}\cdot s_{\ed}},\frac{\hat{h}_{\ed,\alpha}}{h_{\ed,\alpha}}\right\}\cdot \ell_{}(\portal(\sigma))
\]
for $\alpha := \scale(\portal(\sigma))$.
\end{lemma}

For clarity, we write $i_{w},i_{z},\alpha_{w},\alpha_{z},S_{w},S_{z}$ for $\level(w_{ \nd})$, $\level(z_{ \nd})$, $\scale(w_{ \nd})$, $\scale(z_{ \nd})$, $\cluster(w_{ \nd})$, $\cluster(z_{ \nd})$. Also recall that, if $\sigma$ has a corresponding cluster $S_{\sigma}$, then the level and length scale of $\sigma$ and $S_{\sigma}$ are the same, i.e., $\level(S_{\sigma}) = i_{\sigma}$ and $\scale(S_{\sigma}) = \alpha_{\sigma}$. The following \Cref{ob:LevelScaleConsistent} will be important to show the existence of different types of links.

\begin{observation}
\label{ob:LevelScaleConsistent}
When $\sigma$ has a corresponding cluster $S_{\sigma}$, the following holds for the clusters $S_{w},S_{z}$ and $S_{\sigma}$.
\begin{itemize}
\item $i_{w},i_{z}>i_{\sigma}$, and $\alpha_{w},\alpha_{z}\geq \alpha_{\sigma}$.
\item If $i_{w}\geq i_{z}$, then $\alpha_{w}\geq \alpha_{z}$; otherwise $\alpha_{w}\leq \alpha_{z}$.
\item $S_{w}$ and $S_{z}$ are affected, while $S_{\sigma}$ is unaffected.
\item Both endpoint vertices $w$ and $z$ are inside $\Lambda(S_{\sigma})$.
\item $S_{w}$ intersects $S_{z}$.
\end{itemize}
\end{observation}
\begin{proof}
For the first two items, recall that $\sigma_w$ and $\sigma_z$, the unique recursive segments having $w$ and $z$ as breakpoints, respectively, are both strict ancestors of $\sigma$. Therefore, we have $i_{\sigma_{w}},i_{\sigma_{z}}>i_{\sigma}$ and $\alpha_{\sigma_{w}},\alpha_{\sigma_{z}}>\alpha_{\sigma}$ (since $\sigma$ is a subpath of $\sigma_{w}$ and $\sigma_{z}$). By similar reason, if $i_{\sigma_{w}}\geq i_{\sigma_{z}}$, then $\alpha_{\sigma_{w}}\geq \alpha_{\sigma_{z}}$. Lastly, by definition, we have $\level(w_{\nd}) = i_{\sigma_{w}}$, $\scale(w_{\nd}) = \alpha_{\sigma_{w}}$, $\level(z_{\nd}) = i_{\sigma_{z}}$ and $\scale(z_{\nd}) = \alpha_{\sigma_{z}}$.

The third item follows from Property \ref{Item:SegProperty4}. The fourth item follows from Property \ref{Item:SegProperty3}. The last item is because, when $\sigma$ has a corresponding cluster $S_{\sigma}$, its interior $\sigma^{\circ}$ is not empty. Furthermore, $\sigma^{\circ}\subseteq \sigma_{w}^{\circ}\subseteq S_{w}$ by Property \ref{Item:SegProperty2}, and similarly, $\sigma^{\circ}\subseteq S_{z}$. Therefore, $S_{w}$ intersects $S_{z}$.
\end{proof}

We will prove \Cref{lemma:SingleMappingPathLength} by considering the following four cases. Let $S_{w,\nd}$ and $S_{z,\nd}$ be the cluster nodes created for $S_{w}$ and $S_{z}$ respectively. We assume $i_{w}\geq i_{z}$ without loss of generality, and thus $\alpha_{w}\geq \alpha_{z}$ by \Cref{ob:LevelScaleConsistent}.

\medskip

\noindent\underline{Case 1.} Suppose $\sigma=P[w_{},z_{}]$ is a single edge $(w_{},z_{})$. Then there is a type-\ref{Item:Type1Link} link connecting $w_{ \nd}$ and $z_{ \nd}$ with link length $ \leng_{G}(\sigma)$.
\begin{enumerate}[label=(\alph*)]
\item If both $w_{ \nd}$ and $z_{ \nd}$ are singleton nodes, define
\[
P_{\qry,\sigma} := \underbrace{(w_{ \nd},z_{ \nd})}_{\text{Type-\ref{Item:Type1Link}}},
\]
and $P_{\qry,\sigma}$ has no portal.
\item If $w_{ \nd}$ is a subtree node (the case in which $z_{ \nd}$ is a subtree node but $w_{ \nd}$ is not can be handled analogously), %
define
\[
P_{\qry,\sigma} := \underbrace{(w_{ \nd},S_{w,\nd})}_{\text{Type-\ref{Item:Type4Link}}}
\circ
\underbrace{(S_{w,\nd},w_{ \nd})}_{\text{Type-\ref{Item:Type4Link}}}
\circ
\underbrace{(w_{ \nd},z_{ \nd})}_{\text{Type-\ref{Item:Type1Link}}},
\]
and define $\portal(\sigma):=S_{w,\nd}$.
\end{enumerate}

\noindent\underline{Proof of \Cref{lemma:SingleMappingPathLength} in Case 1.} In subcase (a), $\leng_{Q_{}}(P_{\qry,\sigma}) = \leng_{G}(\sigma)$ as desired, since $w_{ \nd}$ and $z_{ \nd}$ have node length $0$. In subcase (b), it suffices to show that each node and link on $P_{\qry,\sigma}$ has length at most $\max\{1,h_{\tree,\alpha_{w}}/(h_{\ed,\alpha_{w}}\cdot s_{\ed})\}\cdot \ell_{}(\portal(\sigma))$. To see this, recall that $\ell_{}(\portal(\sigma)) = h_{\ed,\alpha_{w}}\cdot s_{\ed} + 2h_{\diam,\alpha_{w}} + 4h_{\alpha_{w}}$, $w_{ \nd}$ and $z_{ \nd}$ have node length at most $h_{\tree,\alpha_{w}}$, type-\ref{Item:Type4Link} links have length zero, and the type-\ref{Item:Type1Link} link $(w_{ \nd},z_{ \nd})$ has length $\leng_{G}(\sigma)\leq 2h_{\alpha_{\sigma}}\leq 2h_{\alpha_{w}}$.

\medskip

\noindent\underline{Case 2.}
Suppose $S_{\sigma}$ exists and it is both nearby $A_{i_{w}}$-sparse and nearby $A_{i_{z}}$-sparse. Then there is a type-\ref{Item:Type2Link} link connecting $w_{ \nd}$ and $z_{ \nd}$ of length $h_{\diam,\alpha_{\sigma}} + 2h_{\alpha_{\sigma}}$. \Cref{ob:LevelScaleConsistent} and the case assumptions show that all conditions for adding this link are fulfilled. The definitions of $P_{\qry,\sigma}$ and $\portal(\sigma)$ are exactly the same as in Case 1, except that we replace the type-\ref{Item:Type1Link} link $(w_{ \nd},z_{ \nd})$ with this type-\ref{Item:Type2Link} link (in particular, we still need to consider the two subcases).

\medskip

\noindent\underline{Proof of \Cref{lemma:SingleMappingPathLength} in Case 2.} Since the type-\ref{Item:Type2Link} has length $h_{\diam,\alpha_{\sigma}} + 2h_{\alpha_{\sigma}}$, in subcase (a), we have $\leng_{Q_{}}(P_{\qry,\sigma}) = h_{\diam,\alpha_{\sigma}} + 2h_{\alpha_{\sigma}} \leq (h_{\diam,\alpha_{\sigma}}/h_{\alpha_{\sigma}}+2)\cdot \leng_{G}(\sigma)$.
Regarding the subcase (b), the proof is analogous to that in Case 1.

\medskip

\noindent\underline{Case 3.} Suppose $S_{\sigma}$ exists, and it is nearby $A_{i_{w}}$-dense. Define
\[
P_{\qry,\sigma} := \underbrace{(w_{ \nd},S_{w,\nd})}_{\text{Type-\ref{Item:Type5Link}}}\circ \underbrace{(S_{w,\nd},z_{ \nd})}_{\text{Type-\ref{Item:Type5Link}}},
\]
and define $\portal(\sigma):=S_{w,\nd}$.

Before moving on, let us justify the existence of these two type-\ref{Item:Type5Link} links. For the link $(w_{ \nd},S_{w,\nd})$, we can easily verify that the conditions for adding this link are all fulfilled if we plug in 
\[
(i_{2},\alpha_{2},S_{2}):=(i_{w},\alpha_{w},S_{w}),\qquad q_{\nd}:=w_{ \nd},\qquad (i_{1},\alpha_{1},S_{1}):=(i_{w},\alpha_{w},S_{w}).
\]
In particular, $w_{\nd}$ is adjacent-cluster $A_{i_{w}}$-dense because $S_{\sigma}$ is nearby $A_{i_{w}}$-dense and $\Lambda(S_{\sigma})$ intersects $V(q_{\nd})$ (recall that $w_{}\in\Lambda(S_{\sigma})$ from \Cref{ob:LevelScaleConsistent}). 

For the link $(S_{w,\nd},z_{ \nd})$, we can also verify all the conditions by plugging in 
\[
(i_{2},\alpha_{2},S_{2}):=(i_{z},\alpha_{z},S_{z}),\qquad q_{\nd}:=z_{ \nd},\qquad (i_{1},\alpha_{1},S_{1}):=(i_{w},\alpha_{w},S_{w}).
\]
In particular, $z_{\nd}$ is adjacent-cluster $A_{i_{w}}$-dense (by a similar reason), and we also have $S_{1}$ intersects $S_{2}$ by \Cref{ob:LevelScaleConsistent}.

\medskip

\noindent\underline{Proof of \Cref{lemma:SingleMappingPathLength} in Case 3.} Similar to Case 1, we can see that each node and link on $P_{\qry,\sigma}$ has length at most $\max\{1,h_{\tree,\alpha_{w}}/(h_{\ed,\alpha_{w}}\cdot s_{\ed})\}\cdot \ell_{}(\portal(\sigma))$. Note that type-\ref{Item:Type5Link} links have length $0$.

\medskip

\noindent\underline{Case 4.} Suppose $S_{\sigma}$ exists, and it is nearby $A_{i_{w}}$-sparse but nearby $A_{i_{z}}$-dense. Let $\hat{S}_{\sigma}:=\satellite_{i_{z},\alpha_{z}}(S_{\sigma})$. Note that $\hat{S}_{\sigma}$ is the satellite cluster of the cluster $S_{\sigma}$ but it has different level $i_{z}>i_{\sigma}$ and scale $\alpha_{z}>\alpha_{\sigma}$, rather than $i_{\sigma}$ and $\alpha_{\sigma}$. Let $\hat{S}_{\nd}$ be the satellite node created for the satellite cluster $\hat{S}_{\sigma}$. The following \Cref{claim:AffectedSatelliteCluster} shows that $S_{\sigma}$ is affected, and thus $\hat{S}_{\nd}$ indeed exists.

\begin{claim}
\label{claim:AffectedSatelliteCluster}
We have $S_{z}\subseteq \hat{S}_{\sigma}$. Therefore, the satellite cluster $\hat{S}_{\sigma}:=\satellite_{i_{z},\alpha_{z}}(S_{\sigma})$ is affected since $S_{z}$ is affected.
\end{claim}
\begin{proof}
The inclusion $S_{z}\subseteq \hat{S}_{\sigma}$ follows directly from \Cref{ob:SatelliteCover}, together with $z\in \Lambda(S_{\sigma})$ as established in \Cref{ob:LevelScaleConsistent}.
\end{proof}

We consider two subcases.
\begin{enumerate}[label=(\alph*)]
\item If $w_{ \nd}$ is a subtree node, define 
\[
P_{\qry,\sigma} := 
\underbrace{(w_{ \nd},S_{w,\nd})}_{\text{Type-\ref{Item:Type4Link}}}
\circ
\underbrace{(S_{w,\nd},w_{ \nd})}_{\text{Type-\ref{Item:Type4Link}}}
\circ
\underbrace{(w_{ \nd},\hat{S}_{\nd})}_{\text{Type-\ref{Item:Type3Link}}}
\circ 
\underbrace{(\hat{S}_{\nd},z_{ \nd})}_{\text{Type-\ref{Item:Type6Link}}},
\]
and define $\portal(\sigma)$ to be the cluster node $S_{w,\nd}$.
\item If $w_{ \nd}$ is a singleton node, define 
\[
P_{\qry,\sigma} := 
\underbrace{(w_{ \nd},\hat{S}_{\nd})}_{\text{Type-\ref{Item:Type3Link}}}
\circ 
\underbrace{(\hat{S}_{\nd},z_{ \nd})}_{\text{Type-\ref{Item:Type6Link}}}
\circ
\underbrace{(z_{ \nd},S_{z,\nd})}_{\text{Type-\ref{Item:Type5Link}}}
\circ
\underbrace{(S_{z,\nd},z_{ \nd})}_{\text{Type-\ref{Item:Type5Link}}},
\]
and define $\portal(\sigma)$ to be the cluster node $S_{z,\nd}$.
\end{enumerate}
Again, let us justify the existence of each link above. First, consider the type-\ref{Item:Type3Link} link $(w_{ \nd},\hat{S}_{\nd})$ in both subcases (a) and (b). We can verify all the conditions for it by plugging in 
\[
(i,\alpha,S):=(i_{\sigma},\alpha_{\sigma},S_{\sigma}),\qquad (i_{1},i_{2}):=(i_{w},i_{z}),\qquad (u,u_{\nd}):=(w_{},w_{ \nd}),\qquad \text{$\hat{S}_{\nd}$ remains the same}
\]
In particular, the case assumptions guarantee that $S$ is nearby $A_{1}$-sparse but nearby $A_{2}$-dense. %

Regarding the type-\ref{Item:Type6Link} link $(\hat{S}_{\nd},z_{ \nd})$ in both subcases (a) and (b), we can verify the conditions by plugging in 
\[
(i_{2},\alpha_{2},S_{2}):=(i_{z},\alpha_{z},S_{z}),\qquad q_{\nd}:=z_{ \nd},\qquad \hat{S}_{2}:=\hat{S}_{\sigma}.
\]
In particular, $q_{\nd}$ is adjacent-cluster $A_{i_{2}}$-dense since $z\in\Lambda(S_{\sigma})$ (\Cref{ob:LevelScaleConsistent}). Also, $\hat{S}_{2}\ni \pivot(S_{z})$ since $\pivot(S_{z})\in S_{z}\subseteq \hat{S}_{\sigma}$ (\Cref{claim:AffectedSatelliteCluster}).

The type-\ref{Item:Type4Link} links in subcase (a) exist since $w_{ \nd}$ is a subtree node.

The existence of type-\ref{Item:Type5Link} links in subcase (b) follows the argument in Case 3 and the fact that $S_{\sigma}$ is nearby $A_{i_{z}}$-dense.

\medskip

\noindent\underline{Proof of \Cref{lemma:SingleMappingPathLength} in Case 4.} In this case, all the links in $P_{\qry,\sigma}$ have $0$ length. For the node on $P_{\qry,\sigma}$, we will only argue that $\hat{S}_{\nd}$ has node length $\ell_{}(\hat{S}_{\nd})\leq (\hat{h}_{\ed,\alpha}/h_{\ed,\alpha})\cdot \ell_{}(\portal(\sigma))$ for $\alpha := \scale(\portal(\sigma))$. The node lengths of other nodes can be bounded using arguments in previous cases.

Recall that $\ell_{}(\hat{S}_{\nd})=\hat{h}_{\ed,\alpha_{z}}\cdot s_{\ed} + 2h_{\diam,\alpha_{z}} + 4h_{\alpha_{z}}$, and that $\ell_{}(\portal(\sigma))=h_{\ed,\alpha}\cdot s_{\ed} + 2h_{\diam,\alpha} + 4h_{\alpha}$, where $\alpha$ is either $\alpha_{w}$ (subcase (a)) or $\alpha_{z}$ (subcase (b)). Since $\alpha_{w}\geq \alpha_{z}$, the inequality $\ell_{}(\hat{S}_{\nd})\leq (\hat{h}_{\ed,\alpha}/h_{\ed,\alpha})\cdot \ell_{}(\portal(\sigma))$ clearly holds.

\paragraph{Defining the Query Path $P_{\qry}$.}

Finally, we define the query path $P_{\qry}$, the proxy vertices $x_{\px},y_{\px}$, and their proxy nodes $x_{\px,\nd},y_{\px,\nd}$. 

Recall the discussion about leaf segments in Property \ref{remark:LeafSegClassify}. We use $\sigma_{\ell}$ and $\sigma_{r}$ to denote the leftmost and rightmost leaf segments respectively ($\sigma_{\ell},\sigma_{r}$ not necessarily exist), and let $\sigma_{1},\sigma_{2},...,\sigma_{k_{\max}}$ to denote the middle leaf segments in between in the left-to-right order.

Define
\begin{equation*}
(x_{\px},x_{\px,\nd}) := 
\left\{
\begin{aligned}
&(x,\ \text{the left endpoint node of $\sigma_{1}$})&&\text{if $\sigma_{\ell}$ does not exist}\\
&(\text{the right endpoint of $\sigma_{\ell}$},\ \text{the right endpoint node of $\sigma_{\ell}$})&&\text{if $\sigma_{\ell}$ exists}
\end{aligned}
\right.,
\end{equation*}
and similarly
\begin{equation*}
(y_{\px},y_{\px,\nd}) := 
\left\{
\begin{aligned}
&(y,\ \text{the right endpoint node of $\sigma_{k_{\max}}$})&&\text{if $\sigma_{r}$ does not exist}\\
&(\text{the left endpoint of $\sigma_{r}$}\ \text{the left endpoint node of $\sigma_{r}$})&&\text{if $\sigma_{r}$ exists}
\end{aligned}
\right..
\end{equation*}
Intuitively, $x_{\px}$ is the leftmost breakpoint and $y_{\px}$ is the rightmost breakpoint during the decomposition.

The query path $P_{\qry}$ is defined by the following procedure \Cref{algo:ConstructPqry}. Roughly speaking, we construct $P_{\qry}$ by first concatenating $P_{\qry,\sigma_{1}},P_{\qry,\sigma_{2}}...,P_{\qry,\sigma_{k_{\max}}}$ and then shortcut it using portals. 

\begin{algorithm}[H]
\caption{Construction of $P_{\qry}$}
\label{algo:ConstructPqry}
\begin{algorithmic}[1]
\State Initialize $P_{\qry}$ to be empty, $k\gets 1$
\While{$k\leq k_{\max}$}
\If{$P_{\qry,\sigma_{k}}$ has $\portal(\sigma_{k})$, and some later path $P_{\qry,\sigma_{k'}}$ with $k'>k$
\Statex ~~~~~~~has the same $\portal(\sigma_{k'}) = \portal(\sigma_{k})$}
\State Choose the largest such $k'$
\State Let $P'$ be the concatenation of 
\Statex
\begin{center}
~~~the prefix of $P_{\qry,\sigma_{k}}$ before $\portal(\sigma_{k})$, ~~and~~ the suffix of $P_{\qry,\sigma_{k'}}$ after $\portal(\sigma_{k'})$
\end{center}
\State Append $P'$ to the end of $P_{\qry}$
\State $k\gets k'+1$
\Else
\State Append $P_{\qry,\sigma_{k}}$ to the end of $P_{\qry}$
\State $k\gets k+1$
\EndIf
\EndWhile

\end{algorithmic}
\end{algorithm}

The motivation of constructing $P_{\qry}$ using this particular procedure is that, simply concatenating these $P_{\qry,\sigma}$ may give a $P_{\qry}$ that is too long, because those $P_{\qry,\sigma}$ with portals may not have their lengths competitive to the original segment lengths (see \Cref{lemma:SingleMappingPathLength}). To resolve this issue, we will ensure that, for each portal node, $P_{\qry}$ only walks through at most two $P_{\qry,\sigma_{k}}$ with this portal node. Hence, we can charge the lengths of these ``bad'' $P_{\qry,\sigma}$ to the lengths of portal nodes (using \Cref{lemma:SingleMappingPathLength}). The following \Cref{lemma:UpperBoundLengPqry} formally bounds the length of $P_{\qry}$. 

\begin{lemma}
\label{lemma:UpperBoundLengPqry}
$\leng_{Q_{}}(P_{\qry})\leq O((s_{\tree}+s_{\nc})\cdot s_{\nc}\cdot s_{\ed}\cdot i_{\max})\cdot \leng_{G}(P)$
\end{lemma}
\begin{proof}
For each $\sigma\in\{\sigma_{1},...,\sigma_{k_{\max}}\}$, we say $P_{\qry,\sigma}$ is \emph{visited} if $P_{\qry,\sigma}$, either in whole or in part, is appended to $P_{\qry}$. Therefore, $\leng(P_{\qry})$ is upper bounded by the total length of visited $P_{\qry,\sigma}$.

We classify the visited $P_{\qry,\sigma}$ into two groups: those without portals and those with portals. For those without portals, their total length can be simply bounded by
\[
\sum_{\substack{\text{visited}\ P_{\qry,\sigma}\\\text{ without portals}}}\leng_{Q_{}}(P_{\qry,\sigma})\leq \sum_{\text{leaf segments }\sigma} (h_{\diam,\alpha_{\sigma}}/h_{\alpha_{\sigma}} + 2)\cdot \leng_{G}(\sigma)\leq (5 s_{\nc} + 10)\cdot \leng_{G}(P),
\]
where the first inequality uses \Cref{lemma:SingleMappingPathLength}.

For those with portals, as we discussed, the key observation is that a portal is shared by at most two visited $P_{\qry,\sigma}$. Hence, the total length can be bounded by
\begin{align*}
&~~~~2\cdot\sum_{\text{distinct portals }p}5\cdot \max\left\{1,\frac{h_{\tree,\alpha}}{h_{\ed,\alpha}\cdot s_{\ed}},\frac{\hat{h}_{\ed,\alpha}}{h_{\ed,\alpha}}\right\}\cdot \ell_{}(p)~~~~\text{(for an arbitrary $\alpha$)}\\
&\leq 10\cdot\max\{s_{\tree},3s_{\nc}/5+1\}\cdot \sum_{\text{distinct portals }p}\ell_{}(p)~~~~~~~~~~~~(\text{\Cref{lemma:SingleMappingPathLength}})
\end{align*}
Next, since \Cref{lemma:SingleMappingPathLength} guarantees that portals are all cluster nodes of the clusters of recursive segments, we have
\begin{align*}
&~~~~\sum_{\text{distinct portals }p}\ell_{}(p)\\
&\leq \sum_{\substack{\text{recursive}\\\text{segments }\sigma}} \ell_{}(S_{\sigma,\nd})\\
&\leq \sum_{\substack{\text{recursive}\\\text{segments }\sigma}} (h_{\ed,\alpha_{\sigma}}\cdot s_{\ed} + 2h_{\diam,\alpha_{\sigma}} + 4h_{\alpha_{\sigma}})\\
&\leq\sum_{\substack{\text{recursive}\\\text{segments }\sigma}} ((5\cdot(2s_{\nc}\cdot h_{\alpha_{\sigma}})+8h_{\alpha_{\sigma}})\cdot s_{\ed} + 2\cdot(2\cdot s_{\nc}\cdot h_{\alpha_{\sigma}}) + 4h_{\alpha_{\sigma}})\\
&\leq\sum_{\substack{\text{recursive}\\\text{segments }\sigma}} (10s_{\nc}s_{\ed} + 8s_{\ed} + 4s_{\nc} + 4)\cdot (2\cdot\leng_{G}(\sigma))\\
&\leq (10s_{\nc}s_{\ed} + 8s_{\ed} + 4s_{\nc} + 4)\cdot (2\cdot i_{\max})\cdot\leng_{G}(P)
\end{align*}
Therefore, we can conclude that the total length of those visited $P_{\qry,\sigma}$ with portals can be bounded by $O((s_{\tree}+s_{\nc})\cdot s_{\nc}\cdot s_{\ed}\cdot i_{\max})\cdot \leng_{G}(P)$.

\end{proof}

We complete the proof of \Cref{lemma:FindingProxy} using \Cref{lemma:UpperBoundLengPqry}. For convenience, we restate \Cref{lemma:FindingProxy} here.

\FindingProxy*

\begin{proof}
In general, the proxy vertices and nodes $x_{\px},x_{\px,\nd},y_{\px},y_{\px,\nd}$ are exactly what we defined above (and thus Item 3 follows from \Cref{lemma:UpperBoundLengPqry}), except that in a corner case, we may choose a different $y_{\px,\nd}$.

\medskip

\noindent\underline{Item 1.} We first show Item 1 by picking appropriate $i_{x}, \alpha_{x}$ (and possibly $i'_{x},\alpha'_{x}$). 

Item 1(a) corresponds to the case where the leftmost leaf segment $\sigma_{\ell}$ does not exist. Namely, the first middle leaf segment $\sigma_{1}$ has $x$ as its left endpoint. In this case, $x_{\px} = x$ holds by definition. Naturally, we choose 
\[
i_{x} = \level(x_{\px,\nd}),\qquad\alpha_{x} = \level(x_{\px,\nd}).
\]
It suffices to show that $\cluster(x_{\px,\nd})$ is the same as $S_{x}:=\mathcal{N}_{i_{x},\alpha_{x}}(x)$ (which will also imply that $S_{x}$ is affected). Recall the definition of the endpoint node $x_{\px,\nd}$. $\cluster(x_{\px,\nd})$ must be the corresponding cluster $S_{\hat{\sigma}}$ for some ancestral segment $\hat{\sigma}$ of $\sigma_{1}$. Furthermore, $x$ must be a breakpoint of $\hat{\sigma}$, meaning that $\hat{\sigma}$ is the initial segment or a leftmost segment. Therefore, by Property \ref{Item:SegProperty5}, $S_{\hat{\sigma}}$ is exactly $\mathcal{N}_{i_{x},\alpha_{x}}(x)$.

Item 1(b) corresponds to the case that the leftmost leaf segment $\sigma_{\ell}$ exists. Recall that we defined $x_{\px}$ (resp. $x_{\px,\nd}$) to be the right endpoint (resp. endpoint node) of $\sigma_{\ell}$. This time, we choose
\[
i_{x}:=i_{\sigma_{\ell}},\qquad\alpha_{x}:=\alpha_{\sigma_{\ell}},\qquad i'_{x}:=\level(x_{\px,\nd}),\qquad \alpha'_{x} = \scale(x_{\px,\nd}).
\]
Because $S_{x}:=\mathcal{N}_{i_{x},\alpha_{x}}(x)$ is exactly $S_{\sigma_{\ell}}$ (by Property \ref{Item:SegProperty5}), we have $x_{\px}\in \Lambda(S_{x})$ by Property \ref{Item:SegProperty3}, and $S_{x}$ is unaffected by Property \ref{Item:SegProperty4}. 
Moreover, similar to the argument for Item 1(a), $\cluster(x_{\px,\nd})$ is $S_{\hat{\sigma}}$ for an ancestral segment $\hat{\sigma}$ of $\sigma_{\ell}$, meaning that $x$ is in its interior $\hat{\sigma}^{\circ}$. Therefore, $\hat{\sigma}$ is the initial segment or a leftmost segment, and thus $S'_{x}:={\cal N}_{i'_{x},\alpha'_{x}}(x)$ is exactly $S_{\hat{\sigma}}$.

\medskip

\noindent{\underline{Item 2.}} Next, we show Item 2. In general, we can simply follow the proof of Item 1 by replacing all ``leftmost'' with ``rightmost'', and exploit each rightmost segment $\hat{\sigma}$ has $S_{\hat{\sigma}} = {\cal N}_{i_{\hat{\sigma}},\alpha_{\hat{\sigma}}}(y)$ (by Property \ref{Item:SegProperty5}). 

The proof could fail only in the corner case where $\hat{\sigma}$ is the initial segment which has $S_{\hat{\sigma}} = {\cal N}_{i_{\hat{\sigma}},\alpha_{\hat{\sigma}}}(x)$ instead of ${\cal N}_{i_{\hat{\sigma}},\alpha_{\hat{\sigma}}}(y)$. In this case, we change the old $y_{\px,\nd}$ to a new node $y'_{\px,\nd}\in V_{\nd}(y_{\px})$ with $\cluster(y'_{\px,\nd})={\cal N}_{i_{\hat{\sigma}},\alpha_{\hat{\sigma}}}(y)$ which is exactly $S_{y}$ in Item 2(a) or $S'_{y}$ in Item 2(b). Note that Property \ref{Item:SegProperty5} guarantees that the initial segment $\hat{\sigma}$ have both ${\cal N}_{i_{\hat{\sigma}},\alpha_{\hat{\sigma}}}(x)$ and ${\cal N}_{i_{\hat{\sigma}},\alpha_{\hat{\sigma}}}(y)$ affected, so $y'_{\px,\nd}$ is well-defined.

\medskip

\noindent{\underline{Item 3.}} If we did not change the proxy node $y_{\px,\nd}$ in the proof of Item 2, then Item 3 follows directly from \Cref{lemma:UpperBoundLengPqry}. If the old $y_{\px,\nd}$ is changed to the new $y'_{\px,\nd}$, Item 3 still holds because $\dist_{Q}(y_{\px,\nd}, y'_{\px,\nd}) = 0$. Note that there is a type-\ref{Item:Type7Link} link connecting $y_{\px,\nd}$ and $y'_{\px,\nd}$ because $y_{\px}$ is a common vertex in both $V(y_{\px,\nd})$ and $V(y'_{\px,\nd})$.

\end{proof}

Finally, we finish the proof of \Cref{lemma:UpperBoundSolution} using \Cref{lemma:FindingProxy}. 

\UpperBoundSolution*

\begin{proof}
Consider the iteration where $i_{x},i'_{x},i_{y},i'_{y},\alpha_{x},\alpha'_{x},\alpha_{y},\alpha'_{y}$ are those picked by \Cref{lemma:FindingProxy}. Let $x_{\px},x_{\px,\nd},y_{\px},y_{\px,\nd}$ be those proxy vertices and nodes picked by \Cref{lemma:FindingProxy}. We remark that in this iteration, the algorithm will branch into
\begin{itemize}
\item Line 9 if Item 1(a) of \Cref{lemma:FindingProxy} holds,
\item Lines 12 and 13 if Item 1(b) of \Cref{lemma:FindingProxy} holds,
\item Line 17 if Item 2(a) of \Cref{lemma:FindingProxy} holds,
\item Lines 20 and 21 if Item 2(b) of \Cref{lemma:FindingProxy} holds.
\end{itemize}
Furthermore, in this iteration, \Cref{algo:QueryAlgorithm} will successfully pick $x'_{\px,\nd}$ and $y'_{\px,\nd}$ due to the guarantees from \Cref{lemma:FindingProxy} (and thus $\ans$ will be updated on Line 22). In particular, when the algorithm branch into Lines 12 and 13, it will successfully go through Line 12 since either (i) the set $\ProxyList(x,i_{x},\alpha_{x},i'_{x},\alpha'_{x})$ has $f+1$ vertices and thus at least one of them is not in $F$, or (ii) the set $\ProxyList(x,i_{x},\alpha_{x},i'_{x},\alpha'_{x})$ has less than $f+1$ vertices, which must include $x_{\px}\notin F$ due to Item 1(b) of \Cref{lemma:FindingProxy}.

Now we want to prove the inequality in \Cref{lemma:UpperBoundSolution} holds. It suffices to prove that for $x'_{\px,\nd}$, we have either $x'_{\px,\nd} = x_{\px,\nd}$ or
\[
\dist_{Q_{}}(x'_{\px,\nd},x_{\px,\nd})\leq O(s_{\tree}\cdot s_{\nc}\cdot s_{\ed})\cdot \dist_{G\setminus F}(x,y),
\]
and prove the analogous claim for $y'_{\px,\nd}$.
Combining these with Item 3 in \Cref{lemma:FindingProxy} and that $h_{\diam,\alpha_{x}},h_{\diam,\alpha_{y}}\leq h_{\diam,\alpha^{\star}}\leq 2\dist_{G\setminus F}(x,y)$, the inequality in \Cref{lemma:UpperBoundSolution} immediately follows.

Suppose Item 1(a) of \Cref{lemma:FindingProxy} holds, then the algorithm will pick $x'_{\px}$ and $x'_{\px,\nd}$ on Line 9, which are exactly the same as the $x_{\px}$ and $x_{\px,\nd}$ picked by \Cref{lemma:FindingProxy}. Similarly, when Item 2(a) of \Cref{lemma:FindingProxy} holds, $y'_{\px,\nd}$ is the same as $y_{\px,\nd}$.

Suppose Item 2(a) of \Cref{lemma:FindingProxy} holds, and thus $x'_{\px}$ and $x'_{\px,\nd}$ are picked on Lines 12 and 13. We need some more arguments since $x'_{\px,\nd}$ might not be the same as $x_{\px,\nd}$ in this case. Let $S:={\cal N}_{i_{x},\alpha_{x}}(x)$ and $S':={\cal N}_{i'_{x},\alpha'_{x}}(x)$. By \Cref{lemma:FindingProxy} and the algorithm statement, we have the following.
\begin{itemize}
\item $S$ is unaffected and $S'$ is affected.
\item $x_{\px},x'_{\px}\in\Lambda(S)$, $x_{\px,\nd}\in V_{\nd}(x_{\px})$, $x'_{\px,\nd}\in V_{\nd}(x'_{\px})$, and $\cluster(x_{\px,\nd}) = \cluster(x'_{\px,\nd}) = S'$.
\end{itemize}
Namely, both $x_{\px}$ and $x'_{\px}$ are in the surrounding set of an unaffected cluster $S$, and $x_{\px,\nd},x'_{\px,\nd}$ are their corresponding nodes for the same affected cluster $S'$. Thus, we can find a short path $P_{\qry,x}$ in $Q_{}$ connecting them as follows. Let $i':=\level(S')$.
\begin{itemize}
\item Suppose $S$ is nearby $A_{i'}$-sparse. Define
\[
P_{\qry,x}:=\underbrace{(x_{\px,\nd},x'_{\px,\nd})}_{\text{Type-\ref{Item:Type2Link}}}.
\]
\item Suppose $S$ is nearby $A_{i'}$-dense. Define
\[
P_{\qry,x}:= \underbrace{(x_{\px,\nd},S'_{\nd})}_{\text{Type-\ref{Item:Type5Link}}}\circ\underbrace{(S'_{\nd},x'_{\px,\nd})}_{\text{Type-\ref{Item:Type5Link}}}.
\]
\end{itemize}
Similar to the proof of \Cref{lemma:SingleMappingPathLength}, we can show $P_{\qry,x}$ is valid. By the way we set the node and link lengths, $\leng_{Q_{}}(P_{\qry,x})\leq O(\max\{h_{\tree,\alpha'_{x}},h_{\ed,\alpha'_{x}}\cdot s_{\ed}\})\leq O(s_{\tree}\cdot s_{\nc}\cdot s_{\ed}\cdot h_{\alpha'_{x}})\leq O(s_{\tree}\cdot s_{\nc}\cdot s_{\ed})\cdot \dist_{G\setminus F}(x,y)$ as desired.

\end{proof}

\subsubsection{Lower Bounding the Solution}
\label{sect:OracleSolLowerBound}

\begin{lemma}
\label{lemma:OracleSolLowerBound}
The $\ans$ returned by \Cref{algo:QueryAlgorithm} satisfies $\ans\geq \dist_{G\setminus F}(x,y)$.
\end{lemma}
\begin{proof}
If the final $\ans$ is obtained on Line 4, let $S\in\{{\cal N}_{i,\alpha}(x),{\cal N}_{i,\alpha}(y)\}$ be an unaffected cluster that contains both $x$ and $y$. Then clearly $\dist_{G\setminus F}(x,y) \leq \dist_{G[S]}(x,y)\leq h_{\diam,\alpha} = \ans$.

From now, suppose the final $\ans$ is obtained on Line 24. Consider the proxy vertices $x'_{\px},y'_{\px}$ and proxy nodes $x'_{\px,\nd}\in V(x'_{\px})$ and $y'_{\px,\nd}\in V(y'_{\px})$ picked in this iteration. By \Cref{lemma:GeneralSolLowerBound},
\[
\dist_{G\setminus F}(x'_{\px},y'_{\px})\leq \dist_{Q_{}}(x'_{\px,\nd},y'_{\px,\nd}).
\]
Thus it suffices to show that $\dist_{G\setminus F}(x,x'_{\px})\leq 2h_{\diam,\alpha'_{x}}$ and also $\dist_{G\setminus F}(y,y'_{\px})\leq 2h_{\diam,\alpha'_{y}}$. We prove the former below (the latter can be proved similarly). Note that we have either $x=x'_{\px}$ (when branching into Line 9), or for the unaffected cluster $S:={\cal N}_{i_{x},\alpha_{x}}(x)$, $x'_{\px}\in \Lambda(S)$ (when branching into Lines 12 and 13). In the latter case, by definition,
\[
\dist_{G\setminus F}(x,x'_{\px})\leq \dist_{G[S]}(x,x'_{\px}) + h_{\alpha_{x}}\leq h_{\diam,\alpha_{x}} + h_{\alpha_{x}}\leq 2h_{\diam,\alpha_{x}}.
\]

\end{proof}

\subsection{Oracle Construction and the Update Algorithm}
\label{sect:OracleConstruction}

In this section, we first describe how to precompute and store the data structures used in the update and query phases, and then present the update algorithm that extracts the query graph from them.

\subsubsection{The Euler Tours of Steiner Trees}

We first introduce a standard technique for ``linearizing'' each low-degree Steiner tree into a path using the Euler Tour  order. We note that this technique has been widely used in the literature on fault-tolerant connectivity and distance oracles (e.g., \cite{duan2010connectivity,duan2020connectivity,long2022near,HLRS26}).

\begin{definition}[Euler Tours]
Consider a cluster $S$ and its Steiner tree $T_{S}$. We define $\Euler(T_{S})$ to be its Euler tour, which is an ordered list that includes all \emph{occurrences} of vertices according to a DFS traversal of $T_{S}$, starting from an arbitrary root node. For convenience, we add two \emph{virtual occurrences} $\start(T_{S})$ and $\eend(T_{S})$ at the front and the end of $\Euler(T_{S})$ as the ``guard occurences''. For each vertex $v\in V_{\tmn}(S)$, the \emph{principal occurrences} of $v$ on $\Euler(T_{S})$ is its first occurrence. 
\end{definition}

\begin{definition}[Intervals]
For a cluster $S$ and its Steiner tree $T_{S}$, an \emph{interval} $I$ on $\Euler(T_{S})$ is a consecutive subsequence. The terminal vertices of $I$, denoted by $V_{\tmn}(I)$, are those vertices in $V_{\tmn}(S)$ whose principal occurrences fall in $I$. The left (resp. right) \emph{outer endpoints} of $I$ are the occurrences immediately to the left (resp. right) of $I$.
\end{definition}

\begin{definition}[Residual Intervals under Failures]
Consider the update with failed vertices $F$. For the Steiner tree $T$ of an \underline{affected} cluster, we use ${\cal I}_{T}$ to denote the \emph{residual intervals} of $T$, which are intervals obtained by removing $\start(T),\eend(T)$ and all occurrences of failed vertices from $\Euler(T)$.
\end{definition}

The following lemma is straightforward and describes the relationship between residual subtrees and residual intervals.

\begin{lemma}
\label{lemma:OracleUseIntervals}
In the update phase, for each Steiner tree $T$, there is an assignment that assigns a subset of residual intervals ${\cal I}_{\tau}\subseteq {\cal I}_{T}$ to each residual subtree $\tau\in T\setminus F$, satisfying the following.
\begin{itemize}
\item For each subtree $\tau\in T\setminus F$, $V_{\tmn}(\tau)=\bigsqcup_{I\in{\cal I}_{\tau}}V_{\tmn}(I)$.
\item $\{{\cal I}_{\tau}\mid \tau\in T\setminus F\}$ forms a partition of ${\cal I}_{T}$.
\end{itemize}
Moreover, we can compute the partition ${\Ical_{\tau}\mid \tau\in T\setminus F}$ in $\Otilde(|V(T)|)$ time.
\end{lemma}

Intuitively, \Cref{lemma:OracleUseIntervals} says that a residual subtree can be ``represented'' by a collection of residual intervals. As we will see soon, when we want to detect shortcut-based links incident to a subtree node, or determine different kinds of heaviness and denseness of a subtree node, we can work with each of its residual intervals individually, and then merge the answers. The following \Cref{lemma:OracleResidualIntervalNum} shows that the number of residual intervals is asymptotically the same as the number of residual subtrees, so for simplicity, we can assume each residual subtree is ``represented'' by only one residual interval. This assumption does not affect the space or update-time bounds asymptotically.

\begin{lemma}
\label{lemma:OracleResidualIntervalNum}
In the update phase, the number of residual intervals is at most $\Otilde(f)\cdot n^{O(1/k)}$.
\end{lemma}
\begin{proof}
This is because, for each Steiner tree $T$, the number of residual intervals generated by $T$ is at most $\Delta\cdot |T\cap F|$, where $\Delta = \Otilde(1)\cdot n^{O(1/k)}$ is the maximum degree. Furthermore, recall that each vertex can appear in at most $\gamma_{\tree} = \Otilde(1)\cdot n^{O(1/k)}$ many Steiner trees.
\end{proof}

\begin{definition}[Canonical Intervals]
\label{def:CanonicalIntervals}
Consider a Steiner tree $T$. For each occurrence $v_{\oc}\in \Euler(T)$, let ${\cal I}^{\rightarrow}_{T}(v_{\oc})$ denote the \emph{canonical intervals} on the right side of $v_{\oc}$, defined as follows. Let $v'_{\oc}\in \Euler(T)$ be the occurrence next to $v_{\oc}$ on the right. Then ${\cal I}^{\rightarrow}_{T}(v_{\oc})$ includes all intervals starting from $v'_{\oc}$ of size $1,2,4,8,...$ until the largest such interval, namely the suffix beginning with $v'_{\oc}$. Similarly, we can define the canonical intervals $\Ical^{\larr}_{T}(v_{\oc})$.
\end{definition}

\begin{observation}
\label{ob:OracleUseCanonical}
In the update phase, for each residual interval $I\in\Ical_{T}$ with left and right outer endpoints $u_{\oc}$ and $v_{\oc}$, there exists canonical intervals $I^{\rarr}\in \Ical^{\rarr}(u_{\oc})$ and $I^{\larr}\in \Ical^{\larr}(v_{\oc})$ such that $V_{\tmn}(I^{\rarr})\cup V_{\tmn}(I^{\larr}) = V_{\tmn}(I)$.
\end{observation}

For convenience, we will use $I^{\rarr}$ and $I^{\larr}$ to denote the canonical intervals found for a residual interval $I$ by the above \Cref{ob:OracleUseCanonical}.

\subsubsection{2D-Counting Tables for Shortcut-Based Links}

As pointed out by previous work (e.g. \cite{duan2020connectivity,long2022near}), the shortcut-based links can be easily detected by precomputing 2D counting tables. Since this part is essentially the same as in those works, we only sketch the idea below.

\paragraph{Indices of the Tables.} For each Steiner tree $T$ of a cluster $S$, let $\Euler_{\tmn}(T)$ be the sublist of $\Euler(T)$ which includes only the principal occurrences. In other words, $\Euler_{\tmn}(T)$ is an ordered list of vertices in $V_{\tmn}(T)$. We call the elements in $\Euler_{\tmn}(T)$ \emph{points}.  

Let $U$ denote the ordered list obtained by concatenating $\Euler_{\tmn}(T_{S})$ for all clusters $S$, in an arbitrary order. Moreover, we append some \emph{satellite points} to the end of $U$: one for each satellite cluster $\hat{S}$. Then all our tables are indexed by $U \times U$. For each vertex $w\in V$, let $V_{\ppt}(w)\subseteq U$ be the set of all points corresponding to $w$ (a vertex $w$ may appear in multiple Steiner trees, and thus have multiple principal occurrences).

\paragraph{Tables.} We will construct a \emph{global table} $M_{\glo}$ and for each cluster $S$, a \emph{cluster table} $M_{S}$. In fact, there is another dimension corresponding to link lengths. Since there are at most $O(\log n)$ distinct link lengths, one can equivalently maintain a separate collection of the above tables for each link length. We omit this dimension for simplicity, and we will not talk about link lengths in what follows.

We only access these tables by orthogonal range counting queries. Thus these tables can be easily implemented using textbook data structures with query time $\Otilde(1)$ and space nearly linear to the number of non-zero entries.

\paragraph{Adding Entries into Tables.} \ 

\medskip

\noindent\underline{For Type-\ref{Item:Type1Link} Links.} For each original edge $(u,v)\in E$, add $1$ for all point pairs $(u_{\ppt},\vpt)\in V_{\ppt}(u)\times V_{\ppt}(v)$ in the global table $M_{\glo}$.

\medskip

\noindent\underline{For Type-\ref{Item:Type2Link} Links.} For each cluster $S$, each levels $i_{1}\geq i_{2}>i:=\level(S)$, if $S$ is both nearby $A_{i_{1}}$-sparse and nearby $A_{i_{2}}$-sparse, then for all vertex pairs $(u,v)\in \Lambda_{i_{1}}(S)\times \Lambda_{i_{2}}(S)$ and all point pairs $(\upt,\vpt)\in V_{\ppt}(u)\times V_{\ppt}(v)$, add $1$ to the entry indexed by $(\upt,\vpt)$ in both the global table $M_{\glo}$ and cluster table $M_{S}$.

\medskip

\noindent\underline{For Type-\ref{Item:Type3Link} Links.} For each cluster $S$, each levels $i_{1}>i_{2}>i:=\level(S)$, if $S$ is both nearby $A_{i_{1}}$-sparse but nearby $A_{i_{2}}$-dense, then let $\hat{S}_{\ppt}$ be the satellite point of the satellite cluster $\satellite_{i_{2},\alpha_{2}}(S)$, and for each vertex $u\in \Lambda_{i_{1}}(S)$ and each point $\upt\in V_{\ppt}(u)$, add $1$ to the entry indexed by $(u_{\ppt},\hat{S}_{\ppt})$ in both the global table $M_{\glo}$ and cluster table $M_{S}$.

\medskip

\noindent\underline{For Type-\ref{Item:Type7Link} Links.} For each vertex $v$, add $1$ to each entry $(\vpt,\vpt')\in V_{\ppt}(v)\times V_{\ppt}(v)$ in the global table $M_{\glo}$.

\begin{lemma}
\label{lemma:OracleSpaceTable}
All counting tables together take space $\Otilde(f^{2})\cdot n^{1+O(1/k)}$.
\end{lemma}
\begin{proof}
For each vertex $v$, its corresponding set of points $V_{\ppt}(v)$ has size $\lambda_{\ppt}:=i_{\max}\cdot\alpha_{\max}\cdot \omega = \Otilde(1)\cdot n^{O(1/k)}$. Recall that we replaced the input graph by a spanner at the beginning of this section, so the number of original edges is $\Otilde(f^{1-1/k}\cdot n^{1+1/k})$. Therefore, the number of non-zero entries contributed by the type-\ref{Item:Type1Link} and type-\ref{Item:Type7Link} links are at most $\Otilde(f)\cdot n^{1+O(1/k)}$.

The number of clusters is $i_{\max}\cdot \alpha_{\max}\cdot \omega\cdot n = \Otilde(1)\cdot n^{1+O(1/k)}$. Each of them will contribute $(i_{\max}\cdot \rho_{\heavy}\cdot \lambda_{\ppt})^{2} = \Otilde(f^{2})\cdot n^{O(1/k)}$ non-zero entries when dealing when type-\ref{Item:Type2Link} and type-\ref{Item:Type3Link} links (note that $\rho_{\heavy}$ comes from that $S$ is nearby $A_{i_{1}}$-sparse). 
\end{proof}

\paragraph{Update Algorithm Part 1: Detecting Shortcut-Based Links.} 

Detecting shortcut-based links between two nodes $u_{\nd},v_{\nd}\in Q$ can be done by range counting queries. For example, for two subtree node $u_{\nd}$ and $v_{\nd}$. Recall that we have assumed that each residual subtree is ``represented'' by one residual interval. Thus, by querying these two intervals, we can obtain the number of links connecting $u_{\nd}$ and $v_{\nd}$ in a table.

We note that we need to perform range counting queries on the global table and the cluster tables for each affected cluster, because we need to subtract the contribution from the invalid shortcuts provided by affected clusters.

\begin{lemma}
\label{lemma:OracleUpdateTime1}
The running time for detecting shortcut-based links is $\Otilde(f^{5})\cdot n^{O(1/k)}$.
\end{lemma}
\begin{proof}
Recall from \Cref{lemma:NumAffectedCluster} that the number of affected clusters is $\lambda_{\affS}:=\Otilde(f)\cdot n^{O(1/k)}$. Recall from \Cref{lemma:NumQNodes} that the number of nodes in $Q$ is $\lambda_{Q}:=\Otilde(f^{2})\cdot n^{O(1/k)}$. %

The overall running time is then $\Otilde(\lambda_{\affS}\cdot \lambda_{Q}^{2}) = \Otilde(f^{5})\cdot n^{O(1/k)}$. Roughly speaking, we want to detect shortcut-based links for each pair of nodes, so we have the term $\lambda_{Q}^{2}$. The term $\lambda_{\affS}$ is from subtracting the contributions of affected clusters.
\end{proof}

\subsubsection{Structures for Expansion-Based Links and Update Algorithm Part 2}

Type-\ref{Item:Type4Link} links can be detected trivially using the Euler Tour trick, so we omit discussing it. The key to detecting type-\ref{Item:Type5Link} and type-\ref{Item:Type6Link} links is to determine the adjacent-cluster $A_{i}$-denseness of a singleton node $v_{\nd}$ or a subtree node $\tau_{\nd}$. Take singleton node $v_{\nd}$ of a vertex $v$ for example. Recall \Cref{def:OracleAdjClusterDenseness} of adjacent-cluster denseness. The only update-phase sensitive information is that the witness cluster $S$ must be unaffected. This small issue can be easily overcame: because the number of affected cluster is at most $\lambda_{\affS}:=\Otilde(f)\cdot n^{O(1/k)}$, we simply store up to $\lambda_{\affS}+1$ candidate witness clusters for this vertex $v$. Then in the update phase, we will find one if there exists an unaffected one. To check the adjacent-cluster denseness for a residual subtree, use the Euler Tour trick again and work with canonical intervals instead.

\begin{lemma}
\label{lemma:OracleSpaceExpLinks}
The data structures for detecting expansion-based links take $\Otilde(f)\cdot n^{1+O(1/k)}$ space.
\end{lemma}

The space bound in \Cref{lemma:OracleSpaceExpLinks} basically comes from: (1) the number of vertices and canonical intervals is bounded by $\Otilde(n^{1+O(1/k)})$, and (2) for each of them and each level, we store a list of $\lambda_{\affS}+1 = \Otilde(f)\cdot n^{O(1/k)}$ candidate clusters.

\begin{lemma}
\label{lemma:OracleUpdateTime2}
In the update phase, detecting all expansion-based links takes $\Otilde(f^{3})\cdot n^{O(1/k)}$ time.
\end{lemma}

The running time bound in \Cref{lemma:OracleUpdateTime2} basically comes from the following. 

First, we need to check the adjacent-cluster denseness for $\lambda_{Q}=\Otilde(f^{2})\cdot n^{O(1/k)}$ nodes and $\Otilde(1)$ levels. Each check needs to scan the candidate list, which takes $\lambda_{\affS}+1 = \Otilde(f)\cdot n^{O(1/k)}$ time. Therefore, the total time for checking adjacent-cluster denseness is $\Otilde(f^{3})\cdot n^{O(1/k)}$.

Next, after passing the adjacent-cluster denseness check for a node $q_{\nd}$ (w.r.t. some level), we still need to spend additional time to add the expansion-based links. 

More concretely, for type \ref{Item:Type5Link}, we need to add a link between $q_{\nd}$ and $S_{1,\nd}$ for every affected cluster $S_{1}\in \mathcal{N}_{i_{1},\alpha_{1}}$ that intersects $S_{2}:=\cluster(q_{\nd})$. To do this, we can scan through every affected cluster $S_{1}\in {\cal N}_{i_{1},\alpha_{1}}$ and check if $S_{1}$ intersects $S_{2}$. In fact, in order to check whether $S_{1}$ intersects $S_{2}$ in $\Otilde(1)$ time, we will preprocess an \emph{intersection table} of clusters. The intersection table takes only $\Otilde(n^{1+O(1/k)})$ space since each vertex will contribute at most $n^{O(1/k)}$ pairs of intersecting clusters. Therefore, adding type-\ref{Item:Type5Link} links for a node $q_{\nd}$ needs $\Otilde(f\cdot n^{O(1/k)})$ time (basically the number of affected clusters), and the total time for adding type-\ref{Item:Type5Link} links is $\Otilde(f^{3}\cdot n^{O(1/k)})$.

For type \ref{Item:Type6Link}, we need to add a link between $q_{\nd}$ and the satellite node $\hat{S}_{2,\nd}$ for every affected satellite cluster $\hat{S}_{2}\in \hat{\cal N}_{i_{2},\alpha_{2}}$ containing the vertex $\pivot(S_{2})$. Since the number of satellite clusters containing $\pivot(S_{2})$ is $n^{O(1/k)}$ (i.e., the width of satellite neighborhood covers), the total time for adding type-\ref{Item:Type6Link} is $\Otilde(f^{2}\cdot n^{O(1/k)})$.

\subsubsection{Update Algorithm Part 3}
\label{sect:OracleUpdateTime3}

After computing the query graph $Q$ explicitly, the update algorithm will compute the all-pair distances of $Q$. Recall that $Q$ has $\Otilde(f^{2})\cdot n^{O(1/k)}$ vertices and thus $\Otilde(f^{4})\cdot n^{O(1/k)}$ edges. We can compute exact all-pair distances of $Q$ using textbook algorithm in $\Otilde(f^{6})\cdot n^{O(1/k)}$ time, or compute $3$-approximate all-pair distances in $\tilde{O}(f^{4})\cdot n^{O(1/k)}$ time using results in e.g., \cite{cohen2001all}.

\subsubsection{Proxy Lists}

The proxy lists required by the query algorithm can be constructed directly, which take space $\Otilde(n\cdot f)$ in total.

\subsection{Proof of \Cref{thm:OracleMain}}
\label{sect:ProofOracleMain}

The approximation and correctness is given by \Cref{lemma:UpperBoundSolution} (the upper-bound direction) and \Cref{lemma:OracleSolLowerBound} (the lower-bound direction) with one more multiplicative factor $2k-1$ from the vertex-fault-tolerant spanner (as discussed in the beginning of this section).

The construction time is clearly polynomial, if we use the polynomial-time version of LC vertex-expnader decomposition \Cref{thm:PolyED}. The space complexity is given by \Cref{lemma:OracleSpaceTable} and \Cref{lemma:OracleSpaceExpLinks}. The update time is given by \Cref{lemma:OracleUpdateTime1}, \Cref{lemma:OracleUpdateTime2} and \Cref{sect:OracleUpdateTime3}. The query time is clearly $\Otilde(f)$ from the description of \Cref{algo:QueryAlgorithm}.

\section{Vertex-Failure Approximate Distance Labeling Schemes}

In this section, we present our vertex-failure distance (and connectivity) labeling schemes. 

In \Cref{sect:LabelBuildingBlocks}, we discuss some building blocks that are slightly different from those in the oracle section. In \Cref{sect:LabelShortcutGraph}, we define the shortcut graph $H$. In \Cref{sect:LabelQueryGraph}, we define the query graph $Q$, but this time it is only for analysis.
In \Cref{sect:LabelSparsifiedQ}, we define the sparsified query graph $\Qtilde$, and describe a query algorithm that answers the query by constructing $\Qtilde$, with some interfaces deferred to \Cref{sect:SolveDownDense} and \Cref{sect:Interfaces}.

\begin{theorem}[Vertex-Failure Distance Labeling Schemes]
\label{thm:DistLabelMain}
Let $G$ be a graph with integral, nonnegative, polynomially bounded edge and vertex lengths $\ell$. Given parameters $f\geq 1$ and $k\geq 1$, there is an algorithm that constructs a label $L(v)$ for each vertex $v$ and supports the following.
\begin{itemize}
\item(Query) Given two vertices $x,y$ and a set of failed vertices $F\subseteq V$ of size at most $f$, it computes a value $d$ such that $\dist_{G\setminus F}(x,y)\leq d\leq s\cdot \dist_{G\setminus F}(x,y)$ by reading only the labels $\{L(v)\mid v\in F\cup\{x,y\}\}$.
\end{itemize}
The approximation is $s=O(k^{6})$ for a polynomial-time construction, and $s=O(k^{5})$ for an exponential-time construction. Each label has $f^{3}\cdot n^{1/k}\cdot \log^{O(k)}n$ bits. Moreover, the query time is polynomial in the total size of accessed label, and the algorithm is randomized which can answer each query correctly with probability $1-n^{-10}$.
\end{theorem}

The approximation of \Cref{thm:DistLabelMain} is proven in \Cref{sect:DistLabelCorrectness}. The label size is shown in \Cref{sect:DistLabelSize}.

\begin{theorem}[Vertex-Failure Connectivity Labeling Schemes]
\label{thm:ConnLabel}
Let $G$ be a graph. Given parameters $f\geq 1$, there is an algorithm that constructs a label $L(v)$ for each vertex $v$ and supports the following.
\begin{itemize}
\item (Query) Given two vertices $x,y$ and a set of failed vertices $F\subseteq V$ of size at most $f$, it determines whether $x$ and $y$ is connected in $G\setminus F$, by reading only the labels $\{L(v)\mid v\in F\cup\{x,y\}\}$.
\end{itemize}
The labels can be constructed in polynomial time. Each label has $\Otilde(f^{2})$ bits. Moreover, the query time is polynomial in the total size of accessed label, and the algorithm is deterministic.
\end{theorem}
We finally discuss how to achieve \Cref{thm:ConnLabel} in \Cref{sect:ConnLabelSize}.

\subsection{Building Blocks}
\label{sect:LabelBuildingBlocks}

For labeling schemes, we still use the same based length scale $\{h_{\alpha}\mid 1\leq \alpha\leq \alpha_{\max}\}$ and the parameter $\rho_{\heavy} = f\cdot O(\log n/\phi)+1$.

\subsubsection{Expander Hierarchy and Expanding Graphs} we still use the same LC vertex expander hierarchy in \Cref{sect:OracleHierarchy} and the expanding graphs $G_{i,\alpha}$ in \Cref{sect:OracleExpandingGraphs}, except for the following small modifications.
\begin{itemize}
\item We no longer need the second decompositions $\hat{C}$ when building the hierarchy, since the idea of satellite cluster is not needed in this section. Formally, set all moving cuts $\hat{C}_{i,\alpha}$ to be zero function.
\item Below, we will change the length slack of neighborhood covers to $s_{\nc} = k^{2}$. So this change propagates to $h_{\ed,\alpha}$, which now becomes $h_{\ed,\alpha} = \Theta(k^{2}\cdot h_{\alpha})$.
\end{itemize}

\subsubsection{The Tree-Structured Hierarchy of Neighborhood Covers} The major change of the building blocks is about the neighborhood covers. This time we enforce the clusters across levels form a tree structure. Gaining this stronger structure will make the algorithm conceptually simpler and bypass some technical issues, at a cost of a worse tradeoff between $s_{\nc}$ and the final width. The new neighborhood covers are constructed as follows.

Define parameters
\[
s_{\nc} = k^{2},\qquad h_{\cov,\alpha} = 2\cdot h_{\alpha},\qquad h_{\diam,\alpha} = h_{\cov,\alpha}\cdot s_{\nc} = \Theta(k^{2} h_{\alpha})
\]

For each level $i$ from $i_{\max}$ to $0$, 
each level-$(i+1)$ cluster $\hat{S}$ with scale $\hat{\alpha}:=\scale(\hat{S})$ (when $i=i_{\max}$, let $\hat{S} = V$ be the whole vertex set and $\hat{\alpha} = \alpha_{\max}$), and
each length scale $\alpha\in[1,\alphahat]$, 
\begin{itemize}
\item construct a level-$i$ cover ${\cal N}_{i,\alpha,\hat{S}}$ in the subgraph $G_{i,\alpha}[\hat{S}]$ with covering radius $h_{\cov,\alpha}$, diameter $h_{\diam,\alpha}$ and width $\omega' = O(n^{1/s_{\nc}}\cdot s_{\nc})$, %
\end{itemize}
For each level-$i$ cluster, we still define $V_{\tmn}(S) = S\cap A_{i}$.

The following observation bounds the total width of all the neighborhood covers.

\begin{observation}
\label{ob:LabelWidth}
Each vertex can appear in at most
\[
\omega = n^{O(1/k)}\cdot \log^{O(k)} n
\]
clusters. We call this $\omega$ the \emph{total width} of the neighborhood cover hierarchy.
\end{observation}
\begin{proof}
Observe that each level contributes a multiplicative factor of $\alpha_{\max}\cdot \omega'$, so the final width is $(\alpha_{\max}\cdot \omega')^{i_{\max}} = (O(\log n)\cdot n^{1/k^{2}}\cdot k^{2})^{O(k)} = n^{O(1/k)}\cdot \log^{O(k)}n$.
\end{proof}

\paragraph{Ancestor-Descendant Relations among Clusters.} For each cluster $S\in {\cal N}_{i,\alpha,\hat{S}}$, the \emph{parent} of $S$ is $\hat{S}$, and $S$ is a child of $\hat{S}$. This naturally defines the ancestor-descendant relations among clusters. When $\hat{S}$ is an ancestor of $S$, we write $\hat{S}\anceq S$ (use $\anc$ for strict ancestor). Similarly, we write $\check{S}\desceq S$ (or $\check{S}\desc S$) if $\check{S}$ is a descendant (or strict descendant) of $S$.

\paragraph{The New Surrounding Set $\Lambda(S)$ and Leader Sets $B_{i'}(S)$.} For each cluster $S$ with level $i:=\level(S)$ and length scale $\alpha:=\scale(S)$, define
\[
\Lambda(S):=\bigcup_{\hat{S}\anc S}\{v\in V_{\tmn}(\hat{S})\mid v\in S \text{ or there is an edge in $G$ connecting $v$ and $S$ of length $\leq h_{\alpha}$}\}.
\]
and naturally, for each higher level $i'>i$, $\Lambda_{i'}(S) = \Lambda(S)\cap V_{\tmn}(\hat{S}_{i'})$, where $\hat{S}_{i'}$ is the level-$i'$ ancestor of $S$.

We have the same definition of nearby denseness as in \Cref{def:nearbyDensenes}. For each level $i'>i$, we additional define a \emph{leader set} $B_{i'}(S)\subseteq \Lambda_{i'}(S)$ as follows. If $S$ is nearby $A_{i'}$-sparse, then set $B_{i'}(S) = \Lambda_{i'}(S)$. If $S$ is nearby $A_{i'}$-dense, then let $B_{i'}(S)$ be an arbitrary subset of size $|B_{i'}(S)| = \min\{f+1,|\Lambda_{i'}(S)|\}$. 

\begin{remark}
\label{remark:benefits}
The key benefits we obtain from the new hierarchy of neighborhood covers are as follows. 

First, it is compatible with the decomposition process of the witness path in \Cref{sect:OraclePathDecompose}, so we only need to add shortcut-based link between (the terminal sets) of ancestor-descendant cluster pairs. 

Second, by the tree structure, once a cluster  is affected in the query phase, all its ancestor will also be affected. This means, for a ``maximal'' unaffected cluster $S$, all vertices in $\Lambda(S)$ must belong to (the terminal set) of an affected ancestral cluster. This enables us to replace the ``satellite cluster'' in the oracle section with the leader set $B_{i'}(S)$ for representing $\Lambda_{i'}(S)$. This is because $B_{i'}(S)$ will contain a non-failed vertex $v$ in $\Lambda_{i'}(S)$, and $v$ indeed shows up in the query graph (as a singleton node or part of a subtree node). Then $v$ will connect to all vertices in $\Lambda_{i'}(S)$ via expansion-based links. In contrast, in the oracle section, even though we can define such a leader set, but the vertex $v$ may not show up in the query graph.

\end{remark}

\subsubsection{Low-Degree Low-Diameters Steiner Trees}

This part is analogous to \Cref{sect:SteinerTrees}. We construct a collection of Steiner trees $\{T_{S}\mid \text{all clusters }S\}$ in $G$ that satisfies the following.
\begin{itemize}
\item Each tree $T_{S}$ spans the terminals $V_{\tmn}(S)$ of $S$
\item Each tree $T_{S}$ has diameter $h_{\tree,\alpha} = 4h_{\ed,\alpha}\cdot s_{\ed}\cdot s_{\tree} = \Theta(k^{4}\cdot h_{\alpha})$. Note that $h_{\tree,\alpha}$ changed because of the new $s_{\nc}$.
\item Each tree has maximum degree $\Delta = \Otilde(n^{1/s_{\tree}}\cdot \omega/\phi) = n^{O(1/k)}\cdot \log^{O(k)}n$.
\item Each vertex appears in at most $\gamma_{\tree} = \Otilde(n^{1/s_{\tree}}\cdot \omega/\phi) = n^{O(1/k)}\cdot \log^{O(k)}n$ Steiner trees.
\end{itemize}

\subsection{The Shortcut Graph}
\label{sect:LabelShortcutGraph}

In this section, we define the \emph{shortcut graph} $H$ which is purely for analysis. Intuitively, this graph contains the shortcuts behind the shortcut-based links in the query graph. Compared to the oracle section, this time we have a cleaner construction of shortcuts due to the benefits explained in \Cref{remark:benefits}. To distinguish between input graph $G$, shortcut graph $H$, and query graph $Q$, we call vertices and edges in $H$ \emph{points} and \emph{shortcuts}.

\paragraph{Points of $H$.} For each cluster $S$ and each vertex $v\in V_{\tmn}(S)$, create a \emph{point} with respect to $v$ and $S$, denoted by $\point(v,S)$. Define the \emph{core} of $S$ to be the point set \[
\core(S):=\{\point(v,S)\mid v\in V_{\tmn}(S)\}.
\]
The node set of $H$ is then 
\[
V(H):=\bigsqcup_{\text{cluster }S}\core(S).
\]
For each point $\vpt = \point(v,S)$, let $\pi_{\vtx}(\vpt) = v$ denote its corresponding vertex. 

\paragraph{Shortcuts of $H$.} We add the following two types of shortcuts.

\begin{enumerate}
\item\label{Type:Shortcut1} For each cluster $S_{2}$ and each ancestor cluster $S_{1}\succeq S_{2}$, if there is an original edge $e=(u,v)$ with $u\in V_{\tmn}(S_{1})$, $v\in V_{\tmn}(S_{2})$ and length $\ell_{G}(e)\leq h_{\alpha_{2}}$ (for $\alpha_{2} = \scale(S_{2})$), then add a shortcut between $\point(u,S_{1})$ and $\point(v,S_{2})$ with length $\ell_{G}(e)$.

\item\label{Type:Shortcut2} For each triple of clusters $S_{1},S_{2},S$ s.t. $S_{1}\succ S_{2}\succ S$, for each unordered pair of distinct vertices $(u,v)\in B_{i_{1}}(S)\times B_{i_{2}}(S)$ where $i_{1}:=\level(S_{1})$ and $i_{2}:=\level(S_{2})$, add a shortcut between $\point(u,S_{1})$ and $\point(v,S_{2})$ with length $h_{\diam,\alpha} + 2h_{\alpha}$ where $\alpha:=\scale(S)$. This shortcut is \emph{owned} by the cluster $S$.
\end{enumerate}

For convenience, we may also use $H$ to denote the set of all shortcuts. For each shortcut $(\point(u,S_{1}),\point(v,S_{2}))$ connecting two different levels, i.e. $S_{1}\succ S_{2}$, we call $\point(u,S_{1})$ its \emph{upper endpoint} and $\point(v,S_{2})$ its \emph{lower endpoint}.

\begin{definition}[The counting function $\delta$ and length function $\beta$] 
\label{def:CountLengthFunctions}
For a subset of shortcuts $H'\subseteq H$, two sets of points $U_{\ppt}$ and $V_{\ppt}$, we define
\begin{align*}
\delta(U_{\ppt},V_{\ppt},H')&:= \text{the number of shortcuts in $H'$ with one endpoint in $U_{\ppt}$ and the other in $V_{\ppt}$},\\
\beta(U_{\ppt},V_{\ppt},H')&:= \text{the minimum length over shortcuts in $H'$ with}\\ &~~~~~\text{one endpoint in $U_{\ppt}$ and the other in $V_{\ppt}$}.
\end{align*}
Note that if a shortcut $(u_{\ppt},v_{\ppt})\in H'$ has $u_{\ppt},v_{\ppt}\in U_{\ppt}\cap V_{\ppt}$, it will be counted twice in $\delta(U_{\ppt}, V_{\ppt}, H')$.
\end{definition}

\subsection{The Query Graph}
\label{sect:LabelQueryGraph}

Now we define our query graph $Q$ in the query phase with input $(x,y,F)$ (recall that we merge the update and query phase in the labeling scheme model). Note that this query graph $Q$ is also purely for analysis. The graph we will actually construct and use to answer queries is the \emph{sparsified query graph} $\Qtilde$, introduced later in \Cref{sect:LabelSparsifiedQ}.

\subsubsection{Affected Clusters, Residual Subtrees, and Affected Shortcuts}

Let $F_{\ppt}$ collect all points corresponding to failed vertices, called \emph{failed points}.

The same as \Cref{sect:OracleAffected}, the affected clusters are still clusters hit by failed vertices, and we generate residual subtrees for affected clusters in the same way. For each residual subtree $\tau$ of an affected cluster $S$, similarly we define its core point set to be
\[
\core(\tau):=\{\point(v,S)\mid v\in V_{\tmn}(\tau)\}.
\]

For a shortcut $e_{\sc}\in H$, it is affected if it is owned by an affected cluster; otherwise it is unaffected. It is \emph{non-failed} if both of its endpoints are not failed points; otherwise it is failed. Note that these two terms ``affected'' and ``non-failed'' of shortcuts are independent. An affected shortcut could be non-failed, while a failed shortcut could be unaffected. Below are some useful observations about shortcuts.

Let $\lambda_{\affS}$ denote the number of affected clusters. Let $\lambda_{\nb}$ upper bound $\sum_{i'}|B_{i'}(S)|$ for an arbitrary cluster $S$. We clearly have
\[
\lambda_{\affS} \leq f\cdot \omega = f\cdot n^{O(1/k)}\cdot \log^{O(k)}n,\qquad \lambda_{\nb} \leq i_{\max}\cdot \rho_{\heavy} = \Otilde(f)\cdot n^{O(1/k)}.
\]

\begin{observation}
\label{ob:AffectedSCEndpoints}
There are at most $\lambda_{\affS}\cdot \lambda_{\nb}$ distinct points incident to affected shortcuts.
\end{observation}
\begin{proof}
Each shortcut owned by a cluster $S$ must have both endpoints in the leader sets $B_{i'}(S)$.
\end{proof}

\begin{observation}
\label{ob:ClusterOwnShortcut}
Each cluster $S$ owns at most $\lambda^{2}_{\nb}$ shortcuts. 
\end{observation}

\subsubsection{Nodes in the Query Graph $Q$} Similar to the oracle setting, the query graph $Q$ has three types of nodes.
\begin{itemize}
\item(Singleton Nodes) For each affected cluster $S$ and each vertex $v\in V_{\tmn}(S)$ such that $v$ is inside a light residual subtree of $S$, create a singleton node $\node(v,S)$. The node length is $0$
\item(Subtree Nodes) For each affected cluster $S$ and each heavy residual subtree of $S$, create a subtree node $\node(\tau)$. The node length is $h_{\tree,\alpha} = \Theta(s_{\nc}\cdot s_{\ed}\cdot s_{\tree}\cdot h_{\alpha})$ for $\alpha = \scale(S)$.
\item(Cluster Nodes) For each affected cluster $S$, create a cluster node $\node(S)$. The node length is $h_{\ed,\alpha}\cdot s_{\ed} + 2h_{\diam,\alpha} + 4h_{\alpha} + h_{\tree,\alpha} = \Theta(s_{\nc}\cdot s_{\ed}\cdot s_{\tree}\cdot h_{\alpha})$. Note that the cluster node length has an extra $h_{\tree,\alpha}$ additive factor compared to the oracle section, because of some technical details in the proof of \Cref{lemma:LabelSolLowerBound}.
\end{itemize}
For each node $q_{\nd}$ defined above, set $\cluster(q_{\nd}):=S$, and let $\level(q_{\nd})$ and $\scale(q_{\nd})$ inherit the level and length scale of $S$, respectively. Again singleton nodes and subtree nodes are \emph{concrete nodes}, and cluster nodes are \emph{virtual nodes}.
Below are some notations.

\paragraph{Mapping Between Points and Nodes.} For each point $\vpt = \point(v,S)$, if $v$ is in a light residual subtree of $S$, then $\pi_{\nd}(\vpt) = \{\node(v,S)\}$; if $v$ is in a heavy residual subtree $\tau$ of $S$, then $\pi_{\nd}(\vpt) = \node(\tau)$. 

For each concrete node $q_{\nd}$, if $q_{\nd}=\node(v,S)$ is a singleton node, then $\pi_{\ppt}(q_{\nd}) = \{\point(v,S)\}$; if $q_{\nd} = \node(\tau)$ is a subtree node, then $\pi_{\ppt}(q_{\nd}) = \core(\tau)$.

\paragraph{Mapping Between Vertices and Concrete Nodes.} For each node $q_{\nd}$, if $q_{\nd} = \node(v,S)$ is a singleton node, then $V(q_{\nd}) = \{v\}$; if $q_{\nd} = \node(\tau)$ is a subtree node, then $V(q_{\nd}) = V_{\tmn}(\tau)$.

For each vertex $v$, $V_{\nd}(v) = \{q_{\nd}\mid v\in V(q_{\nd})\}$.

\subsubsection{Links in the Query Graph $Q$} 

\paragraph{Shortcut-Based Links.} The shortcut-based links defined below can be viewed as mappings of the $H$-shortcuts into the query graph $Q$.

\begin{tcolorbox}[
  colback=black!3,
  colframe=black,
  boxrule=1.5pt,
  arc=0pt,
  left=6pt,
  right=6pt,
  top=4pt,
  bottom=4pt
]
\setcounter{Qtypectr}{0}
\refstepcounter{Qtypectr}\label{Type:Qlink1}
\noindent\underline{\textcolor{DarkRed}{Type \arabic{Qtypectr}}.} For each non-failed and unaffected $H$-shortcut $e_{\sc}:=(\upt, \vpt)$, add a link $e_{\lk}$ between nodes $\pi_{\nd}(\upt)$ and $\pi_{\nd}(\vpt)$ with length $\ell(e_{\lk})= \ell(e_{\sc})$.
\end{tcolorbox}

\paragraph{Expansion-Based Links.} We define two types of expansion-based links: type \ref{Type:Qlink2} and type \ref{Type:Qlink3}, which are analogous to type \ref{Item:Type4Link} and type \ref{Item:Type5Link} in the oracle section respectively.

\begin{tcolorbox}[
  colback=black!3,
  colframe=black,
  boxrule=1.5pt,
  arc=0pt,
  left=6pt,
  right=6pt,
  top=4pt,
  bottom=4pt
]
\refstepcounter{Qtypectr}\label{Type:Qlink2}
\noindent\underline{\textcolor{DarkRed}{Type \arabic{Qtypectr}}.} For each affected cluster $S$, and each heavy residual subtree $\tau$ of $S$, add a link between $\node(\tau)$ and $\node(S)$ with length $0$.
\end{tcolorbox}

\medskip

The following definition of adjacent-cluster denseness is almost identical to its counterpart in the oracle section, except that it is now formulated for point sets.

\begin{definition}[Adjacent-Cluster Denseness]
For each cluster $S$, each level $i'\geq i:=\level(S)$, a point $\vpt\in \core(S)$ is \emph{adjacent-cluster $A_{i'}$-dense} if 
\begin{itemize}
\item there exists an unaffected cluster $\check{S}\desc S$ s.t. $\vpt\in \Lambda(\check{S})$ and $\check{S}$ is nearby $A_{i'}$-dense.
\end{itemize}
A residual subtree $\tau$ of $S$ (or equivalently, the point set $\core(\tau)$) is adjacent-cluster $A_{i'}$-dense if there exists a point $\vpt\in\core(\tau)$ that is adjacent-cluster $A_{i'}$-dense.
\end{definition}

\begin{tcolorbox}[
  colback=black!3,
  colframe=black,
  boxrule=1.5pt,
  arc=0pt,
  left=6pt,
  right=6pt,
  top=4pt,
  bottom=4pt
]
\refstepcounter{Qtypectr}\label{Type:Qlink3}
\noindent\underline{\textcolor{DarkRed}{Type \arabic{Qtypectr}}.} For each affected cluster $S$, and level $i'\geq i:=\level(S)$, let $\hat{S}\anceq S$ be the level-$i'$ ancestor.
\begin{itemize}
\item For each singleton node $v_{\nd}$ with $\cluster(v_{\nd}) = S$, if the point $\vpt:=\pi_{\ppt}(v_{\nd})$ is adjacent-cluster $A_{i'}$-dense, then add a link between $v_{\nd}$ and $\node(\hat{S})$ with length $0$.
\item For each heavy residual subtree $\tau$ of $S$, if $\tau$ is adjacent-cluster $A_{i'}$-dense, then add a link between $\node(\tau)$ and $\node(\hat{S})$ with length 0.
\end{itemize}
\end{tcolorbox}

\subsubsection{Approximation of the Query Graph}

The following lemma essentially establishes the upper-bound direction for query answering with the query graph $Q$, assuming that $Q$ is given explicitly. It basically follows the strategy of proving \Cref{lemma:FindingProxy}, but stops at a middle step \Cref{lemma:SingleMappingPathLength}. In fact, after simplifying $P_{\qry}$ as in the proof of \Cref{lemma:FindingProxy}, it will certify that $\dist_{Q}(x_{\nd},y_{\nd})\leq O(s_{\nc}\cdot s_{\ed}\cdot s_{\tree}\cdot i_{\max})\cdot \dist_{G\setminus F}(x,y)$

\begin{lemma}
\label{lemma:FindQPath}
There exists a (not necessarily simple) path $P_{\qry}$ in $Q$ connecting some $x_{\nd}\in V_{\nd}(x)$ and $y_{\nd}\in V_{\nd}(y)$, satisfying the following.
\begin{enumerate}
\item\label{Item:QPathProperty1} Each subpath whose endpoints are concrete and all internal nodes are virtual is in the form either
\begin{enumerate}[label=(\alph*)]
\item\label{Item:SubpathForma} $(w_{\nd},z_{\nd})$, which is a single shortcut-based link, or
\item\label{Item:SubpathFromb} $(w_{\nd},S_{\nd})\circ(S_{\nd},z_{\nd})$, i.e., two expansion-based links sharing a common cluster node $S_{\nd}$. Moreover, either $\cluster(S_{\nd})=\cluster(w_{\nd})\succeq \cluster(z_{\nd})$ or $\cluster(S_{\nd})=\cluster(z_{\nd})\succeq \cluster(w_{\nd})$.
\end{enumerate}
\item\label{Item:QPathProperty2} For each occurrence of a subtree node $\tau_{\nd}$ on $P$, this occurrence serves as an endpoint of a form-\ref{Item:SubpathFromb} subpath.
\item\label{Item:QPathProperty3} The total length of shortcut-based links on $P_{\qry}$ is at most $O(s_{\nc})\cdot \dist_{G\setminus F}(u,v)$.
\item\label{Item:QPathProperty4} Each form-\ref{Item:SubpathFromb} subpath has length at most $O(1)\cdot \ell_{}(S_{\nd})$.
\item\label{Item:QPathProperty5} The total length of cluster nodes in 
\begin{align*}
\{\node(S)\mid S = \cluster(q_{\nd})\text{ for some node $q_{\nd}$ on $P_{\qry}$}\}
\end{align*}
is at most $O(s_{\nc}\cdot s_{\ed}\cdot s_{\tree}\cdot i_{\max})\cdot \dist_{G\setminus F}(x,y)$
\end{enumerate}
\end{lemma}
\begin{proof}
Fix an arbitrary shortest path $P$ connecting $u$ and $v$ in $G\setminus F$. Following the proof strategy of \Cref{lemma:FindingProxy}, we will first recursively decompose $P$ into leaf segments. 

\medskip

\noindent{\underline{The Decomposition Process}.} The decomposition process is almost the same as that in the proof of \Cref{lemma:FindingProxy}, except for the following changes.
\begin{itemize}
\item We now start from the top level $i_{\max}$ of the hierarchy, i.e. the initial segment has $i_{\sigma}:=i_{\max}$.

\item For each segment $\sigma$, let $\hat{\sigma}$ be its parent and let $\hat{S}:=S_{\hat{\sigma}}$ be the corresponding cluster of $\hat{\sigma}$. Then the corresponding cluster $S_{\sigma}$ of $\sigma$ should be selected from the neighborhood cover ${\cal N}_{i_{\sigma},\alpha_{\sigma},\hat{S}}$, i.e. the cover constructed for the level-$(i_{\sigma}+1)$ cluster $\hat{S}$. In particular, the initial segment should select its corresponding cluster from ${\cal N}_{i_{\max},\alpha_{\sigma},V}$.
\end{itemize}

\medskip

\noindent\underline{Extra Properties of Segments.} In addition to Properties \ref{Item:SegProperty1} to \ref{remark:LeafSegClassify}, we have the following extra properties.
\begin{enumerate}[start=6]
\item For the old Property \ref{remark:LeafSegClassify} about the leaf segments, now additionally, every leaf segment must be a middle segment. To see this, observe that each segment $\sigma$ whose interior $\sigma^{\circ}$ contains $x$ or $y$ must be recursive, since we declare all cluster containing $x$ or $y$ to be affected. Hence, the interior $\sigma^{\circ}$ of a leaf segment $\sigma$ does not contains $x$ and $y$, i.e., $\sigma$ is a middle segment. We note that this extra property is not from the new neighborhood cover hierarchy, it is purely because we do not separate the update and query phases in the labeling scheme model. 
\end{enumerate}

\noindent\underline{Mapping Leaf Segments to $Q$.} Again, we will map each (middle) leaf segment $\sigma$ to a path $P_{\qry,\sigma}$ in $Q$, and then concatenating all $P_{\qry,\sigma}$ to obtain $P_{\qry}$. Now focus on a leaf segment $\sigma = P[w_{},z_{}]$. The first step is still define the endpoint nodes of $w$ and $z$. 

We take $w$ as an example. Let $\sigma_{w}$ be the unique segment that has $w$ as its breakpoint. Let $S_{w}$ be the corresponding cluster $S_{\sigma_{w}}$ of $\sigma_{w}$. Then we define
\begin{align*}
&w_{ \ppt}:=\point(w_{},S_{w}),&&w_{ \nd}:=\pi_{\nd}(w_{ \ppt}).
\end{align*}
Let $i_{w}$ and $\alpha_{w}$ be the level and length scale of $S_{w}$. Let $i_{z},\alpha_{z},S_{z},z_{\ppt},z_{\nd}$ be the analogous notations for the other endpoint $z$. Note that the definition of endpoint nodes is analogous to that in the oracle section, and we have \Cref{ob:LabelShortcutCondition} analogous to \Cref{ob:LevelScaleConsistent}.

\begin{observation}
\label{ob:LabelShortcutCondition}
When $\sigma$ has a corresponding cluster $S_{\sigma}$, we have the following.
\begin{itemize}
\item Both $S_{w}$ and $S_{z}$ are strict ancestors of $S_{\sigma}$, i.e. $S_{w},S_{z}\anc S_{\sigma}$.
\item Either $S_{w}\anceq S_{z}$ or $S_{z}\anceq S_{w}$.
\item $S_{w}$ and $S_{z}$ are affected, while $S_{\sigma}$ is unaffected.
\item Both endpoint vertices $w$ and $z$ are inside $\Lambda(S_{\sigma})$.
\end{itemize}
\end{observation}
\begin{proof}
The proof is essentially the same as \Cref{ob:LevelScaleConsistent}. Here we only emphasize some minor differences. For the first and second items, now we have the stronger ancestor-descendant guarantees among $S_{w},S_{z},S_{\sigma}$ because of the new hierarchy of neighborhood covers. This is also why the fourth item still holds under the new definition of $\Lambda(\cdot)$.
\end{proof}

We now define the path $P_{\qry,\sigma}$ which will connect $w_{\nd}$ and $z_{\nd}$. Again, we consider the same four cases in the proof of \Cref{lemma:SingleMappingPathLength}. The first three can be handled exactly as in the previous proof, so we present their constructions directly (with some simplifications). The key difference lies in Case 4.

\medskip

\noindent\underline{Case 1.} Suppose $P[w_{},z_{}]$ is a single edge $e=(w_{},z_{})$. Define 
\[
P_{\qry,\sigma}:=\underbrace{(w_{ \nd},z_{ \nd})}_{\text{shortcut-based link}}.
\]
This link is based on the shortcut $(w_{ \ppt},z_{ \ppt})$ corresponding to the edge $e$. Note that $S_{w}\anceq S_{z}$, $\ell_{G}(e)\leq \leng_{G}(\sigma)\leq h_{\alpha_{z}}$, so the shortcut exists according to rule \ref{Type:Shortcut1} of shortcuts. %

\medskip

\noindent\underline{Case 2.} Suppose $S_{\sigma}$ exists and it is both nearby $A_{i_{w}}$-sparse and nearby $A_{i_{z}}$-sparse. Again define 
\[
P_{\qry,\sigma}:=\underbrace{(w_{ \nd},z_{ \nd})}_{\text{shortcut-based link}}.
\]
This link is based on the shortcut $(w_{ \ppt},z_{ \ppt})$ owned by the unaffected cluster $S_{\sigma}$. The shortcut exists since $w_{ \nd}\in \Lambda_{i_{w}}(S_{\sigma}) = B_{i_{w}}(S_{\sigma})$ (here we exploit that $S_{\sigma}$ is nearby $A_{i_{w}}$-sparse) and similarly $z_{ \nd}\in \Lambda_{i_{z}}(S_{\sigma})=B_{i_{z}}(S_{\sigma})$. The link exists because the shortcut is non-failed (since $w_{ \ppt},z_{ \ppt}\notin F_{\ppt}$) and unaffected (since $S_{\sigma}$ is unaffected).

\medskip

\noindent\underline{Case 3.} Suppose $S_{\sigma}$ exists and it is nearby $A_{i_{w}}$-dense. Define
\[
P_{\qry,\sigma}:=\underbrace{(w_{ \nd},\node(S_{w}))}_{\text{Type-\ref{Type:Qlink3}}}\circ \underbrace{(\node(S_{w}),z_{ \nd})}_{\text{Type-\ref{Type:Qlink3}}}.
\]
Both type-\ref{Type:Qlink3} expansion-based links exist because the corresponding point sets $\pi_{\ppt}(w_{ \nd})$ and $\pi_{\ppt}(z_{ \nd})$ of nodes $w_{ \nd}$ and $z_{ \nd}$ are adjacent-cluster $A_{i_{w}}$-dense (because of $S_{\sigma}$).

\medskip

\noindent\underline{Case 4}. Suppose $S_{\sigma}$ exists and it is nearby $A_{i_{w}}$-sparse and nearby $A_{i_{z}}$-dense. This time we no longer have the links to satellite nodes, but informally speaking, we can still use a non-failed vertex $b_{i_{z}}\in B_{i_{z}}(S_{\sigma})\setminus F$ as a ``satellite'' which is connected to all $\Lambda_{i_{w}}$-vertices through type-\ref{Type:Qlink3} links.

Formally, take an arbitrary vertex $b_{i_{z}}\in B_{i_{z}}(S_{\sigma})\setminus F$, which exist since $|B_{i_{z}}(S_{\sigma})|=\min\{f+1,|\Lambda_{i_{z}}(S_{\sigma})|\}=\min\{f+1,\rho_{\heavy}\} = f+1$. To avoid clutter, write $b_{z}$ for $b_{i_{z}}$. Define point $b_{z,\ppt}:=\point(b_{z},S_{z})$ and node $b_{z,\nd}:=\pi_{\nd}(b_{z,\ppt})$.
Then construct
\[
P_{\qry,\sigma}:=
\underbrace{(w_{ \nd},b_{z,\nd})}_{\text{Shortcut-based link}}
\circ
\underbrace{(b_{z,\nd},\node(S_{z}))}_{\text{Type-\ref{Type:Qlink3}}}
\circ 
\underbrace{(\node(S_{z}),z_{ \nd})}_{\text{Type-\ref{Type:Qlink3}}}.
\]
The shortcut-based link is based on shortcut $(w_{ \ppt},b_{z,\ppt})$. The shortcut exists since $w_{}\in B_{i_{w}}(S_{\sigma})\setminus F$ and $b_{z,\nd}\in B_{i_{z}}(S_{\sigma})\setminus F$. 
All the type-\ref{Type:Qlink3} expansion-based links exist because the point sets $\pi_{\ppt}(b_{z,\nd}), \pi_{\ppt}(z_{ \nd})$ are adjacent-cluster $A_{i_{z}}$-dense. 

\medskip

\noindent\underline{Constructing $P_{\qry}$.} The path $P_{\qry}$ is the concatenation of all $P_{\qry,\sigma}$ with one small modification. Note that the current $P_{\qry}$ already satisfies the desired \Cref{Item:QPathProperty1}, while \Cref{Item:QPathProperty2} may still not hold (for example, there could be two case-1 subpaths $P_{\qry,\sigma}$ with a sharing subtree node). This can be addressed easily like before: for each (occurrence of) subtree node $\node(\tau)$ on the current $P_{\qry}$, take a detour to the cluster node $\node(\cluster(\tau))$ through type-\ref{Type:Qlink2} expansion-based link. Namely, insert a subpath
\[
\underbrace{(\node(\tau),\node(\cluster(\tau)))}_{\text{Type-\ref{Type:Qlink2}}}
\circ 
\underbrace{(\node(\cluster(\tau)),\node(\tau))}_{\text{Type-\ref{Type:Qlink2}}}
\]

To see \Cref{Item:QPathProperty3}, note that total length of shortcut-based links on $P_{\qry}$ is $O(1)\cdot \sum_{\text{leaf }\sigma} h_{\diam,\alpha_{\sigma}} = O(s_{\nc})\cdot \sum_{\text{leaf }\sigma} h_{\alpha_{\sigma}} = O(s_{\nc})\cdot \leng_{G}(P)$. \Cref{Item:QPathProperty4} is straightforward, where the $O(1)$ factor comes from the ratio of subtree-node length to cluster-node length. For \Cref{Item:QPathProperty5}, note that for a node $q_{\nd}$ on $P_{\qry}$, regardless of whether it is a cluster node, a singleton node, or a subtree node, the cluster $\cluster(q_{\nd})$ must be the corresponding cluster $S_{\sigma}$ for some $\sigma$ in the decomposition. The total segment length is clearly $O(i_{\max})\cdot \leng_{G}(P)$ since there are at most $i_{\max}$ levels. The factor $O(s_{\nc}\cdot s_{\ed}\cdot s_{\tree})$ comes from the ratio of cluster-node length to the ``based length'' (i.e. $h_{\alpha}$).

\end{proof}

\subsection{The Sparsified Query Graph and the Query Algorithm}
\label{sect:LabelSparsifiedQ}

\subsubsection{Store and Access Labels}

Since we will answer queries by running a shortest path algorithm on $\Qtilde$, instead of just defining $\Qtilde$, we should be able to construct it through a labeling scheme. Therefore, we construct $\Qtilde$ along the way using \emph{interfaces} of the labeling scheme, and defer some discussion on the implementations to \Cref{sect:Interfaces}. We note that most of the ideas underlying the lower-level interface implementations already appear in \cite{long2025connectivity}. But for better understanding, we will still sketch the implementation idea right after we define an interface. Therefore, we make several assumptions to simplify our presentation, which are all without loss of generality.

First, we assume that we can store information in ``labels'' of clusters and in the query phase, can access ``labels'' of affected clusters. Formally, this means at the end we will further store the ``label'' of a cluster $S$ at each vertex in $S$ (this is affordable since a vertex belongs to a small number of clusters). Then in the query phase, we can access the label of an affected cluster $S$ since $S\cap (F\cup\{s,t\})$ is non-empty. Similarly, we assume we have ``labels'' of points by storing them in the corresponding vertices, and thus in the query phase, we can access ``labels'' of failed points. 

Regarding residual subtrees, recall that we can ``linearize'' each Steiner tree into a path in the construction phase, so each residual subtree $\tau$ in the query phase is (roughly) an interval $I=\tau$ on the path whose two outer endpoints are failed points. We do not have ``labels'' for each residual subtree, since the number of possible residual subtrees for a cluster could be quadratic over all possible queries. But remember that we can access the failed point label of its two outer endpoints. Moreover, recall the idea of \emph{canonical intervals} in \Cref{def:CanonicalIntervals} in \Cref{sect:OracleConstruction}, since each point creates a small number of canonical intervals, we can assume ``labels'' for canonical intervals. Also, recall that each residual interval $I$ is the union of its two corresponding canonical intervals $I^{\rarr}$ and $I^{\larr}$ created by the two failed outer endpoints of $I$. So once the information we required from $\tau$ is \emph{mergeable} (namely, the information of $I$ immediately follows from the information of $I^{\rarr}$ and $I^{\larr}$). Then we can focus on canonical intervals instead, or equivalently, assume $\tau$ is a canonical interval which has a ``label'' to record its information. 

With the above assumption, we can immediately implement the following interfaces.
\begin{itemize}
\item $\mathsf{ListAffectedClusters}()$, which outputs (the identifiers of) all affected clusters.
\item $\mathsf{ListFailedPoints}()$, which outputs (the identifiers of) all failed points.
\item $\mathsf{ListResidualSubtrees}()$, which outputs (the representation of) all residual subtrees. Each residual subtree is represented by its two outer endpoints.
\item $\mathsf{ListAffectedShortcuts}()$, which outputs all affected shortcuts (including the basic information such as length and ids of endpoints). To implement this, note that a cluster only owns a small number of shortcuts by \Cref{ob:ClusterOwnShortcut}, so we just store all of them into the cluster label, and then access the labels of affected clusters to obtain affected shortcuts.
\end{itemize}

\subsubsection{Nodes in the Sparsified Query Graph $\Qtilde$}

The nodes in the sparsified query graph $\Qtilde$ are exactly those in the query graph $Q$. Cluster nodes can be constructed using interface $\mathsf{ListAffectedClusters}()$ as we discussed above. Using the following interface $\mathsf{ListCorePoints}(\tau)$, we can construct the singleton nodes and subtree nodes.

\begin{tcolorbox}[
  colback=blue!3,
  colframe=black,
  boxrule=1.5pt,
  arc=0pt,
  left=6pt,
  right=6pt,
  top=4pt,
  bottom=4pt
]
\noindent\textcolor{DarkBlue}{\textit{Interface}} $\mathsf{ListCorePoints}(\tau)$. It receives a residual subtree $\tau$. It determines whether $\tau$ is heavy, and if not, lists all points in $\core(\tau)$. 
\end{tcolorbox}
This is possible since a light $\tau$ has $|\core(\tau)|< \rho_{\heavy}$. To implement this, note that the required output is mergeable, so we can assume $\tau$ is a canonical interval and store the heaviness and $\core(\tau)$ in its label. We refer to \Cref{sect:InterfaceHeaviness} for more details.

\subsubsection{Links in the Sparsified Query Graph $\Qtilde$}
\label{sect:QtildeLink}

Next, we define the links in the sparsified query graph $\Qtilde$ and explain how they can be detected through the labeling-scheme interfaces. We refer to \Cref{sect:OverviewLabelingSchemes} for further intuition behind these definitions, particularly the notions of upward and downward denseness.

\begin{definition}[Upward Shortcut-Neighbors]
Consider a subset $H'\subseteq H$ of shortcuts. For each cluster $S$, each higher level $i'\geq i:=\level(S)$, each point $\vpt\in\core(S)$, let $\hat{S}\anceq S$ be the level-$i'$ ancestor of $S$ (note that possibly $i'=i$ and $S'=S$). The \emph{level-$i'$ shortcut-neighbors} of $\vpt$ w.r.t. $H'$ are
\[
\Nbr(\vpt, i', H'):=\{u_{\ppt}\in\core(\hat{S})\mid \text{there exists an $H'$-shortcut connecting $u_{\ppt}$ and $\vpt$}\}
\]
More generally, for a set of points $K\subseteq \core(S)$, the level-$i'$ shortcut-neighbors of $K$ w.r.t. $H'$ are
\[
\Nbr(K,i',H'):=\bigcup_{\vpt\in K}\Nbr(\vpt, i', H').
\]

\end{definition}

\paragraph{Shortcut-Based Links.} 

\begin{definition}[Upward Denseness]
\label{def:UpwardDenseness}
For each affected cluster $S$, each higher level $i'\geq i:=\level(S)$, let $\hat{S}\anceq S$ be the level-$i'$ ancestor of $S$. A point $\vpt\in\core(S)\setminus F_{\ppt}$ is \emph{upward $A_{i'}$-dense} if
\[
|\Nbr(\vpt,i',H\setminus H_{\aff})|\geq \rho_{\heavy},
\]
otherwise it is \emph{upward $A_{i'}$-sparse}. A residual subtree $\tau$ of $S$ (or equivalently, its core point set $\core(\tau)$) is \emph{upward $A_{i'}$-dense} if $\Nbr(\core(\tau),i',H\setminus H_{\aff})\geq \rho_{\heavy}$, and upward $A_{i'}$-sparse otherwise.
\end{definition}

\begin{tcolorbox}[
  colback=black!3,
  colframe=black,
  boxrule=1.5pt,
  arc=0pt,
  left=6pt,
  right=6pt,
  top=4pt,
  bottom=4pt
]
\setcounter{Qtildetypectr}{0}
\refstepcounter{Qtildetypectr}\label{Type:Qlink11}
\noindent\underline{\textcolor{DarkRed}{Type \arabic{Qtildetypectr}}.} For each affected cluster $S$, and level $i'\geq i:=\level(S)$, let $\hat{S}\anceq S$ be the level-$i'$ ancestor.
\begin{itemize}
\item For each singleton node $v_{\nd}$ with $\cluster(v_{\nd}) = S$, if the point $\vpt:=\pi_{\ppt}(v_{\nd})$ is upward $A_{i'}$-sparse, then for each $\upt\in \Nbr(\vpt, i',H\setminus H_{\aff})\setminus F_{\ppt}$, add a link between $v_{\nd}$ and $u_{\nd}:=\pi_{\nd}(\upt)$ with length $\beta(\vpt, \upt, H\setminus H_{\aff})$.
\item For each heavy residual subtree $\tau$ of $S$, if $\tau$ is upward $A_{i'}$-sparse, then for each $\upt\in\Nbr(\core(\tau),i',H\setminus H_{\aff})\setminus F_{\ppt}$, add a link between $\node(\tau)$ and $\pi_{\nd}(\upt)$ with length $\beta(\core(\tau),\upt,H\setminus H_{\aff})$.
\end{itemize}
\end{tcolorbox}
Moreover, we can construct these shortcuts by invoking the following two interfaces.

\begin{tcolorbox}[
  colback=blue!3,
  colframe=black,
  boxrule=1.5pt,
  arc=0pt,
  left=6pt,
  right=6pt,
  top=4pt,
  bottom=4pt
]
\noindent\textcolor{DarkBlue}{\textit{Interface}} $\mathsf{ListUpwardNeighbors}(\tau,i')$. The input is a residual subtree $\tau$, and a level $i'\geq i:=\level(\tau)$. The interface should first determine whether $\tau$ is upward $A_{i'}$-dense, and if not, output all points $\upt$ in $\Nbr(\core(\tau),i',H\setminus H_{\aff})$ explicitly with the length $\beta(\core(\tau),\upt,H\setminus H_{\aff})$. 
\end{tcolorbox}

To implement it, again the required information is mergeable, so assume $\tau$ is a canonical interval. Let us first focus on obtaining $\Nbr(\core(\tau),i',H\setminus H_{\aff})$ when $\tau$ is upward $A_{i'}$-sparse. By definition, $|\Nbr(\core(\tau),i',H\setminus H_{\aff})|\leq \rho_{\heavy}$. Although the size is small, we cannot really store it in the construction phase since $H_{\aff}$ is unknown. Fortunately, we have $|\Nbr(\core(\tau),i',H)|\leq |\Nbr(\core(\tau),i',H\setminus H_{\aff})| + \lambda_{\affS}\cdot \lambda_{\nb}$ by \Cref{ob:AffectedSCEndpoints}, so instead we can store $\Nbr(\core(\tau),i',H)$ and for each point $\upt$ inside, store $\delta(\core(\tau),\upt,H)$. Then, we can determine whether $\upt\in \Nbr(\core(\tau),i',H\setminus H_{\aff})$ by obtaining $\delta(\core(\tau),\upt,H\setminus H_{\aff})$ by subtracting the contributions of affected shortcuts. Note that we know all affected shortcuts from $\mathsf{ListAffectedShortcuts}()$. Finally, to get the length $\beta(\core(\tau),\upt, H\setminus H_{\aff})$, we simply group the shortcuts by length (there are only $\Otilde(1)$ distinct lengths) and apply the above argument to each length class separately.

\begin{tcolorbox}[
  colback=blue!3,
  colframe=black,
  boxrule=1.5pt,
  arc=0pt,
  left=6pt,
  right=6pt,
  top=4pt,
  bottom=4pt
]
\noindent\textcolor{DarkBlue}{\textit{Interface}} $\mathsf{ListUpwardNeighbors}(\vpt,\tau,i')$. The input is a \emph{light} residual subtree $\tau$, a point $\vpt\in\core(\tau)$, and a level $i'\geq i:=\level(\tau)$. The interface should first determine whether $\vpt$ is upward $A_{i'}$-dense, and if not, output all points $\upt$ in $\Nbr(\vpt,i',H\setminus H_{\aff})$ explicitly with the length $\beta(\core(\tau),\upt,H\setminus H_{\aff})$. 
\end{tcolorbox}

Suppose we have a ``label'' for $\vpt$. Then we can proceed exactly as above. But where should we store the ``label'' for $\vpt$? The key is that $\tau$ is guaranteed to be light. Recall the interface $\mathsf{ListCorePoints(\tau)}$, which will output the whole $\core(\tau)$ that includes $\vpt$. Therefore, we can store the ``label'' of $\vpt$ as a \emph{secondary label} of the labels behind $\mathsf{ListCorePoints(\tau)}$. We refer to \Cref{sect:InterfacesUpwardNbr} for more details.

\paragraph{Expansion-Based Links.} Next, we will add four types of expansion-based links to $\Qtilde$. We add to $\Qtilde$ the same type-\ref{Type:Qlink2} and type-\ref{Type:Qlink3} links as in $Q$, which are based on residual subtrees heaviness and adjacent-cluster denseness. In addition, the remaining two types will be defined below, which are based on upward denseness (in \Cref{def:UpwardDenseness}) and \emph{downward denseness} (in \Cref{def:DownwardDenseness}).

Before we move on, let us discuss how to construct type-\ref{Type:Qlink2} and type-\ref{Type:Qlink3} links. Type-\ref{Type:Qlink2} links can be added easily once we construct the node set. For type-\ref{Type:Qlink3} link, we need the following interface for checking adjacent-cluster denseness.

\begin{tcolorbox}[
  colback=blue!3,
  colframe=black,
  boxrule=1.5pt,
  arc=0pt,
  left=6pt,
  right=6pt,
  top=4pt,
  bottom=4pt
]
\noindent\textcolor{DarkBlue}{\textit{Interface}} $\mathsf{CheckAdjCluster}(\tau\text{ or }\vpt, i')$. It receives either a residual subtree $\tau$ or a point $\vpt\in\core(\tau)$ for a light residual subtree, and a level $i'>i:=\level(\tau)$. It should determine whether $\tau$ (or $\vpt$) is adjacent-cluster $A_{i'}$-dense. 
\end{tcolorbox}

To implement the above interface, let $S$ be the cluster of $\tau$. This task is asking whether there is an unaffected $\check{S}\desc S$ such that $\Lambda(\check{S})$ intersects $\core(\tau)$ (or $\vpt$) and $\check{S}$ is nearby $A_{i'}$-dense. All conditions of $\check{S}$ can be easily checked in the construction phase, except that we do not know whether $\check{S}$ is affected. However, the number of affected clusters is bounded by $\lambda_{\affS}$, so simply store $\lambda_{\affS}+1$ candidates. When the input is $\tau$, we can store these candidate clusters in the label of canonical intervals. When the input is $\vpt$, since it is in a light residual subtree, similar to the interface $\mathsf{ListUpwardNeighbors}$ above, we can store the information of $\vpt$ as a \emph{secondary label} of the labels behind $\mathsf{ListCorePoints}(\tau)$. We refer to \Cref{sect:InterfacesUpwardNbr} for more details.

\begin{tcolorbox}[
  colback=black!3,
  colframe=black,
  boxrule=1.5pt,
  arc=0pt,
  left=6pt,
  right=6pt,
  top=4pt,
  bottom=4pt
]
\setcounter{Qtildetypectr}{3}
\refstepcounter{Qtildetypectr}\label{Type:Qlink4}
\noindent\underline{\textcolor{DarkRed}{Type \arabic{Qtildetypectr}}.} For each affected cluster $S$, and level $i'\geq i:=\level(S)$, let $\hat{S}\anceq S$ be the level-$i'$ ancestor.
\begin{itemize}
\item For each singleton node $v_{\nd}$ with $\cluster(v_{\nd}) = S$, if the point $\vpt:=\pi_{\ppt}(v_{\nd})$ is upward $A_{i'}$-dense, add a link between $v_{\nd}$ and $\node(\hat{S})$ with length $0$.
\item For each heavy residual subtree $\tau$ of $S$, if $\tau$ is upward $A_{i'}$-dense, then add a link between $\node(\tau)$ and $\node(\hat{S})$ with length $0$.
\end{itemize}
\end{tcolorbox}
Moreover, we can construct type-\ref{Type:Qlink4} by invoking the above interfaces $\mathsf{ListUpwardNeighbors}(\tau,i')$ and $\mathsf{ListUpwardNeighbors}(\vpt,\tau,i')$.

\begin{definition}[Downward Denseness] 
\label{def:DownwardDenseness}
For each affected cluster $S$, let $i:=\level(S)$. A point $\upt\in\core(S)\setminus F_{\ppt}$ is \emph{downward $A_{i}$-dense} if there exists an \underline{unaffected} and \underline{non-failed} shortcut $(\upt,\vpt)$ where $\vpt$ satisfies one of the following criteria.
\begin{enumerate}
\item\label{Crit:DownDense1} $\vpt$ is in $\core(\check{S})$ for some cluster $\check{S}\desc S$ s.t. $\check{S}$ is unaffected and nearby $A_{i}$-dense.
\item\label{Crit:DownDense2} $\vpt$ is in $\core(\check{\tau})$ for some residual subtree $\check{\tau}$ of some cluster $\check{S}\desc S$ (note that $\check{S}$ must be affected) s.t. $\check{\tau}$ is upward $A_{i}$-dense.
\end{enumerate}
\end{definition}

\begin{tcolorbox}[
  colback=black!3,
  colframe=black,
  boxrule=1.5pt,
  arc=0pt,
  left=6pt,
  right=6pt,
  top=4pt,
  bottom=4pt
]
\setcounter{Qtildetypectr}{4}
\refstepcounter{Qtildetypectr}\label{Type:Qlink5}
\noindent\underline{\textcolor{DarkRed}{Type \arabic{Qtildetypectr}}.} For each affected cluster $S$, let $i:=\level(S)$.
\begin{itemize}
\item For each singleton node $v_{\nd}$ with $\cluster(v_{\nd}) = S$, if the point $\pi_{\ppt}(v_{\nd})$ is downward $A_{i}$-dense, add a link between $v_{\nd}$ and $\node(S)$ with length $0$.
\end{itemize}
\end{tcolorbox}
This type of shortcuts can be constructed by invoking the following interface. However, its implementation is relatively involved which cannot be explained in a few lines. We will discuss it more in \Cref{sect:SolveDownDense}.

\begin{tcolorbox}[
  colback=blue!3,
  colframe=black,
  boxrule=1.5pt,
  arc=0pt,
  left=6pt,
  right=6pt,
  top=4pt,
  bottom=4pt
]
\noindent\textcolor{DarkBlue}{\textit{Interface}}  $\mathsf{DownwardDenseness}(\upt,\tau)$. It receives a point $\upt$, a residual subtree $\tau$ s.t. $\tau$ is a \underline{light} residual subtree and $\upt\in \core(\tau)$. It should determine whether $\upt$ is downward $A_{i}$-dense for level $i:=\level(\tau)$. 
\end{tcolorbox}

\subsubsection{Answering the Query with the Sparsified Query Graph}
\label{sect:DistLabelCorrectness}

After constructing the sparsified query graph $\Qtilde$, we simply return $\dist_{\Qtilde}(V_{\nd}(x),V_{\nd}(y))$ by running a shortest path algorithm. This gives an approximation $O(s_{\nc}\cdot s_{\ed}\cdot s_{\tree}\cdot i_{\max}) = O(k^{5})$ for exponential-time construction. When polynomial-time construction is required, $s_{\ed}$ will increase from $k$ to $k^{2}$, which gives an approximation of $O(k^{6})$. \Cref{lemma:LabelUB} and \Cref{lemma:LabelSolLowerBound} show the upper-bound direction and the lower-bound direction respectively.

\paragraph{The Upper-Bound Direction.} 

\begin{lemma}
\label{lemma:LabelUB}
There exist two nodes $x_{\nd}\in V_{\nd}(x)$ and $y_{\nd}\in V_{\nd}(y)$ s.t.
\[
\dist_{\Qtilde}(x_{\nd},y_{\nd})\leq O(s_{\nc}\cdot s_{\ed}\cdot s_{\tree}\cdot i_{\max})\cdot \dist_{G\setminus F}(x,y).
\]
\end{lemma}

The above \Cref{lemma:LabelUB} is an immediate consequence of \Cref{lemma:FindQPath} and the following \Cref{lemma:ReplaceShortcut}. Roughly speaking, we take the path $P_{\qry}$ in the query graph $Q$ from Lemma \ref{lemma:FindQPath}, and then replace each shortcut-based link on it with a path in $\Qtilde$. Finally, following \Cref{algo:ConstructPqry}, we shortcut this path using cluster nodes as ``portals''. The desired bound immediately follows.

\begin{lemma}
\label{lemma:ReplaceShortcut}
For each shortcut-based link $e_{\lk}:=(u_{\nd},v_{\nd})\in Q$, there exists a path $\Ptilde$ in $\Qtilde$ connecting $u_{\nd}$ and $v_{\nd}$ in one of the following forms.
\begin{enumerate}[label=(\alph*)]
\item $\Ptilde = (u_{\nd},v_{\nd})$ consists of a single link with link length at most $\ell(e_{\lk})$.
\item $\Ptilde = (u_{\nd},\node(S))\circ(\node(S),v_{\nd})$, where $S$ is the higher-level cluster among $\cluster(u_{\nd})$ and $\cluster(v_{\nd})$. Moreover, $\leng(\Ptilde)\leq O(1)\cdot \ell(\node(S))$.
\end{enumerate}
\end{lemma}
\begin{proof}
Let $e_{\sc} = (\upt,\vpt)$ be the shortcut that $e_{\lk}$ is based on, which means $\upt\in \pi_{\ppt}(u_{\nd})$, $\vpt\in \pi_{\ppt}(v_{\nd})$ and $\ell(e_{\sc}) = \ell(e_{\lk})$. Note that by the rule of adding shortcut-based link to $Q$, $e_{\sc}$ must be unaffected and non-failed (i.e., $\upt,\vpt\notin F_{\ppt}$) Also, by the definition of $\pi_{\ppt}(\cdot)$, when $u_{\nd}$ is a singleton node, $\pi_{\ppt}(u_{\nd})$ is exactly the point $\upt$, and when $u_{\nd} = \node(\tau)$ is a subtree node, $\pi_{\ppt}(u_{\nd})=\core(\tau)$.

Let $S=\cluster(u_{\nd})$,  $\check{S}=\cluster(v_{\nd})$, $i=\level(S)$ and $\check{i}=\level(\check{S})$. By the definition of shortcuts, one of $S$ and $\check{S}$ must be an ancestor of the other.  Without loss of generality, assume $S\anceq \check{S}$.

\medskip

\noindent{Case 1.} Suppose $\pi_{\ppt}(v_{\nd})$ is upward $A_{i}$-sparse. The existence of $e_{\sc}$ implies that $\upt$ is inside $\Nbr(\pi_{\ppt}(v_{\nd}),i,H\setminus H_{\aff})\setminus F_{\ppt}$, and thus there is a type-\ref{Type:Qlink11} link in $\Qtilde$ between $v_{\nd}$ and $u_{\nd}$ of length
\[
\beta(\pi_{\ppt}(v_{\nd}), \upt, H\setminus H_{\aff})\leq \beta(\vpt,\upt,H\setminus H_{\aff})\leq \ell(e_{\sc})=\ell(e_{\lk}).
\]
We simply take this link as our desired $\Ptilde$.

\medskip

\noindent\underline{Case 2.} From now on, suppose $\pi_{\ppt}(v_{\nd})$ is upward $A_{i}$-dense. Then there is a type-\ref{Type:Qlink4} link in $\Qtilde$ connecting $v_{\nd}$ and $\node(S)$ with link length 0. In what follows, we will argue the ``$u$-side'', and there are three subcases.

\medskip

\noindent\underline{Case 2(a).} Suppose $S=\check{S}$. Because now the ``$u$-side'' and ``$v$-side'' are at the same level, if $\pi_{\ppt}(u_{\nd})$ is upward $A_{i}$-sparse, we also get the desired $\Ptilde$ as in Case 1. 

Otherwise, both $\pi_{\ppt}(u_{\nd})$ and $\pi_{\ppt}(v_{\nd})$ are upward $A_{i}$-dense, and there is a type-\ref{Type:Qlink4} link in $\Qtilde$ connecting $u_{\nd}$ and $\node(S)$ with link length $0$. Thus we obtain a desired
\[
\Ptilde = \underbrace{(u_{\nd},\node(S))}_{\text{Type-\ref{Type:Qlink4}}}\circ\underbrace{(\node(S),v_{\nd})}_{\text{Type-\ref{Type:Qlink4}}}.
\]

\medskip

\noindent\underline{Case 2(b).} Suppose $S\anc\check{S}$ and $u_{\nd} = \node(\tau)$ is a subtree node. By definition, $\tau$ is a heavy residual subtree of $S$ and thus there exists a type-\ref{Type:Qlink2} link in $\Qtilde$ connecting $u_{\nd}$ and $\node(S)$ with link length $0$, and
\[
\Ptilde = \underbrace{(u_{\nd},\node(S))}_{\text{Type-\ref{Type:Qlink2}}}\circ\underbrace{(\node(S),v_{\nd})}_{\text{Type-\ref{Type:Qlink4}}}.
\]
\medskip

\noindent\underline{Case 2(c).} Suppose $S\anc\check{S}$ and $u_{\nd}$ is a singleton node. We will argue that the point $\upt=\pi_{\ppt}(u_{\nd})$ is downward $A_{i}$-dense, and thus there exists a type-\ref{Type:Qlink5} link in $\Qtilde$ connecting $u_{\nd}$ and $\node(S)$ with link length $0$, and we obtain
\[
\Ptilde = \underbrace{(u_{\nd},\node(S))}_{\text{Type-\ref{Type:Qlink5}}}\circ\underbrace{(\node(S),v_{\nd})}_{\text{Type-\ref{Type:Qlink4}}}.
\]

To see that $\upt$ is downward $A_{i}$-dense, note that we have an unaffected and non-failed shortcut $(\upt,\vpt)$. Furthermore, let $\check{\tau}$ be the residual subtree of $\check{S}$ s.t. $\vpt\in\core(\check{\tau})$. Since $\vpt$ is upward $A_{i}$-dense, $\check{\tau}$ is also upward $A_{i}$-dense, and thus the second criteria in \Cref{def:DownwardDenseness} of downward denseness is fulfilled.

\end{proof}

\paragraph{The Lower-Bound Direction.} \ 

\begin{lemma}
\label{lemma:LabelSolLowerBound}
Let $x_{\nd}$ and $y_{\nd}$ be two arbitrary concrete nodes in $\Qtilde$. For each pair of vertices $x\in V(x_{\nd})$ and $y\in V(y_{\nd})$, we have
\[
\dist_{G\setminus F}(x,y)\leq \dist_{\Qtilde}(x_{\nd},y_{\nd}).
\]
\end{lemma}
\begin{proof}
Let $\Ptilde$ be a shortest path in $\Qtilde$ between $x_{\nd}$ and $y_{\nd}$. We decompose $\Ptilde$ into maximal internally virtual subpaths. Consider such a subpath $\Ptilde_{\sub} = \Ptilde[u_{\nd},v_{\nd}]$. It suffices to show
\[
\dist_{G\setminus F}(V(u_{\nd}),V(v_{\nd}))\leq \leng(\Ptilde_{\sub}) - \ell(u_{\nd}) - \ell(v_{\nd}),
\]
since for each concrete node $q_{\nd}$, we have $\diam_{G\setminus F}(V(q_{\nd}))\leq \ell(q_{\nd})$. Furthermore, observe that $\Ptilde_{\sub}$ is either (a) a single shortcut-based link $(u_{\nd},v_{\nd})$ or (b) two expansion-based links $(u_{\nd},S_{\nd})\circ(S_{\nd},v_{\nd})$ sharing a common cluster node $S_{\nd}$.

\medskip

\noindent\underline{Case 1.} Suppose $\Ptilde_{\sub} = (u_{\nd},v_{\nd})$ is in form (a). It is equivalent to show that the length of this shortcut is at least $\dist_{G\setminus F}(V(u_{\nd}),V(v_{\nd}))$. Observe that every shortcut-based link in $\Qtilde$ also appears in $Q$, and every shortcut-based link in $Q$ is based on an unaffected and non-failed shortcut in $H$. Analogous to Case 1 in \Cref{lemma:GeneralSolLowerBound}, any unaffected and non-failed $H$-shortcut $(u_{\ppt},v_{\ppt})$ will not underestimate $\dist_{G\setminus F}(u,v)$, where $u,v$ are the vertices corresponding to points $u_{\ppt},v_{\ppt}$, regardless it is type-\ref{Type:Shortcut1} (from an original edge) or type-\ref{Type:Shortcut2} (from a cluster, which is unaffected since the shortcut is unaffected).

\medskip

\noindent\underline{Case 2.} Suppose $\Ptilde_{\sub} = (u_{\nd},S_{\nd})\circ(S_{\nd},u_{\nd})$ is in form (b). Let $S=\cluster(S_{\nd}), i=\level(S_{\nd})$ and $\alpha = \scale(S_{\nd})$. It suffices to show that $\ell(S_{\nd})\geq \dist_{G\setminus F}(V(u_{\nd}),V(v_{\nd}))$. These two links are expansion-based links, which have types \ref{Type:Qlink2}, \ref{Type:Qlink3}, \ref{Type:Qlink4}, or \ref{Type:Qlink5}.

For better understanding and to avoid tedious case analysis, let us repeat our strategy of exploiting the LC-expanding property \Cref{lemma:LCExpanding} w.r.t. $A_{i}$. Take the expansion-based link $(u_{\nd},S_{\nd})$ as an example. Regardless it has which type, at the end we want to find a vertex set $V_{u}$ such that
\begin{enumerate}[label=(\alph*)]
\item $V_{u}\subseteq A_{i}$ and $|V_{u}|\geq \rho_{\heavy}$.
\item each vertex $u'\in V_{u}$ is close to $S$ \underline{in $G_{i,\alpha}$}, or formally, $\dist_{G_{i,\alpha}}(u',S)\leq \theta$ for some $\theta$ such that $2\theta+h_{\diam,\alpha}\leq h_{\ed,\alpha}$.
\item each vertex $u'\in V_{u}\setminus F$ is close to $V(u_{\nd})$ \underline{in $G\setminus F$}, or formally, $\dist_{G\setminus F}(u',V(u_{\nd}))\leq \eta$ for some $\eta$ such that $2\eta + h_{\ed,\alpha}\cdot s_{\ed}\leq \ell(S_{\nd})$.
\end{enumerate}
We first argue why finding such a set $V_{u}$ (and $V_{v}$ for the ``$v$-side'') will upper bound $\dist_{G\setminus F}(V(u_{\nd}),V(v_{\nd}))$. First, we have 
\[
\diam_{G_{i,\alpha}}(V_{u},V_{v})\leq 2\theta + \diam_{G_{i,\alpha}}(S)\leq 2\theta + h_{\diam,\alpha}\leq h_{\ed,\alpha}.
\]
Therefore, 
we can apply \Cref{lemma:LCExpanding} on graph $G_{i,\alpha}$, node-weighting $A_{i}$, and the two sets $V_{u},V_{v}$, which gives
$\dist_{G\setminus F}(V_{u}\setminus F,V_{v}\setminus F)\leq h_{\ed,\alpha}\cdot s_{\ed}$.
Finally, using (c), we have 
\[
\dist_{G\setminus F}(V(u_{\nd}),V(v_{\nd}))\leq 2\eta + h_{\ed,\alpha}\cdot s_{\ed}\leq \ell(S_{\nd}).
\]
and the proof completes. 

Therefore, it remains to determine $V_{u},\theta,\eta$ for $(u_{\nd},S_{\nd})$ in different types. In fact, for all types, we will pick a $V_{u}\subseteq S$, so we can safely set $\theta = 0$.

\medskip

\noindent{\underline{$(u_{\nd},S_{\nd})$ is type-\ref{Type:Qlink2}.}} Let $V_{u}:= V(u_{\nd})$ and $\eta = 0$. In this case, $u_{\nd}$ must be a subtree node $\tau_{\nd}$ where $\tau$ is a heavy residual subtree of $S$. Thus, (a), (b) and (c) clearly follows.

\medskip

\noindent{\underline{$(u_{\nd},S_{\nd})$ is type-\ref{Type:Qlink3}.}} Let $V_{u}:=\Lambda_{i}(\check{S})$, where $\check{S}\desc S_{u}:=\cluster(u_{\nd})\desceq S$ is the unaffected cluster witnessing the adjacent-cluster $A_{i}$-denseness of the point set $\pi_{\ppt}(u_{\nd})$. Let $\eta = h_{\diam,\check{\alpha}} + 2h_{\check{\alpha}}$, where $\check{\alpha} = \scale(\check{S})$. We have (b) since $\Lambda_{i}(\check{S})\subseteq S$. Such a $\eta$ fulfills (c) since this is the bound of $\diam_{G\setminus F}(\Lambda(\check{S}))$ (note that $V(u_{\nd})$ must intersect $\Lambda(\check{S})$).

\medskip

\noindent{\underline{$(u_{\nd},S_{\nd})$ is type-\ref{Type:Qlink4}.}} We first define $V_{u}$. Let $U_{\ppt}$ be the point set $\pi_{\ppt}(u_{\nd})$. let $U'_{\ppt}$ be the point set $\Nbr(U_{\ppt},i,H\setminus H_{\aff})$, which satisfies $U'_{\ppt}\subseteq \core(S)$ and $|U'_{\ppt}|\geq \rho_{\heavy}$ since $U_{\ppt}$ is upward $A_{i}$-dense. Finally, define $V_{u}$ to be $\pi_{\vtx}(U'_{\ppt})$, which satisfies (a). Again, we have $V_{u}\subseteq S$ which fulfills (b). Lastly, set $\theta = h_{\diam,\alpha} + 2h_{\alpha}$, since each point in $U'_{\ppt}$ is adjacent to a point in $U_{\ppt}$ through an unaffected shortcut of length at most $h_{\diam,\alpha} + 2h_{\alpha}$, which never underestimates the $G\setminus F$-distance.

\medskip

\noindent{\underline{$(u_{\nd},S_{\nd})$ is type-\ref{Type:Qlink5}.}} Note that $u_{\nd}$ is a singleton node in this case, so let $u_{\ppt}$ be the point $\pi_{\ppt}(u_{\nd})$, and let $u$ be the vertex $\pi_{\vtx}(u_{\ppt})$. To define $V_{u}$, there are two cases depending on the object witnessing the downward $A_{i}$-denseness of $u_{\ppt}$.
\begin{itemize}
\item Suppose the witness object is a cluster $\check{S}\desc S_{u}:=\cluster(u_{\ppt})\desceq S$, which is unaffected and nearby $A_{i}$-dense. Then let $V_{u}:=\Lambda_{i}(\check{S})$. We can set $\eta = (h_{\diam,\alpha} + 2h_{\alpha}) + (h_{\diam,\check{\alpha}} + 2h_{\check{\alpha}})$, where $\check{\alpha}=\scale(\check{S})$, since each vertex in $V_{u}\setminus F$ is connected to $u$ through the unaffected cluster $\check{S}$ plus an unaffected, non-failed shortcut connecting $\check{S}$ and $u$. 
\item Suppose the witness object is a upward $A_{i}$-dense residual subtree $\check{\tau}$ of some affected cluster $\check{S}\desc S_{u}$. Then let $V_{u}:=\Nbr(\core(\check{\tau}),i,H\setminus H_{\aff})$. We can set $\eta = (h_{\diam,\alpha} + 2h_{\alpha}) + h_{\tree,\check{\alpha}} + (h_{\diam,\alpha} + 2h_{\alpha})$, because each vertex $u'$ in $V_{u}\setminus F$ can go to $u$ (avoiding failures) by first using a shortcut from $u'$ to $V_{\tmn}(\check{\tau})$, walking inside $\check{\tau}$, and lastly using a shortcut from $V_{\tmn}(\check{\tau})$ to $u$. Note that the term $h_{\tree,\check{\alpha}}$ comes from walking inside $\check{\tau}$ as a residual subtree has diameter $h_{\tree,\check{\alpha}}$.
\end{itemize}

\end{proof}

\subsection{Solving Downward Denseness}
\label{sect:SolveDownDense}

In this section, we will discuss our approach to implementing the interface $\mathsf{DownwardDenseness}(\upt,\tau)$ in the query phase. Recall that the input satisfies that $\tau$ is a \underline{light} residual subtree of $S$, and $\upt\in\core(\tau)$. Let $i:=\level(S)$. Our goal is to detect an arbitrary unaffected and non-failed shortcut that satisfies Criterion \ref{Crit:DownDense1} or Criterion \ref{Crit:DownDense2} in \Cref{def:DownwardDenseness}. We call such a shortcut a \emph{witness shortcut}.

To this end, for the cluster $S$, define
\[
\region(S):=\bigcup_{\text{cluster $S'\prec S$}}\core(S)
\]
to be the union of the cores of all strict descendants of $S$. Let $H_{S}$ be the set consisting of all the shortcuts with their higher endpoints in $\core(S)$ and lower endpoints in $\region(S)$. By definition, $e^{\star}_{\sc}$ must belong to $H_{S}$.

Next, we will partition $H_{S}$ into three parts: the \emph{downward shortcuts} $H_{S}^{\down}$, \emph{upward shortcuts} $H_{S}^{\up}$ and \emph{sparse shortcuts} $H_{S}^{\spar}$. Then try to detect the witness shortcut from each of them separately.

\paragraph{Partition of $H_{S}$.} Consider a shortcut $e_{\sc}=(u_{\ppt},\vpt)\in H_{S}$ where $\vpt$ is its lower endpoint. Thus $v_{\ppt}\in\core(\check{S})$ for some cluster $\check{S}\desc S$. 
\begin{itemize}
\item(Sparse Shortcuts) Suppose $\check{S}$ is nearby $A_{i}$-sparse. We classify $e_{\sc}$ as a sparse shortcut and put it into $H_{S}^{\spar}$.
\item(Downward Shortcuts) Suppose $\check{S}$ is nearby $A_{i}$-dense, and
\[
|\Nbr(\vpt, i, H)|\geq \lambda_{\affS}\cdot \lambda_{\nb}+ \rho_{\heavy}.
\]
We classify $e_{\sc}$ as a downward shortcut and put it into $H_{S}^{\down}$.
\item (Upward Shortcuts) Suppose $\check{S}$ is nearby $A_{i}$-dense and $|\Nbr(\vpt,i,H)|<\lambda_{\affS}\cdot \lambda_{\nb} + \rho_{\heavy}$. We classify $e_{\sc}$ as an upward shortcut and put it into $H_{S}^{\up}$
\end{itemize}

We emphasize that all these shortcut sets $H_{S},H^{\spar}_{S},H^{\down}_{S},H^{\up}_{S}$ are defined independently of the query $(x,y,F)$, so we can use them to construct labels in the construction phase, as we will see. 

\paragraph{Detect Witness Shortcuts in $H_{S}^{\down}$.} By the lemma below, we can simply invoke the interface $\mathsf{FindDownwardShortcuts}(\upt, S)$ below.

\begin{lemma}
In the query phase, any unaffected shortcut $e_{\sc}=(\upt,\vpt)\in H_{S}^{\down}$ with $\vpt\notin F_{\ppt}$ is a witness shortcut.
\end{lemma}
\begin{proof}
The lemma condition guarantees that $e_{\sc}$ is unaffected and non-failed, so it suffices to prove that $\vpt$ satisfies either Criterion \ref{Crit:DownDense1} or Criterion \ref{Crit:DownDense2} in \Cref{def:DownwardDenseness}. 

If the cluster $\check{S}$ with $\core(\check{S})\ni \vpt$ is unaffected, then $\check{S}$ must be nearby $A_{i}$-dense by the definition of $H_{S}^{\down}$, and thus Criterion \ref{Crit:DownDense1} is fulfilled.

Suppose the cluster $\check{S}$ is affected. Since $\vpt\notin F_{\ppt}$, there must be a residual subtree $\check{\tau}$ of $\check{S}$ s.t. $\vpt\in\core(\check{\tau})$. Now, observe that 
\begin{align*}
|\Nbr(\core(\check{\tau}),i,H\setminus H_{\aff})|&\geq |\Nbr(\core(\check{\tau}),i,H)| -\lambda_{\affS}\cdot \lambda_{\nb} \\
&\geq  |\Nbr(\vpt,i,H)| - \lambda_{\affS}\cdot \lambda_{\nb}\\
&\geq \rho_{\heavy},
\end{align*}
where the first inequality holds because at most $\lambda_{\affS}\cdot \lambda_{\nb}$ points can be endpoints of affected shortcuts (see \Cref{ob:AffectedSCEndpoints}), and the third inequality is by the definition of $H^{\down}_{S}$. Hence, Criterion \ref{Crit:DownDense2} is fulfilled.
\end{proof}

\begin{tcolorbox}[
  colback=blue!3,
  colframe=black,
  boxrule=1.5pt,
  arc=0pt,
  left=6pt,
  right=6pt,
  top=4pt,
  bottom=4pt
]
\noindent\textcolor{DarkBlue}{\textit{Interface}} $\mathsf{FindDownwardShortcuts}(\upt, S)$. It receives point $\upt$ and an affected $S$ such that $\upt\in\core(\tau)$ for a light residual subtree $\tau$ of $S$. It determines whether there exists an unaffected and non-failed shortcut $(\upt,\vpt)\in H^{\down}_{S}$.
\end{tcolorbox}

If we have a ``label'' of $\upt$, then this interface can be viewed as a ``downward'' version of the interface $\mathsf{ListUpwardNeighbors}(\vpt,\tau,i')$ in \Cref{sect:QtildeLink}. We do not know the ``unaffected'' and ``non-failed'' information of each shortcut $(u_{\ppt},v_{\ppt})\in H^{\down}_{S}$ in the preprocessing phase. To address this, we use a strategy similar to the implementation of $\mathsf{ListUpwardNeighbors}$. We store up to $\lambda_{\affS}\cdot \lambda_{\nb} + f + 1$ shortcuts $(u_{\ppt},v_{\ppt})\in H^{\down}_{S}$, and one of them must be unaffected and non-failed in the query phase. The additive term $\lambda_{\affS}\cdot \lambda_{\nb}$ enables us to tolerate all affected shortcuts (because of \Cref{ob:AffectedSCEndpoints}), while the additive term $f$ enables us to tolerant all failed $v_{\ppt}$.

To store this ``label'' of $\upt$, recall that $\upt\in\core(\tau)$ for a light residual subtree. Therefore, we can again append $\upt$'s ``label'' as a \emph{secondary label} to the labels behind $\mathsf{ListCorePoints}(\tau)$.

\paragraph{Detect Witness Shortcuts in $H^{\spar}_{S}$.} 
Note that for a witness shortcut $e_{\sc}^{\star} = (\upt,\vpt^{\star})$ in $H^{\spar}_{S}$, we must have $v^{\star}_{\ppt}\in\core(\check{\tau})$ for some upward $A_{i}$-dense residual subtree $\check{\tau}$ of some affected cluster $\check{S}\desc S$, because $e_{\sc}^{\star}$ cannot satisfy Criterion \ref{Crit:DownDense1} (by the definition of $H^{\spar}_{S}$) and it must satisfy Criterion \ref{Crit:DownDense2}. Also, the cluster $\check{S}$ must be nearby $A_{i}$-sparse by the definition of $H^{\spar}_{S}$.

Therefore, we will check the existence of a witness shortcut by simply counting the number of witness shortcuts. Let $\check{\cal S}_{\spar}$ denote the set of all clusters $\check{S}\desc S$ that is nearby $A_{i}$-sparse, and recall from \Cref{def:CountLengthFunctions} that  $\delta(U_{\ppt},V_{\ppt}, H')$ denotes the number $H'$-shortcuts connecting two disjoint point sets $U_{\ppt}$ and $V_{\ppt}$. Define
\begin{itemize}
\item $\delta_{0}:=$ the summation of $\delta(\upt,\core(\check{S}),H^{\spar}_{S})$ over all $\check{S}\in\check{\cal S}_{\spar}$.
\item $\delta_{\fail}:=$ the summation of $\delta(\upt,\vpt,H^{\spar}_{S})$ over all failed points $\vpt\in F_{\ppt}$ s.t. $\vpt\in\core(\check{S})$ for some $\check{S}\in\check{\cal S}_{\spar}$.
\item $\delta_{\tree}:=$ the summation of $\delta(\upt,\core(\check{\tau}),H^{\spar}_{S})$ over all upward $A_{i}$-sparse residual subtrees $\check{\tau}$ of all $\check{S}\in\check{\cal S}_{\spar}$.
\item $\delta_{\aff}:=$ the number of affected shortcuts $(\upt,\vpt)$ in $H^{\spar}_{S}$ s.t. for some $\check{S}\in\check{\cal S}_{\spar}$,  $\vpt\in \core(\check{\tau})$ of some upward $A_{i}$-dense residual subtree $\check{\tau}$ of $\check{S}$.
\end{itemize}
Then a witness shortcut in $H^{\spar}_{S}$ exists if and only if $\delta_{0}-\delta_{\fail}-\delta_{\tree}-\delta_{\aff}>0$. To see this, roughly speaking, subtracting $\delta_{\fail}$ and $\delta_{\tree}$ guarantees that we only count shortcut $(\upt,\vpt)$ with $\vpt^{\star}\in\core(\check{\tau})$ for an upward $A_{i}$-dense $\check{\tau}$, and subtracting $\delta_{\aff}$ rules out all the affected shortcuts. Moreover, no shortcut is doubly subtracted. Next, we compute $\delta_{0},\delta_{\fail},\delta_{\tree}$ and $\delta_{\aff}$ separately with the help of some additional interfaces.

\medskip

\noindent\underline{Compute $\delta_{0}$.} First, we can obtain the list of all affected clusters. Then we can obtain the list of clusters in $\check{\cal S}$ by checking affected clusters one by one, since it is trivial to store and access information of affected clusters about ancestors and nearby denseness. Therefore, it suffices to invoke the following interface.

\begin{tcolorbox}[
  colback=blue!3,
  colframe=black,
  boxrule=1.5pt,
  arc=0pt,
  left=6pt,
  right=6pt,
  top=4pt,
  bottom=4pt
]
\noindent\textcolor{DarkBlue}{\textit{Interface}} $\mathsf{CountSparseShortcuts}(\upt,S,\check{S})$. It receives point $\upt$, affected clusters $S$ and $\check{S}$, such that $\upt\in \core(S)$, $S\anc \check{S}$ and $\check{S}$ is nearby $A_{i}$-sparse for $i:=\level(S)$. It returns $\delta(\upt,\core(\check{S}),H^{\spar}_{S})$.
\end{tcolorbox}

The above interface is actually easy to implement because for a particular cluster $\check{S}$, the number of its ancestors $S$ is small. The number of possible $\upt$ could be large, but a $\upt$ with non-zero $\delta(\upt,\core(\check{S}),H^{\spar}_{S})$ must belong to $\Lambda_{i}(\check{S})$, and $|\Lambda_{i}(\check{S})|$ is small since $\check{S}$ is nearby $A_{i}$-sparse. Therefore, it is affordable to precompute all the answers and store them in the label of the cluster $\check{S}$.

\medskip

\noindent\underline{Compute $\delta_{\fail}$.} Similarly, we can enumerate all cluster $\check{S}\in\check{\cal S}$, then enumerate all failed points $\vpt \in F_{\ppt}$ (again it is trivial to list all failed points), and lastly find those inside $\core(\check{S})$. Then invoke the following interface.

\begin{tcolorbox}[
  colback=blue!3,
  colframe=black,
  boxrule=1.5pt,
  arc=0pt,
  left=6pt,
  right=6pt,
  top=4pt,
  bottom=4pt
]
\noindent\textcolor{DarkBlue}{\textit{Interface}} $\mathsf{CountSparseShortcuts}(\upt,S,\vpt,\check{S})$. The $\upt,S,\check{S}$ in the input are as above, and the additional $\vpt$ should satisfy $\vpt\in\core(\check{S})$ and $\vpt\in F_{\ppt}$. It returns $\delta(\upt,\vpt,H^{\spar}_{S})$.
\end{tcolorbox}

Again, the answer is non-zero only when $\upt\in \Lambda_{i}(\check{S})$ by the nearby $A_{i}$-denseness of $\check{S}$, so just simply store all the non-zero answers in the label of point $\vpt$.

\medskip

\noindent\underline{Compute $\delta_{\tree}$.} The computation of $\delta_{\tree}$ requires some more work. Again, it is easy to list all residual subtrees of affected clusters, so now we focus on computing $\delta(\upt,\core(\check{\tau}),H^{\spar}_{S})$ for one $\check{\tau}$ in the expression.

\begin{tcolorbox}[
  colback=blue!3,
  colframe=black,
  boxrule=1.5pt,
  arc=0pt,
  left=6pt,
  right=6pt,
  top=4pt,
  bottom=4pt
]
\noindent\textcolor{DarkBlue}{\textit{Interface}} $\mathsf{CountSparseShortcuts}(\upt, S, \check{\tau}, \check{S})$. Inputs $\upt, S, \check{S}$ are as above, and the additional $\check{\tau}$ should be a upward $A_{i}$-sparse residual subtree of $\check{S}$. It returns $\delta(\upt, \core(\check{\tau}), H^{\spar}_{S})$.
\end{tcolorbox}

To implement the additional interface above, in fact we only rely on the fact that $\check{\tau}$ is upward $A_{i}$-sparse, and the nearby $A_{i}$-denseness of $\check{S}$ is not important. The key observation is still that $\delta(\upt,\core(\check{\tau}), H^{\spar}_{S})$ is non-zero only if $\upt\in \Nbr(\core(\check{\tau}),i,H)$. Since $\check{\tau}$ is upward $A_{i}$-sparse, we know the size of $\Nbr(\core(\check{\tau}),i,H\setminus H_{\aff})$ is small (i.e. at most $\rho_{\heavy}$). Then $|\Nbr(\core(\check{\tau}),i,H)|\leq |\Nbr(\core(\check{\tau}),i,H\setminus H_{\aff})| + \lambda_{\affS}\cdot \lambda_{\nb}$ by \Cref{ob:AffectedSCEndpoints}.

A small issue is that the output is no longer mergeable, so we are unable to exploit the canonical intervals. Nevertheless, we can still somehow store something in the point labels, such that in the query phase, each point $\upt\in \Nbr(\core(\check{\tau}),i,H)$ appears in the labels of \emph{both} outer endpoints of $\tau$. Then if we precompute two prefix-sum numbers and store them along with these two $\upt$ in the outer endpoint labels, then the desired $\delta(\upt,\core(\check{\tau}),H^{\spar}_{S})$ will be the differences of these two prefix-sum numbers. See \Cref{sect:PrefixSumLabel} for more details.

\medskip

\noindent{\underline{Compute $\delta_{\aff}$}.} Since we can obtain all affected shortcuts, we simply check for each of them whether it contributes to $\delta_{\aff}$. Note that checking whether a shortcut belongs to $H_{S}^{\spar}$ is simple because an $H$-shortcut belongs to at most one group in $\{H_{S}^{\down},H_{S}^{\up},H_{S}^{\spar}\mid \text{all clusters $S$}\}$.

\paragraph{Detect Witness Shortcuts in $H^{\up}_{S}$.}

We will use a similar counting argument. For each shortcut $(\upt,\vpt)\in H^{\up}_{S}$ where $\vpt$ is the lower endpoint, by definition $\vpt$ must belong to one of the following: (1) $\core(\check{\tau})$ of a upward $A_{i}$-sparse residual subtree $\check{\tau}$, (2) $F_{\ppt}$, (3) $\core(\check{\tau})$ of a upward $A_{i}$-dense residual subtree, and (4) a nearby $A_{i}$-dense unaffected cluster $\check{S}$. We want to detect an unaffected shortcut in $H^{\up}_{S}$ whose $\vpt$ falls in (3) or (4). To this end, define
\begin{itemize}
\item $\delta_{0}:=\delta(\upt, \cdot, H_{S}^{\up})$, i.e., the number of $H_{S}^{\up}$-shortcuts with $\upt$ as one of its endpoints.
\item $\delta_{\fail}:=$ the summation of $\delta(\upt, \vpt, H_{S}^{\up})$ over all $\vpt\in F_{\ppt}$ with $\level(\vpt)<i:=\level(\upt)$ and $|\Nbr(\vpt,i,H)|<\lambda_{\affS}\cdot\lambda_{\nb}+\rho_{\heavy}$ (note that the last condition is by the definition of $H^{\up}_{S}$).
\item $\delta_{\tree}:=$ the summation of $\delta(\upt,\core(\check{\tau}),H^{\up}_{S})$ over all upward $A_{i}$-sparse residual subtrees $\check{\tau}$ s.t. $\check{\tau}$ belongs to a nearby $A_{i}$-dense cluster $\check{S}\desc S$.
\item $\delta_{\aff}:=$ the number of affected shortcuts $(\upt,\vpt)$ in $H^{\up}_{S}$ s.t. for some nearby $A_{i}$-dense affected cluster $\check{S}\desc S$, $\vpt\in \{\core(\check{\tau})\mid \text{upward $A_{i}$-dense residual subtree $\check{\tau}$ of $\check{S}$}\}$
\end{itemize}
Again, it is clear that a witness shortcut in $H^{\up}_{S}$ exists if and only if $\delta_{0} - \delta_{\fail} - \delta_{\tree} - \delta_{\aff}>0$. Next we discuss how to compute each of them.

\medskip

\noindent{\underline{Compute $\delta_{0}$}.} $\delta_{0}$ is just a number that can be precomputed in the construction phase. To see where to store it, we will exploit that the input of $\mathsf{DownwardDenseness}(\upt,\tau)$ guarantees $\upt\in\core(\tau)$ and $\tau$ is a light residual subtree. Recall that $\mathsf{ListCorePoints}(\tau)$ can list all points in $\core(\tau)$ for a light residual subtree. Therefore, we can store $\delta_{0}$ as a secondary label affiliated to labels behind $\mathsf{ListCorePoints}(\tau)$. Then obtain $\delta_{0}$ in the query phase along with $\upt$.

\medskip

\noindent{\underline{Compute $\delta_{\fail}$.}} For this, we first enumerate all $\vpt\in F_{\ppt}$ satisfying the conditions of the summation. For one such $\vpt$, to compute $\delta(\upt,\vpt,H^{\up}_{S})$, we need to exploit that $|\Nbr(\vpt,i,H)|<\lambda_{\affS}\cdot \lambda_{\nb} + \rho_{\heavy}$. Note that $\delta(\upt,\vpt,H^{\up}_{S})$ is non-zero only when $\upt\in \Nbr(\vpt,i,H)$. Hence, we can store in the label of point $\vpt$ all points $u'_{\ppt}\in |\Nbr(\vpt,i,H)|$ along with the number $\delta(u'_{\ppt},\vpt,H^{\up}_{S})$.

\medskip

\noindent{\underline{Compute $\delta_{\tree}$.}} This part is essentially the same as computing $\delta_{\tree}$ for $H^{\spar}_{S}$. Namely, implement the interface below using the same approach, which works because it only relies on that $\check{\tau}$ is upward $A_{i}$-sparse.

\begin{tcolorbox}[
  colback=blue!3,
  colframe=black,
  boxrule=1.5pt,
  arc=0pt,
  left=6pt,
  right=6pt,
  top=4pt,
  bottom=4pt
]
\noindent\textcolor{DarkBlue}{\textit{Interface}} $\mathsf{CountUpwardShortcuts}(\upt,S,\check{\tau},\check{S})$. It receives point $\upt$, affected clusters $S,\check{S}$, and residual tree $\check{\tau}$ of $\check{S}$, such that $\upt\in \core(S)$, $S\anc\check{S}$, $\check{S}$ is nearby $A_{i}$-dense for $i:=\level(S)$, and $\check{\tau}$ is upward $A_{i}$-sparse. It returns $\delta(\upt, \core(\check{\tau}), H^{\up}_{S})$.
\end{tcolorbox}

\noindent{\underline{Compute $\delta_{\aff}$}} It is the same as computing $\delta_{\aff}$ for $H^{\spar}_{\up}$.

\subsection{Interfaces}
\label{sect:Interfaces}

In this section, we explain some of the details behind the interfaces mentioned above. In particular, we elaborate on $\mathsf{CountSparseShortcuts}(\upt,S,\check{\tau},\check{S})$, whose technical details were previously omitted. However, all of these bottom-level interfaces have already appeared, either explicitly or implicitly, in \cite{long2025connectivity}. Thus, we do not give fully formal arguments here.

\subsubsection{The Euler Tours of Steiner Trees}

For clarity, we restate the preliminaries on Euler tour order here.

\begin{definition}[Euler Tours]
Consider a cluster $S$ and its Steiner tree $T_{S}$. We define $\Euler(T_{S})$ to be its Euler tour, which is an ordered list that includes all \emph{occurrences} of vertices according to a DFS traversal of $T_{S}$, starting from an arbitrary root node. For convenience, we add two \emph{virtual occurrences} $\start(T_{S})$ and $\eend(T_{S})$ at the front and the end of $\Euler(T_{S})$ as the ``guard occurences''. For each vertex $v\in V_{\tmn}(S)$, the \emph{principal occurrences} of $v$ on $\Euler(T_{S})$ is its first occurrence. 
\end{definition}

\begin{definition}[Intervals]
For a cluster $S$ and its Steiner tree $T_{S}$, an \emph{interval} $I$ on $\Euler(T_{S})$ is a consecutive subsequence. The terminal vertices of $I$, denoted by $V_{\tmn}(I)$, are those vertices in $V_{\tmn}(S)$ whose principal occurrences fall in $I$. The core of $I$ is the corresponding node set of $V_{\tmn}(I)$, i.e. $\core(I):=\{\point(v,S)\mid v\in V_{\tmn}(I)\}$. The left (resp. right) \emph{outer endpoints} of $I$ are the occurrences immediately to the left (resp. right) of $I$.
\end{definition}

\begin{definition}[Residual Intervals under Failures]
Consider the query phase for some $(x,y,F)$. For the Steiner tree $T$ of an \underline{affected} cluster, we use ${\cal I}_{T}$ to denote the \emph{residual intervals} of $T$, which are intervals obtained by removing $\start(T),\eend(T)$ and all occurrences of failed vertices from $\Euler(T)$.
\end{definition}

\begin{lemma}
\label{lemma:UseIntervals}
In the query phase, there is an assignment that assigns a subset of residual intervals ${\cal I}_{\tau}\subseteq {\cal I}_{T}$ to each residual subtree $\tau\in T\setminus F$, satisfying the following.
\begin{itemize}
\item For each subtree $\tau\in T\setminus F$, $V_{\tmn}(\tau)=\bigsqcup_{I\in{\cal I}_{\tau}}V_{\tmn}(I)$.
\item $\{{\cal I}_{\tau}\mid \tau\in T\setminus F\}$ forms a partition of ${\cal I}_{T}$.
\end{itemize}
\end{lemma}

\begin{definition}[Canonical Intervals]
Consider a Steiner tree $T$. For each occurrences $v_{\oc}\in \Euler(T)$, let ${\cal I}^{\rightarrow}_{T}(v_{\oc})$ denote the \emph{canonical intervals} on the right side of $v_{\oc}$, defined as follows. Let $v'_{\oc}\in \Euler(T)$ be the occurrence next to $v_{\oc}$ on the right. Then ${\cal I}^{\rightarrow}_{T}(v_{\oc})$ includes all intervals starting from $v'_{\oc}$ of size $1,2,4,8,...$ until the largest such interval, namely the suffix beginning with $v'_{\oc}$. Similarly, we can define the canonical intervals $\Ical^{\larr}_{T}(v_{\oc})$.
\end{definition}

\begin{observation}
\label{ob:UseCanonical}
In the query phase, for each residual interval $I\in\Ical_{T}$ with left and right outer endpoints $u_{\oc}$ and $v_{\oc}$, there exists canonical intervals $I^{\rarr}\in \Ical^{\rarr}(u_{\oc})$ and $I^{\larr}\in \Ical^{\larr}(v_{\oc})$ such that $V_{\tmn}(I^{\rarr})\cup V_{\tmn}(I^{\larr}) = V_{\tmn}(I)$.
\end{observation}

\subsubsection{Determine Heaviness}
\label{sect:InterfaceHeaviness}

In what follows, for ease of presentation, we may store information in ``labels of occurrences''. Formally, this means at the very end, we will further store the label of an occurrence $v_{\oc}$ at its corresponding vertex $v$. Exceptionally, for the guard occurrences $\start(T_{S})$ and $\eend(T_{S})$ which have no corresponding vertex, we will store their labels at each vertex in the cluster $S$.

\begin{observation}
\label{ob:AccessIntervalEndpoints}
In the query phase, for each residual interval $I$, we can access the labels of its left and right outer endpoints.
\end{observation}
\begin{proof}
This is because the left and right outer endpoints of $I$ are either occurrences of failed vertices, or the guard occurrences $\start(T_{S}), \eend(T_{S})$. Since a residual interval $I$ is defined under the affected cluster $S$, we can access the labels of $\start(T_{S})$ and $\eend(T_{S})$ at any vertex in $S\cap F$ (which is not empty by the definition of affected clusters).
\end{proof}

\medskip

\noindent\underline{Labels $\CoreNodes(I^{\rarr},v_{\oc},T)$ and $\CoreNodes(I^{\larr},v_{\oc},T)$.} For each Steiner tree $T$, each occurrence $v_{\oc}\in \Euler(T)$, and each canonical interval $I^{\rarr}\in\Ical^{\rarr}_{T}(v_{\oc})$, the label $\CoreNodes(I^{\rarr},v_{\oc},T)$ is a set includes as many nodes in $\core(I^{\rarr})$ as possible, up to $\rho_{\heavy}$ of them. $\CoreNodes(I^{\larr},v_{\oc},T)$ can be defined analogously. We store $\CoreNodes(I^{\rarr},v_{\oc},T)$ and $\CoreNodes(I^{\larr},v_{\oc},T)$ in the label of the occurrence $v_{\oc}$.

\medskip

\noindent\underline{Interface $\mathsf{ListCoreNodes}(\tau,S)$.} Recall that in the input, $S$ is an affected cluster and $\tau$ is a residual subtree of $S$. The goal is to either determine $|V_{\tmn}(\tau)|=|\core(\tau)|\geq \rho_{\heavy}$, or return $\core(\tau)$ explicitly. By \Cref{lemma:UseIntervals} and \Cref{ob:UseCanonical},
\[
\core(\tau) = \bigcup_{I\in\Ical_{\tau}} \core(I^{\rarr})\cup \core(I^{\larr}).
\]
Moreover, for each canonical interval $I^{\rarr}$ (or $I^{\larr}$) in the above expression, we can either determine $|\core(I^{\rarr})|\geq \rho_{\heavy}$ or obtain $\core(I^{\rarr})$ from the label $\CoreNodes(I^{\rarr},v_{\oc},T)$, where $v_{\oc}$ is the left outer endpoint of $I$ and thus it can be accessed by \Cref{ob:AccessIntervalEndpoints}. Therefore, we can merge the answers for all such canonical intervals to obtain the answer for $\tau$.

\subsubsection{Obtain Affected Shortcuts}

For ease of presentation, we may also store information in ``labels of clusters''. Formally, this means at the end we will further store the label of a cluster $S$ at each vertex in $S$. Note that in the query phase, we can access the label of an affected cluster $S$ since $S\cap (F\cup\{s,t\})$ is non-empty.

\medskip

\noindent\underline{Labels $\Shortcuts(S)$.} For each cluster $S$, we use $\Shortcuts(S)$ to denote the set of all shortcuts owned by $S$. We store $\Shortcuts(S)$ at the label of the cluster $S$.

\medskip

\noindent\underline{Interface $\mathsf{ListAffectedShortcuts}()$.} We can simply obtain the explicit set ${\cal S}_{\aff}$ of all affected shortcuts by inspecting $\Shortcuts(S)$ for all affected clusters.

\subsubsection{Labels for Interfaces $\mathsf{ListUpwardNeighbors}$ and $\mathsf{CheckAdjCluster}$}
\label{sect:InterfacesUpwardNbr}

The following two certificates are exactly the information needed for a point $\vpt$ (in a light residual subtree) or a residual subtree. As we mentioned, we will store the following certificates $\mathsf{AdjCluster}$ and $\Upward$ (1) as secondary labels affiliated with elements in the above $\mathsf{CoreNodes}$, and (2) for each canonical interval. 

\begin{definition}[Adjacent-Cluster Denseness Certificate]
Let $S$ be a cluster and $K\subseteq \core(S)$ be a subset of its core nodes. For each level $i'\geq i:=\level(S)$, an adjacent-cluster $A_{i'}$-denseness certificate of $S$, denoted by $\AdjCluster(K,i')$, consists of (the identifiers of) all clusters $\check{S}\preceq S$ s.t.
\begin{itemize}
\item $\Lambda(\check{S})$ intersects $K$, and
\item $\check{S}$ is nearby $A_{i'}$-dense.
\end{itemize}
However, if there are at least $\lambda_{\aff}+1$ such $\check{S}$, then only keep arbitrary $\lambda_{\aff}+1$ of them in $\AdjCluster(S,i')$.
\end{definition}

\begin{lemma}
In the query phase for arbitrary $(x,y,F)$, with access to $\AdjCluster(K,i')$ and the set ${\cal S}_{\aff}$ of all affected clusters, one can determine whether there exists an \underline{unaffected} cluster $\check{S}\preceq S$ (as above, $S$ is the cluster with $\core(S)\supseteq K$) s.t. $\Lambda(\check{S})$ intersects $\core(S)$ and $\check{S}$ is nearby $A_{i'}$-dense.
\end{lemma}
\begin{proof}
This is because the number of affected clusters is bounded by $\lambda_{\aff}$, and we can simply rule out all of them from $\AdjCluster(S,i')$ with access to ${\cal S}_{\aff}$. 
\end{proof}

\begin{definition}[Upward Denseness Certificate]
Consider a shortcut set $H'\subseteq H$. Let $S$ be a cluster and $K\subseteq \core(S)$ be a subset of its core nodes. For each level $i'\geq i:=\level(S)$, an \emph{upward $A_{i'}$-denseness certificate} of $K$, denoted by $\Upward(K,i',H')$, is a subset of $\Nbr(K,i',H')$ of size  $\min\{\lambda_{\affS}\cdot\lambda_{\nb} + \rho_{\heavy},|\Nbr(K,i',H')|\}$. Moreover, each node $v_{\nd}\in \Upward(K,i',H')$ is associated with the number of $H'$-shortcuts in $\{v_{\nd}\}\times K$.

\end{definition}

\begin{lemma}
In the query phase for arbitrary $(x,y,F)$, with access to 
\begin{itemize}
\item $\Upward(K,i',H')$, 
\item the set $H_{\aff}$ of all affected shortcuts,
\item interfaces that determine whether a shortcut is inside $H'$ and whether a node is contained by $K$. 
\end{itemize}
one can determine whether
\begin{itemize}
\item $K$ is upward $A_{i'}$-dense w.r.t. $H'$, i.e. $|\Nbr(K,i',H'\setminus H_{\aff})|\geq \rho_{\heavy}$, or
\item $K$ is upward $A_{i'}$-sparse w.r.t. $H'$, and furthermore, obtain $\Nbr(K,i',H'\setminus H_{\aff})$ explicitly.
\end{itemize}
\end{lemma}
\begin{proof}
First, observe that when $|\Upward(K,i',H')| = \lambda_{\affS}\cdot \lambda_{\nb} + \rho_{\heavy}$, we must have $|\Nbr(K,i',H'\setminus H_{\aff})|\geq \rho_{\heavy}$ since the number of distinct endpoints of affected shortcuts is bounded by $\lambda_{\affS}\cdot \lambda_{\nb}$.

From now, assume $\Upward(K,i',H') = \Nbr(K,i',H')$. Then for each node $v_{\nd}\in\Nbr(K,i',H')$, we can compute the number $\delta$ of $(H'\setminus H_{\aff})$-shortcuts in $\{v_{\nd}\}\times K$, since it is equal to the number of $H'$-shortcuts in $\{v_{\nd}\}\times K$ minus the number of $(H'\cap H_{\aff})$-shortcuts in $\{v_{\nd}\}\times K$. Since $v_{\nd}\in \Nbr(K,i',H'\setminus H_{\aff})$ iff $\delta >0$, we obtain $\Nbr(K,i',H'\setminus H_{\aff})$ explicitly.
\end{proof}

\subsubsection{Labels for $\mathsf{CountSparseShortcuts}(u_{\ppt},S,\check{\tau},\check{S})$}
\label{sect:PrefixSumLabel}

Recall that $\upt\in \core(S)$ is at level $i$, and $\check{\tau}$ of $\check{S}$ is at a lower level. Our goal is to return $\delta(\upt, \core(\check{\tau}), H^{\spar}_{S})$. A key property is that $\check{\tau}$ is upward $A_{i}$-sparse, and the answer is non-zero only when $\upt\in \Nbr(\core(\check{\tau}),i,H^{\spar}_{S}\setminus H_{\aff})$.

Without loss of generality, assume $\check{\tau}$ has only one residual interval $I$. We have mentioned that we want to store a superset of $\Nbr(\core(I),i,H^{\spar}_{S}\setminus H_{\aff})$ at both the outer endpoints of $I$, denoted by $w_{\oc}$ and $z_{\oc}$, so that we can apply the prefix-sum trick. This superset can be simply a upward $A_{i}$-denseness certificate $\Upward(\core(I),i,H^{\spar}_{S})$ of $I$ as discussed above. The difficulty, however, is that $I$ is not known during the construction phase. Fortunately, the certificate $\Upward(\core(I),i,H^{\spar}_{S})$ (considered purely as a node set, without the associated numbers) satisfies a certain monotonicity property: fixing a left outer endpoint, when we grows the interval, we can maintain an incremental certificate (this is straightforward from the definition). Therefore, at the left endpoint $u_{\oc}$ of $I$, we can store a maximal certificate generated by the above growing procedure, which will automatically be a superset of the certificate of $I$. Now, $u_{\ppt}$ appears in the labels of $w_{\oc}$ and $z_{\oc}$, allowing us to store two prefix-sum values whose difference gives $\delta(\upt,\core(I),H^{\spar}_{S})$.

\subsection{Label Size Analysis}
\label{sect:DistLabelSize}

For clarity, let us count the $f$ factors in the label size, the overhead will always bounded by $n^{O(1/k)}\cdot \log^{O(k)}n$.

It is not hard to see that the bottleneck appears in the \emph{secondary labels} affiliated to elements in each $\CoreNodes$ labels. From the definition of $\mathsf{CoreNodes}$, each such label has $\rho_{\heavy}$ many elements, incurring an $f$ factor. Let us list the interfaces which store secondary labels affiliated to $\CoreNodes$ as follows. 
\begin{itemize}
\item The first one is $\mathsf{ListUpwardNeighbors}(\vpt,\tau,i')$, the single-point version. The affiliated information is a upward denseness certificate, which takes $\lambda_{\affS}\cdot \lambda_{\nb} + \rho_{\heavy}$ size, incurring an $f^{2}$ factor.
\item 
The second one is $\mathsf{CheckAdjCluster(\vpt,i')}$, the single-point version. The affiliated information is (the identifiers of) $\lambda_{\affS}+1$ candidate clusters, which incurs one $f$ factor.
\item 
The third one is $\mathsf{FindDownwardShortcuts}(\upt,S)$. The affiliated information is a ``downward version'' of upward denseness certificate. It takes $\lambda_{\affS}\cdot \lambda_{\nb} + f + 1$ size, which incurs an $f^{2}$ factor. %
\end{itemize}
We can verify that all other interfaces only store labels with $f^{2}\cdot n^{O(1/k)}\cdot \log^{O(k)}n$ size.
Therefore, the final label size is
\[
f^{3}\cdot n^{O(1/k)}\cdot \log^{O(k)}n.
\]

\subsection{Improved Vertex-Failure Connectivity Labeling Schemes}
\label{sect:ConnLabelSize}

Finally, let us discuss how to improve the $f^{3}$ factor to an $f^{2}$ factor in the connectivity setting. The key observation is that, in this setting, there is no need to store information for individual points. In other words, we will create subtree nodes for every residual subtree, regardless of whether it is light or heavy, but no singleton node. The subtree versions of the above three interfaces each incur only an $f^{2}$ factor, and they no longer need to be affiliated with $\CoreNodes$.

We also want to point out that, even in the distance setting, single-point versions of interfaces $\mathsf{CheckAdjCluster}$ and $\mathsf{FindDownwardShortcuts}$
are already unnecessary, because when we are checking denseness, we can always view a (no matter light or heavy) residual subtree as a whole. Hence, the real bottleneck of our distance labeling scheme is the interface $\mathsf{ListUpwardNeighbors}$, which still needs to output single-point information when working with a light residual subtree.

Finally, the overhead in the connectivity setting will only be $\Otilde(1)$, following \cite{long2025connectivity}. Also, when working with a (non-length-constrained) vertex expander hierarchy, we no longer need the sampling step to obtain integral versions of fractional moving cuts. This removes the only randomized step in our algorithm. Therefore, we obtain a deterministic vertex-failure connectivity labeling scheme with label size 
\[
\Otilde(f^{2}).
\]

\bibliographystyle{alpha}
\bibliography{ref}

@inproceedings{haeupler2025length,
  author = {Haeupler, Bernhard and Long, Yaowei and Saranurak, Thatchaphol and Wang, Shengzhe},
  title = {Length-Constrained Directed Expander Decomposition and Length-Constrained Vertex-Capacitated Flow Shortcuts},
  booktitle = {33rd Annual European Symposium on Algorithms (ESA 2025)},
  series = {Leibniz International Proceedings in Informatics (LIPIcs)},
  volume = {351},
  pages = {107:1--107:17},
  publisher = {Schloss Dagstuhl -- Leibniz-Zentrum f{\"u}r Informatik},
  year = {2025},
  doi = {10.4230/LIPIcs.ESA.2025.107}
}

@inproceedings{HLRS26,
  author = {Haeupler, Bernhard and Long, Yaowei and Roeyskoe, Antti and Saranurak, Thatchaphol},
  title = {A Constant-Approximation Distance Labeling Scheme under Polynomially Many Edge Failures},
  booktitle = {Proceedings of the 58th Annual ACM Symposium on Theory of Computing (STOC 2026)},
  pages = {1014--1024},
  publisher = {ACM},
  year = {2026},
  doi = {10.1145/3798129.3800814}
}

@inproceedings{haeupler2023maximum,
  author = {Haeupler, Bernhard and Hershkowitz, D. Ellis and Saranurak, Thatchaphol},
  title = {Maximum Length-Constrained Flows and Disjoint Paths: Distributed, Deterministic, and Fast},
  booktitle = {Proceedings of the 55th Annual ACM Symposium on Theory of Computing (STOC 2023)},
  pages = {1371--1383},
  publisher = {ACM},
  year = {2023},
  doi = {10.1145/3564246.3585202}
}

@article{awerbuch1998near,
  author = {Awerbuch, Baruch and Berger, Bonnie and Cowen, Lenore and Peleg, David},
  title = {Near-linear time construction of sparse neighborhood covers},
  journal = {SIAM Journal on Computing},
  volume = {28},
  number = {1},
  pages = {263--277},
  year = {1998},
  doi = {10.1137/S0097539794271898}
}

@inproceedings{duan2010connectivity,
  author = {Duan, Ran and Pettie, Seth},
  title = {Connectivity oracles for failure prone graphs},
  booktitle = {Proceedings of the 42nd Annual ACM Symposium on Theory of Computing (STOC 2010)},
  pages = {465--474},
  publisher = {ACM},
  year = {2010}
}

@article{duan2020connectivity,
  author = {Duan, Ran and Pettie, Seth},
  title = {Connectivity oracles for graphs subject to vertex failures},
  journal = {SIAM Journal on Computing},
  volume = {49},
  number = {6},
  pages = {1363--1396},
  year = {2020},
  doi = {10.1137/17M1146610}
}

@inproceedings{long2022near,
  author = {Long, Yaowei and Saranurak, Thatchaphol},
  title = {Near-optimal deterministic vertex-failure connectivity oracles},
  booktitle = {63rd IEEE Annual Symposium on Foundations of Computer Science (FOCS 2022)},
  pages = {1002--1010},
  publisher = {IEEE},
  year = {2022},
  doi = {10.1109/FOCS54457.2022.00098}
}

@inproceedings{parter2022ptimal,
  author = {Parter, Merav},
  title = {Nearly optimal vertex fault-tolerant spanners in optimal time: sequential, distributed, and parallel},
  booktitle = {Proceedings of the 54th Annual ACM SIGACT Symposium on Theory of Computing (STOC 2022)},
  pages = {1080--1092},
  publisher = {ACM},
  year = {2022},
  doi = {10.1145/3519935.3520047}
}

@article{cohen2001all,
  author = {Cohen, Edith and Zwick, Uri},
  title = {All-pairs small-stretch paths},
  journal = {Journal of Algorithms},
  volume = {38},
  number = {2},
  pages = {335--353},
  year = {2001}
}

@inproceedings{long2025connectivity,
  author = {Long, Yaowei and Pettie, Seth and Saranurak, Thatchaphol},
  title = {Connectivity labeling schemes for edge and vertex faults via expander hierarchies},
  booktitle = {Proceedings of the 2025 Annual ACM-SIAM Symposium on Discrete Algorithms (SODA 2025)},
  pages = {1--47},
  publisher = {SIAM},
  year = {2025}
}

@inproceedings{duan2021approximate,
  author = {Duan, Ran and Gu, Yong and Ren, Hanlin},
  title = {Approximate distance oracles subject to multiple vertex failures},
  booktitle = {Proceedings of the 2021 Annual ACM-SIAM Symposium on Discrete Algorithms (SODA 2021)},
  pages = {2497--2516},
  publisher = {SIAM},
  year = {2021}
}

@inproceedings{parter2024connectivity,
  author = {Parter, Merav and Petruschka, Asaf and Pettie, Seth},
  title = {Connectivity labeling and routing with multiple vertex failures},
  booktitle = {Proceedings of the 56th Annual ACM Symposium on Theory of Computing (STOC 2024)},
  pages = {823--834},
  publisher = {ACM},
  year = {2024}
}

@inproceedings{patrascu2007planning,
  author = {P{\u{a}}tra{\c{s}}cu, Mihai and Thorup, Mikkel},
  title = {Planning for fast connectivity updates},
  booktitle = {48th Annual IEEE Symposium on Foundations of Computer Science (FOCS 2007)},
  pages = {263--271},
  publisher = {IEEE},
  year = {2007}
}

@inproceedings{kapron2013dynamic,
  author = {Kapron, Bruce M. and King, Valerie and Mountjoy, Ben},
  title = {Dynamic graph connectivity in polylogarithmic worst case time},
  booktitle = {Proceedings of the 2013 Annual ACM-SIAM Symposium on Discrete Algorithms (SODA 2013)},
  pages = {1131--1142},
  publisher = {SIAM},
  year = {2013}
}

@inproceedings{kosinas2023connectivity,
  author = {Kosinas, Evangelos},
  title = {Connectivity queries under vertex failures: not optimal, but practical},
  booktitle = {31st Annual European Symposium on Algorithms (ESA 2023)},
  series = {Leibniz International Proceedings in Informatics (LIPIcs)},
  volume = {274},
  pages = {75:1--75:13},
  publisher = {Schloss Dagstuhl -- Leibniz-Zentrum f{\"u}r Informatik},
  year = {2023},
  doi = {10.4230/LIPIcs.ESA.2023.75}
}

@inproceedings{jiang2026new,
  author = {Jiang, Yonggang and Parter, Merav and Petruschka, Asaf},
  title = {New oracles and labeling schemes for vertex cut queries},
  booktitle = {Proceedings of the 2026 Annual ACM-SIAM Symposium on Discrete Algorithms (SODA 2026)},
  pages = {2760--2791},
  publisher = {SIAM},
  year = {2026}
}

@inproceedings{kosinas2025optimal,
  author = {Kosinas, Evangelos},
  title = {An optimal 3-fault-tolerant connectivity oracle},
  booktitle = {52nd International Colloquium on Automata, Languages, and Programming (ICALP 2025)},
  series = {Leibniz International Proceedings in Informatics (LIPIcs)},
  volume = {334},
  pages = {110:1--110:20},
  publisher = {Schloss Dagstuhl -- Leibniz-Zentrum f{\"u}r Informatik},
  year = {2025},
  doi = {10.4230/LIPIcs.ICALP.2025.110}
}

@inproceedings{chechik20171+,
  author = {Chechik, Shiri and Cohen, Sarel and Fiat, Amos and Kaplan, Haim},
  title = {$(1+\epsilon)$-Approximate {$f$}-Sensitive Distance Oracles},
  booktitle = {Proceedings of the 2017 Annual ACM-SIAM Symposium on Discrete Algorithms (SODA 2017)},
  pages = {1479--1496},
  publisher = {SIAM},
  year = {2017}
}

@article{chechik2012f,
  author = {Chechik, Shiri and Langberg, Michael and Peleg, David and Roditty, Liam},
  title = {{$f$}-Sensitivity distance oracles and routing schemes},
  journal = {Algorithmica},
  volume = {63},
  number = {4},
  pages = {861--882},
  year = {2012}
}

@inproceedings{bilo2023fault,
  author = {Bil{\`o}, Davide and Choudhary, Keerti and Cohen, Sarel and Friedrich, Tobias and Krogmann, Simon and Schirneck, Martin},
  title = {Fault-tolerant {ST}-diameter oracles},
  booktitle = {50th International Colloquium on Automata, Languages, and Programming (ICALP 2023)},
  series = {Leibniz International Proceedings in Informatics (LIPIcs)},
  volume = {261},
  pages = {24:1--24:20},
  publisher = {Schloss Dagstuhl -- Leibniz-Zentrum f{\"u}r Informatik},
  year = {2023},
  doi = {10.4230/LIPIcs.ICALP.2023.24}
}

@inproceedings{bilo2023approximate,
  author = {Bil{\`o}, Davide and Chechik, Shiri and Choudhary, Keerti and Cohen, Sarel and Friedrich, Tobias and Krogmann, Simon and Schirneck, Martin},
  title = {Approximate distance sensitivity oracles in subquadratic space},
  booktitle = {Proceedings of the 55th Annual ACM Symposium on Theory of Computing (STOC 2023)},
  pages = {1396--1409},
  publisher = {ACM},
  year = {2023}
}

@inproceedings{duan2022maintaining,
  author = {Duan, Ran and Ren, Hanlin},
  title = {Maintaining exact distances under multiple edge failures},
  booktitle = {Proceedings of the 54th Annual ACM SIGACT Symposium on Theory of Computing (STOC 2022)},
  pages = {1093--1101},
  publisher = {ACM},
  year = {2022}
}

@inproceedings{dey2024nearly,
  author = {Dey, Dipan and Gupta, Manoj},
  title = {Nearly optimal fault-tolerant distance oracle},
  booktitle = {Proceedings of the 56th Annual ACM Symposium on Theory of Computing (STOC 2024)},
  pages = {944--955},
  publisher = {ACM},
  year = {2024}
}

@inproceedings{bilo2024improved,
  author = {Bil{\`o}, Davide and Chechik, Shiri and Choudhary, Keerti and Cohen, Sarel and Friedrich, Tobias and Schirneck, Martin},
  title = {Improved distance (sensitivity) oracles with subquadratic space},
  booktitle = {65th IEEE Annual Symposium on Foundations of Computer Science (FOCS 2024)},
  pages = {1550--1558},
  publisher = {IEEE},
  year = {2024}
}

@misc{gibb2015dynamic,
  author = {Gibb, David and Kapron, Bruce M. and King, Valerie and Thorn, Nolan},
  title = {Dynamic graph connectivity with improved worst case update time and sublinear space},
  howpublished = {arXiv preprint arXiv:1509.06464},
  year = {2015},
  url = {https://arxiv.org/abs/1509.06464}
}

@inproceedings{long2024better,
  author = {Long, Yaowei and Wang, Yunfan},
  title = {Better decremental and fully dynamic sensitivity oracles for subgraph connectivity},
  booktitle = {51st International Colloquium on Automata, Languages, and Programming (ICALP 2024)},
  series = {Leibniz International Proceedings in Informatics (LIPIcs)},
  volume = {297},
  pages = {109:1--109:20},
  publisher = {Schloss Dagstuhl -- Leibniz-Zentrum f{\"u}r Informatik},
  year = {2024},
  doi = {10.4230/LIPIcs.ICALP.2024.109}
}

@article{thorup2005approximate,
  author = {Thorup, Mikkel and Zwick, Uri},
  title = {Approximate distance oracles},
  journal = {Journal of the ACM},
  volume = {52},
  number = {1},
  pages = {1--24},
  year = {2005}
}

@article{baswana2013approximate,
  author = {Baswana, Surender and Khanna, Neelesh},
  title = {Approximate shortest paths avoiding a failed vertex: near optimal data structures for undirected unweighted graphs},
  journal = {Algorithmica},
  volume = {66},
  number = {1},
  pages = {18--50},
  year = {2013}
}

@inproceedings{bernstein2008improved,
  author = {Bernstein, Aaron and Karger, David R.},
  title = {Improved distance sensitivity oracles via random sampling},
  booktitle = {Proceedings of the 2008 Annual ACM-SIAM Symposium on Discrete Algorithms (SODA 2008)},
  pages = {34--43},
  publisher = {SIAM},
  year = {2008}
}

@inproceedings{bernstein2009nearly,
  author = {Bernstein, Aaron and Karger, David R.},
  title = {A nearly optimal oracle for avoiding failed vertices and edges},
  booktitle = {Proceedings of the 41st Annual ACM Symposium on Theory of Computing (STOC 2009)},
  pages = {101--110},
  publisher = {ACM},
  year = {2009}
}

@inproceedings{chechik2020distance,
  author = {Chechik, Shiri and Cohen, Sarel},
  title = {Distance sensitivity oracles with subcubic preprocessing time and fast query time},
  booktitle = {Proceedings of the 52nd Annual ACM Symposium on Theory of Computing (STOC 2020)},
  pages = {1375--1388},
  publisher = {ACM},
  year = {2020}
}

@inproceedings{duan2017improved,
  author = {Duan, Ran and Zhang, Tianyi},
  title = {Improved distance sensitivity oracles via tree partitioning},
  booktitle = {15th International Symposium on Algorithms and Data Structures (WADS 2017)},
  pages = {349--360},
  publisher = {Springer},
  year = {2017},
  doi = {10.1007/978-3-319-62127-2_30}
}

@article{ren2022improved,
  author = {Ren, Hanlin},
  title = {Improved distance sensitivity oracles with subcubic preprocessing time},
  journal = {Journal of Computer and System Sciences},
  volume = {123},
  pages = {159--170},
  year = {2022}
}

@inproceedings{dinitz2011fault,
  author = {Dinitz, Michael and Krauthgamer, Robert},
  title = {Fault-tolerant spanners: better and simpler},
  booktitle = {Proceedings of the 30th Annual ACM SIGACT-SIGOPS Symposium on Principles of Distributed Computing (PODC 2011)},
  pages = {169--178},
  publisher = {ACM},
  year = {2011}
}

@inproceedings{bodwin2018optimal,
  author = {Bodwin, Greg and Dinitz, Michael and Parter, Merav and Williams, Virginia Vassilevska},
  title = {Optimal vertex fault tolerant spanners (for fixed stretch)},
  booktitle = {Proceedings of the 2018 Annual ACM-SIAM Symposium on Discrete Algorithms (SODA 2018)},
  pages = {1884--1900},
  publisher = {SIAM},
  year = {2018},
  doi = {10.1137/1.9781611975031.123}
}

@inproceedings{bodwin2019trivial,
  author = {Bodwin, Greg and Patel, Shyamal},
  title = {A trivial yet optimal solution to vertex fault tolerant spanners},
  booktitle = {Proceedings of the 2019 ACM Symposium on Principles of Distributed Computing (PODC 2019)},
  pages = {541--543},
  publisher = {ACM},
  year = {2019},
  doi = {10.1145/3293611.3331588}
}

@inproceedings{chechik2009fault,
  author = {Chechik, Shiri and Langberg, Michael and Peleg, David and Roditty, Liam},
  title = {Fault-tolerant spanners for general graphs},
  booktitle = {Proceedings of the 41st Annual ACM Symposium on Theory of Computing (STOC 2009)},
  pages = {435--444},
  publisher = {ACM},
  year = {2009}
}

@inproceedings{henzinger2015unifying,
  author = {Henzinger, Monika and Krinninger, Sebastian and Nanongkai, Danupon and Saranurak, Thatchaphol},
  title = {Unifying and strengthening hardness for dynamic problems via the online matrix-vector multiplication conjecture},
  booktitle = {Proceedings of the 47th Annual ACM Symposium on Theory of Computing (STOC 2015)},
  pages = {21--30},
  publisher = {ACM},
  year = {2015},
  doi = {10.1145/2746539.2746609}
}

@inproceedings{courcelle2007compact,
  author = {Courcelle, Bruno and Twigg, Andrew},
  title = {Compact forbidden-set routing},
  booktitle = {24th Annual Symposium on Theoretical Aspects of Computer Science (STACS 2007)},
  pages = {37--48},
  publisher = {Springer},
  year = {2007},
  doi = {10.1007/978-3-540-70918-3_4}
}

@inproceedings{DBLP:conf/podc/DoryP21,
  author = {Dory, Michal and Parter, Merav},
  title = {Fault-tolerant labeling and compact routing schemes},
  booktitle = {Proceedings of the 2021 ACM Symposium on Principles of Distributed Computing (PODC 2021)},
  pages = {445--455},
  publisher = {ACM},
  year = {2021},
  doi = {10.1145/3465084.3467929}
}

@inproceedings{DBLP:conf/podc/IzumiEWM23,
  author = {Izumi, Taisuke and Emek, Yuval and Wadayama, Tadashi and Masuzawa, Toshimitsu},
  title = {Deterministic fault-tolerant connectivity labeling scheme},
  booktitle = {Proceedings of the 2023 ACM Symposium on Principles of Distributed Computing (PODC 2023)},
  pages = {190--199},
  publisher = {ACM},
  year = {2023},
  doi = {10.1145/3583668.3594584}
}

@inproceedings{DBLP:conf/wdag/ParterP22a,
  author = {Parter, Merav and Petruschka, Asaf},
  title = {{\~{O}ptimal} dual vertex failure connectivity labels},
  booktitle = {36th International Symposium on Distributed Computing (DISC 2022)},
  series = {Leibniz International Proceedings in Informatics (LIPIcs)},
  volume = {246},
  pages = {32:1--32:19},
  publisher = {Schloss Dagstuhl -- Leibniz-Zentrum f{\"u}r Informatik},
  year = {2022},
  doi = {10.4230/LIPIcs.DISC.2022.32}
}

@article{DBLP:journals/algorithmica/NagamochiI92,
  author = {Nagamochi, Hiroshi and Ibaraki, Toshihide},
  title = {A linear-time algorithm for finding a sparse {$k$}-connected spanning subgraph of a {$k$}-connected graph},
  journal = {Algorithmica},
  volume = {7},
  number = {5--6},
  pages = {583--596},
  year = {1992},
  doi = {10.1007/BF01758778}
}

@inproceedings{DBLP:conf/stoc/HaeuplerR022,
  author = {Haeupler, Bernhard and R{\"a}cke, Harald and Ghaffari, Mohsen},
  title = {Hop-constrained expander decompositions, oblivious routing, and distributed universal optimality},
  booktitle = {Proceedings of the 54th Annual ACM SIGACT Symposium on Theory of Computing (STOC 2022)},
  pages = {1325--1338},
  publisher = {ACM},
  year = {2022},
  doi = {10.1145/3519935.3520026}
}

@article{DBLP:journals/dc/AfekBKCM02,
  author = {Afek, Yehuda and Bremler-Barr, Anat and Kaplan, Haim and Cohen, Edith and Merritt, Michael},
  title = {Restoration by path concatenation: fast recovery of {MPLS} paths},
  journal = {Distributed Computing},
  volume = {15},
  number = {4},
  pages = {273--283},
  year = {2002},
  doi = {10.1007/s00446-002-0080-6}
}

\appendix

\section{Proof Sketch of \Cref{thm:PolyED}}
\label{sect:PolyEDProof}

The proof of polynomial-time construction of LC edge-expander decomposition in \cite{HLRS26} (which is adapted from \cite{DBLP:conf/stoc/HaeuplerR022}) has three building blocks. 

First, the proof uses a result from \cite{DBLP:conf/stoc/HaeuplerR022} stating that the union of a sequence of sparse (edge) moving cuts always has a small size. Here we cite an analogous result in \cite{haeupler2025length} for (vertex) moving cuts.

\begin{lemma}[Lemma 3.17 and Lemma 4.8 in \cite{haeupler2025length}]
Let $G_{0}$ be a graph, $A$ be a node weighting and $(C_{1},...,C_{k})$ be a sequence of $(hs)$-length (vertex) moving cuts where each cut $C_{i}$ is $(h,s)$-length $\phi$-sparse in $G_{i}=G_{0}-\sum_{j<i}(hs\cdot C_{j})$ w.r.t. $A$. Then 
\[
\sum_{i}|C_{i}| \leq \phi\cdot n^{O(1/s)}\log n\cdot |A|.
\]
\end{lemma}

Second, their proof defines a ``worst-case'' demand $D_{A}$, called the \emph{LDD demand}, for the node weighting $A$, in the sense that if $D_{A}$ can be routed in $G$ with low congestion and length, then $A$ is LC-expanding in $G$. We point out that the LDD demand also works for vertex congestion and LC-vertex-expansion following the same proof.

Third, their proof uses a black-box algorithm from \cite{haeupler2023maximum} for computing approximate multi-commodity cutmatches, which either outputs a flow routing $D_{A}$ in $G$ with low congestion and length, or a sparse moving cut w.r.t. $A$. We note that this building block also extends to vertex moving cuts and vertex congestion, since the original result in \cite{haeupler2023maximum} applies even to directed graphs and then we can apply the standard vertex-splitting reduction. We state it below in a form adapted to the undirected vertex-congestion setting (since we will use it again for expander routing).

\begin{theorem}[Theorem A.2 in \cite{haeupler2023maximum}]
\label{thm:CutMatch}
Let $G$ be a graph with length function $\ell$. Given a length parameter $h\geq 1$, a sparsity parameter $0<\phi<1$, and a collection of $k$ equal-size node weighting pairs $\{(A_{i},A'_{i})\}_{i\in[k]}$, there is an algorithm that computes
\begin{itemize}
\item For each $i\in[k]$, a partition of the supports of $A_{i}$ and $A'_{i}$ into the ``matched'' and ``unmatched'' parts $M_{i}\sqcup U_{i} = \supp(A_{i})$ and $M'_{i}\sqcup U'_{i} = \supp(A'_{i})$.
\item An $h$-length flow $F = \sum_{i\in[k]}F_{i}$ satisfying the following: for each $i$, $F_{i}$ is a flow from $M_{i}$ to $M'_{i}$ such that the total flow value of flow paths in $F_{i}$ from any vertex $v\in M_{i}$ equals $A_{i}(v)$, and the total flow value of flow paths in $F_{i}$ to any vertex $v'\in M'_{i}$ equals $A'_{i}(v)$. Moreover, $F$ has vertex congestion
\[
\gamma = \tilde{O}(1/\phi).
\]
\item An $h$-length vertex moving cut $C$ in $G$, such that $\dist_{G-C}(U_{i},U'_{i})\geq h$ for all $i\in[k]$, and the size of $C$ is at most
\[
|C|\leq \phi\cdot\left(\sum_{i\in[k]}|A_{i}|-|F_{i}|\right).
\]
\end{itemize}
The algorithm runs in $\tilde{O}(m\cdot k\cdot h^{17})$ time and it is deterministic.
\end{theorem}

Lastly, providing these three building blocks, the decomposition algorithm is simply as follows. First it computes the LDD demand $D_{A}$ for $A$ on the current graph. Second, it tries to route $D_{A}$ using \Cref{thm:CutMatch}. If the returned cut is empty, meaning that $D_{A}$ can be routed with low congestion and length successfully, we are done. Otherwise, add the cut to the decomposition, apply the cut to the current graph, and repeat. See \cite{HLRS26} for a detailed statement.

\section{Proof of \Cref{lemma:ExpanderRouting}}
\label{sect:ProofRouting}

We want to apply \Cref{thm:CutMatch}. First, we turn the demand $D$ into a collection of equal-size node weighting pairs: for each $(u,v)\in\supp(D)$, create node weighting pair $\{A_{i},A'_{i}\}$ where $A_{i}$ only assigns non-zero weight $A_{i}(u)=D(u,v)$ at vertex $u$ and $A'_{i}$ only assigns non-zero weight $A_{i}(v) = D(u,v)$ at vertex $v$.

Now, we apply \Cref{thm:CutMatch} on the node weighting pairs, with length parameter $hs$, sparsity parameter remaining $\phi$. Note that the all ``unmatched'' parts $\{U_{i},U'_{i}\}_{i\in[k]}$ must be empty, since otherwise we obtain an $(hs)$-length vertex moving cut which is $\phi$-sparse (witnessed by the unmatched parts). Therefore, the $(hs)$-length flow between the matched parts with vertex congestion $\gamma = \Otilde(1/\phi)$ is exactly what we want.

\end{document}